\documentclass[11pt]{article}

\usepackage{enumitem}
 
\usepackage[T1]{fontenc}
\usepackage[utf8]{inputenc}

\usepackage{authblk}

\usepackage{float}
\usepackage{pdfpages}
\usepackage{fullpage}
\usepackage{changepage}
\usepackage{a4wide}
\usepackage{geometry}
\usepackage{nameref}
\usepackage{hyperref}
\hypersetup{
    colorlinks,
    linkcolor={red!50!black},
    citecolor={blue!50!black},
    urlcolor={blue!80!black}
}
\usepackage{booktabs}
\usepackage{amsthm}
\usepackage{amsmath}
\usepackage{multirow}
\usepackage{amssymb}
\usepackage{mathtools}
\usepackage{mathabx}
\usepackage{amsfonts} 
\usepackage{dsfont}
\usepackage{subcaption}

\usepackage{csquotes}
\usepackage{comment}
\usepackage[normalem]{ulem}
\usepackage{cancel}

\usepackage{placeins}
\usepackage{float}
\usepackage{todonotes}
\usepackage{xcolor}

\usepackage{complexity}

\newcommand{\cA}{\mathcal{A}}
\newcommand{\cB}{\mathcal{B}}
\newcommand{\cC}{\mathcal{C}}
\newcommand{\cD}{\mathcal{D}}
\newcommand{\cE}{\mathcal{E}}
\newcommand{\cF}{\mathcal{F}}
\newcommand{\cG}{\mathcal{G}}

\newcommand{\cN}{\mathcal{N}}

\newcommand{\calP}{\mathcal{P}}

\newcommand{\calS}{\mathcal{S}}
\newcommand{\cT}{\mathcal{T}}

\definecolor{myviolet}{rgb}{0.4,0.0,0.8}

\definecolor{mygreen}{rgb}{0.0,0.6,0.0}

\renewcommand{\triangle}{\tikz[scale=0.18]{\draw (-30:1)--(90:1)--(210:1)--(-30:1)}}

\newcommand{\OneCirc}{
\tikz[scale=0.18, baseline=-0.5ex]{
\coordinate (A) at (0,0);
\node[circle, draw, fill=lightgray, inner sep=1.5pt] at (A) {};
}
}

\newcommand{\TwoCirc}{
\tikz[scale=0.18, baseline=-0.5ex]{
\coordinate (B) at (-1,0);
\coordinate (C) at (1,0);
\node[circle, draw, fill=lightgray, inner sep=1.5pt] at (B) {};
\node[circle, draw, fill=lightgray, inner sep=1.5pt] at (C) {};
}
}

\newcommand{\TwoCircLinked}{
\tikz[scale=0.18, baseline=-0.5ex]{
\coordinate (B) at (-1,0);
\coordinate (C) at (1,0);
\draw (B) -- (C);
\node[circle, draw, fill=lightgray, inner sep=1.5pt] at (B) {};
\node[circle, draw, fill=lightgray, inner sep=1.5pt] at (C) {};
}
}

\newcommand{\ThreeCirc}{
\tikz[scale=0.18, baseline=-0.5ex]{
\coordinate (A) at (90:1);
\coordinate (B) at (210:1);
\coordinate (C) at (-30:1);
\draw (B) -- (C);

\node[circle, draw, fill=lightgray, inner sep=1.5pt] at (A) {};
\node[circle, draw, fill=lightgray, inner sep=1.5pt] at (B) {};
\node[circle, draw, fill=lightgray, inner sep=1.5pt] at (C) {};
}
}

\newcommand{\ThreeCircIndependant}{
\tikz[scale=0.18, baseline=-0.5ex]{
\coordinate (A) at (90:1);
\coordinate (B) at (210:1);
\coordinate (C) at (-30:1);
\node[circle, draw, fill=lightgray, inner sep=1.5pt] at (A) {};
\node[circle, draw, fill=lightgray, inner sep=1.5pt] at (B) {};
\node[circle, draw, fill=lightgray, inner sep=1.5pt] at (C) {};
}
}

\newcommand{\ThreeCircTriangle}{
\tikz[scale=0.18, baseline=-0.5ex]{
\coordinate (A) at (90:1);
\coordinate (B) at (210:1);
\coordinate (C) at (-30:1);
\draw (B) -- (C);
\draw (A) -- (B);
\draw (A) -- (C);
\node[circle, draw, fill=lightgray, inner sep=1.5pt] at (A) {};
\node[circle, draw, fill=lightgray, inner sep=1.5pt] at (B) {};
\node[circle, draw, fill=lightgray, inner sep=1.5pt] at (C) {};
}
}

\newcommand{\ThreeCircPath}{
\tikz[scale=0.18, baseline=-0.5ex]{
\coordinate (A) at (90:1);
\coordinate (B) at (210:1);
\coordinate (C) at (-30:1);
\draw (B) -- (C);
\draw (A) -- (B);
\node[circle, draw, fill=lightgray, inner sep=1.5pt] at (A) {};
\node[circle, draw, fill=lightgray, inner sep=1.5pt] at (B) {};
\node[circle, draw, fill=lightgray, inner sep=1.5pt] at (C) {};
}
}

\newlength{\contentheight}

\newcommand{\dir}[1]{\overrightarrow{#1}}

\theoremstyle{plain}
\newtheorem{problem}{Problem}

\newtheorem*{theorem*}{Theorem}
\newtheorem*{corollary*}{Corollary}

\newtheorem{theorem}{Theorem}[section]
\newtheorem{lemma}[theorem]{Lemma}

\newtheorem{claim}[theorem]{Claim}

\newtheorem{corollary}[theorem]{Corollary}

\theoremstyle{remark}
\newtheorem*{remark}{Remark}

\renewenvironment{proof}{\medskip \par
\noindent \textbf{Proof.} \newline
}{\hfill$\Box$\medskip}

\usepackage[
	backend=biber,        
	backref=false,         
	sorting=none,          
	citestyle= numeric,   
	bibstyle= numeric,    
    doi=false,            
    isbn=false,           
    url=false,            
    giveninits=true,      
    eprint=false,         
    hyperref=true,        
    alldates=year,        
    datezeros=true,       
    natbib=true,          
]{biblatex}

\renewbibmacro*{issue+date}{%
	\iffieldundef{year}
	{} 
	{\printtext[parens]{\printfield{year}}} 
}

\AtEveryBibitem{\clearfield{pagetotal}}
\AtEveryBibitem{\clearfield{visited}}
\renewbibmacro{in:}{%
    \ifentrytype{article}
    {}
    {\bibstring{in}%
        \printunit{\intitlepunct}}}

\DeclareCiteCommand{\fullcite} 
    {\boolfalse{citetracker}%
	\boolfalse{pagetracker}%
	\usebibmacro{prenote}}
    {\ifciteindex
        {\indexnames{labelname}} 								
        {}														
	\printnames{labelname}										
	\iffieldundef{year}											
		{}														
		{\setunit{\printdelim{nonameyeardelim}} 				
		\printfield[parens]{year}}								
	\setunit{\printdelim{nonameyeardelim}}
	\printtext[bibhyperref]{									
		\printtext[brackets]{\printfield{labelnumber}}}}		
	{\multicitedelim}
	{\usebibmacro{postnote}}

\DeclareCiteCommand{\cite} 
    {\boolfalse{citetracker}%
	\boolfalse{pagetracker}%
	}
	{\printtext[bibhyperref]{
		\printtext[brackets]{\printfield{labelnumber}}}}  
    {\multicitedelim}  
    {}

\DeclareCiteCommand{\citeold} 
	{\usebibmacro{prenote}} 
	{\printnames{author}
	 \addspace
	 \printdate
	 \addspace
	 \printtext[bibhyperref]{
		\printtext[brackets]{\printfield{labelnumber}}}} 
	{\multicitedelim} 
	{\usebibmacro{postnote}} 

\DeclareCiteCommand{\citetitle}
	{\boolfalse{citetracker}%
	\boolfalse{pagetracker}%
	\usebibmacro{prenote}}
	{\ifciteindex
		{\indexfield{indextitle}}
		{}%
		\printtext[bibhyperref]{\printfield[citetitle]{labeltitle}}}
	{\multicitedelim}
	{\usebibmacro{postnote}}

\DeclareCiteCommand{\citeauthor}
	{\boolfalse{citetracker}%
	\boolfalse{pagetracker}%
	\usebibmacro{prenote}}
	{\ifciteindex
		{\indexnames{labelname}}
		{}%
		\printtext[bibhyperref]{\printnames{labelname}}}
	{\multicitedelim}
	{\usebibmacro{postnote}}

\DeclareFieldFormat{backrefparens}{#1.}

\defbibfilter{education}{
  	keyword={education}
}

\defbibfilter{graph}{
  	keyword={graph}
}

\defbibfilter{no_educ}{
  	not keyword={education}
}

\begin{document}

\title{Some polynomial classes for the acyclic orientation with parity constraint problem }

\author[1,3]{Sylvain Gravier}
\author[1,3]{Matthieu Petiteau}
\author[2,3]{Isabelle Sivignon}

\affil[1]{Univ. Grenoble Alpes, CNRS, Institut Fourier, 38000 Grenoble, France}
\affil[2]{Univ. Grenoble Alpes, CNRS, Grenoble INP, GIPSA-lab, 38000 Grenoble, France}
\affil[3]{Univ. Grenoble Alpes, Maths à Modeler, 38000 Grenoble, France}

\date{\today}
\maketitle
 
\abstract{

We study the problem of finding an acyclic orientation of an undirected
graph with constrained in-degree parities specified by a
subset $T$ of vertices. An orientation is called $T$-odd if a vertex $v$
has odd in-degree if and only if $v \in T$. While the unconstrained
parity orientation problem is polynomial (Chevalier
et al. (1983)), imposing acyclicity makes it more
challenging, and its complexity remains an open question. Szegedy and
Szegedy (2006) proposed a randomized polynomial-time algorithm for
this problem, but it is not known whether it belongs to
co-$\NP$. Furthermore, Gravier et al. (2025) showed the problem becomes
$\NP$-complete on partially directed graphs, even when restricted to
planar cubic graphs. 

\noindent We identify three necessary conditions for the existence of
acyclic $T$-odd orientation: a global parity condition $\calP$, and two
conditions $\calS$ and $\overline{\calS}$ ensuring the existence of
potential sources and sinks. Following the 
work of Frank and Kiraly (2002), we define graph
classes containing the graphs for which a given subset of
the necessary conditions $\calP$, $\calS$ and $\overline{\calS}$ is also
sufficient for the existence of an acyclic $T$-odd orientation. We 
establish the inclusion relationships between these classes.

\noindent We complete the study of these classes by a characterization
of the solvable instances for Cartesian products of paths and
cycles. The proofs of these results are all constructive, so that 
acyclic $T$-odd orientations can be built in polynomial time whenever
they exist. We use these families, along with cliques, to demonstrate
the strictness of the class inclusions in our hierarchy.
}

\tableofcontents

\section{Introduction}

\noindent In this paper, we study an orientation problem with parity constraints on simple graphs.
We first consider a local constraint on the parity of the in-degree of vertices.
For a graph $G$ and a subset $T \subseteq V(G)$, an orientation of $G$ is called $T$-\textbf{odd} if any vertex $v \in V$ has an odd in-degree if and only if it is in $T$.
Chevalier, Jaeger, Payan, and Xuong (1983) \cite{chevalierOddRootedOrientations1983} demonstrated that finding a $T$-odd orientation of an undirected graph $G$ can be solved in polynomial time.
Specifically, they showed that such an orientation exists if and only if $T$ satisfies $|E| + |T| \equiv_2 0$. In the following, we call this property $\calP$.\newline

\noindent Frank and Király (2002) \cite{frankGraphOrientationsEdgeconnection2002} later considered adding a $k$-connectivity global constraint.
They provided a characterization of undirected graphs $G = (V, E)$ for which there exists a $k$-arc-connected $T$-odd orientation if and only if $T \subset V$ satisfies $\calP$.
This work also raised questions about other global constraints such as acyclicity, which can be seen as another type of connectivity constraint.
This leads to the following problem:

\begin{problem}[Acyclic Orientation with Parity Constraints]\label{prob:AOP}\ \newline
	\textbf{Instance:} A pair $(G,T)$ with $G$ an undirected graph and $T \subseteq V(G)$ a subset of its vertex set.\newline
	\textbf{Question:} Does $G$ have an acyclic $T$-odd orientation?
\end{problem}

\noindent The characterization of the satisfying instances of this problem is referenced on the Egres platform\footnote{\url{https://lemon.cs.elte.hu/egres/open/Acyclic_orientation_with_parity_constraints}} and is known as the "Acyclic orientation with parity constraints" problem.
The complexity of this decision problem remains an open question.
Szegedy (2005) \cite{szegedyApplicationsWeightedCombinatorial2005} proposed a randomized polynomial algorithm to answer Problem \ref{prob:AOP}, it is closely related to results on matroïds presented also in Szegedy, Szegedy (2006) \cite{szegedySymplecticSpacesEarDecomposition2006}; however, it is not known whether it belongs to co-$\NP$, and the problem becomes $\NP$-complete on partially directed graphs, even restricted to graphs being both planar and cubic \cite{gravierNoteComplexityAcyclic2025}.\newline

\noindent Following the work of Frank and Kiraly (2002) \cite{frankGraphOrientationsEdgeconnection2002}, an associated problem considering acyclicity is also introduced on the Egres platform:
\begin{problem}[]\label{prob:AOP_allT}\ \newline
	Find a good characterization for undirected graphs $G=(V,E)$ which, for every possible $T\subsetneq V$ satisfying $\calP$, have an acyclic orientation $T$-odd orientation.
\end{problem}

\noindent The condition $\calP$ remains a necessary condition for the existence of an acyclic $T$-odd orientation.
Note also that the acyclicity constraint gives rise to two other necessary conditions: the existence of a potential source and a potential sink.
Any acyclic orientation of a finite graph $G$ must have at least one vertex with in-degree $0$ and one vertex with out-degree $0$.
These vertices are respectively called a \textbf{source} and a \textbf{sink}. 
In the context of $T$-odd orientations, a vertex is called a \textbf{potential source} if it belongs to $Source(T) = V \setminus T$ and a \textbf{potential sink} if it belongs to $Sink(T) = \{v \in T, d(v) \equiv_2 1\} \cup \{v \in V\setminus T, d(v) \equiv_2 0\}$.
These names arise from the fact that, in any acyclic $T$-odd orientation of $G$, the sources and sinks must be potential sources and potential sinks, respectively.
It is also worth noticing that if $Source(T) = Sink(T) = \{v\}$ for some vertex $v$, then no acyclic $T$-odd orientation exists unless the graph is a singleton.
This last necessary condition is relatively weak; indeed, any odd degree vertex in $T$ belongs to $Sink(T)\setminus Source(T)$ and any odd degree vertex in $V\setminus T$ belongs to $Source(T)\setminus Sink(T)$, while even degree vertices in $V\setminus T$ belong to $Source(T)\cap Sink(T)$, so $Source(T)=Sink(T)$ if and only if the graph has no odd degree vertex. Thus, we do not consider it apart from the other conditions, and for symmetry reasons, we consider it in both the \textbf{source condition} $\calS$ and \textbf{sink condition} $\overline\calS$.
This leads to the following necessary conditions for the existence of an acyclic $T$-odd orientation:
\begin{itemize}
	\item [$\calP$:] $|E| + |T| \equiv_2 0$.

	\item [$\calS$:] $Source(T) \neq \emptyset$ and if $|Source(T)|=1$ then either
	$\begin{cases}
    |V|=1, \\
    \text{or} \\
    Source(T) \neq Sink(T).
    \end{cases}$

	\item [$\overline\calS$:] $Sink(T) \neq \emptyset$ and if $|Sink(T)|=1$ then either
	$\begin{cases}
    |V|=1, \\
    \text{or} \\
    Sink(T) \neq Source(T).
    \end{cases}$
\end{itemize}

\noindent We now define the classes $\cC_{\calP}, \cC_{\calS}, \cC_{\overline\calS}, \cC_{\calP\calS}, \cC_{\calP\overline\calS}, \cC_{\calS\overline\calS}, \cC_{\calP \calS \overline\calS}$ as the sets of graphs for which the indexed necessary conditions are also sufficient for the existence of an acyclic $T$-odd orientation.
Specifically, a graph belongs to $\cC_{\cN}$ for $\cN \subseteq \{\calP\calS\overline\calS\}$ if it admits an acyclic $T$-odd orientation for any subset $T$ satisfying all necessary conditions in $\cN$, thereby implying that all the necessary conditions in $\{\calP\calS\overline\calS\} \setminus \cN$ are always satisfied for all such such $T$.
In particular, this gives straightforward inclusions between these classes: if $\cN' \subseteq \cN \subseteq \{\calP\calS\overline\calS\}$, then $\cC_{\cN'} \subseteq \cC_{\cN}$.

\noindent Our first result is the description of all the classes mentioned above. The proof of this theorem is the object of Section \ref{sec:classes}.

\begin{theorem}\label{thm:treilli}~
    \begin{enumerate}[label= $(\alph*)$]
        \item $\cC_{\calS} = \cC_{\overline\calS}= \{\OneCirc, \TwoCirc\}$,
        \item $\cC_{\calS\overline\calS}=\{\OneCirc, \TwoCirc, \TwoCircLinked\}$,
        \item $\cC_{\calP}= \{\OneCirc\}\cup \{G \in \cC_{\calP \calS \overline\calS} \ | \ G \text{ is connected, non-Eulerian, and } |V(G)| + |E(G)| \equiv_2 1\}$,

        \item $\cC_{\calP\calS} \setminus \cC_{\calP}
        =\cC_{\calP\overline\calS} \setminus \cC_{\calP}
        = \{\TwoCirc, \ThreeCircIndependant\}\cup \{G\in \cC_{\calP \calS \overline\calS} \mid G \text{ is Eulerian}\}$,

		\item $\cC_{\calP \calS \overline\calS} \setminus  \cC_{\calP \calS} \subset \{\ThreeCirc\} \cup \{G \ | \ G \text{ is connected, non-Eulerian, and } |V(G)| + |E(G)| \equiv_2 0\}$.
    \end{enumerate}
\end{theorem}

\noindent With this new formalism, Problem \ref{prob:AOP_allT} is essentially asking if $G$ belongs to $\cC_{\calP \calS}$.
 To put it in perspective, let us recall a result from Kiraly and Kisfaludi-Bak:
\begin{theorem}[Theorem 1.8 in \cite{kiralyDualCriticalGraphsNotes}]\label{thm:cp_carac}~\newline
	Let $G = (V,E)$ be a graph.\newline
	$G$ admits an acyclic $T$-odd orientation for all $T \subseteq V$ satisfying $|E(G)| + |T| \equiv_2 0$ if and only if $G$ admits an acyclic $T'$-odd orientation with $T' = V \setminus \{v\}$ for some vertex $v \in V$.
\end{theorem}

\noindent Note the distinction between Theorem \ref{thm:cp_carac} and the question in Problem \ref{prob:AOP_allT}: Theorem \ref{thm:cp_carac} gives a characterization of graphs in $\cC_{\calP}$ while Problem \ref{prob:AOP_allT} asks for an analogous characterization restricted to sets $T \subsetneq V$. This difference is of importance, indeed, as we shall see later, there exist graphs that are not in $\cC_{\calP}$ but admit acyclic $T$-odd orientations for every $T \subsetneq V$ satisfying $\calP$. 

\noindent Moreover, the constraint $T \subsetneq V$ guarantees the existence of a source but breaks the symmetry with the existence of sinks. We restore this symmetry with the study of the classes $\cC_{\calP \calS}, \cC_{\calP \overline\calS}$ and $\cC_{\calP \calS \overline\calS}$.\newline


\noindent In a second section, we study classes of instances defined from the Cartesian product of cycles and/or paths.
We define the \textbf{grid}, \textbf{torus}, and \textbf{cylinder} as the Cartesian product of two paths, two cycles, and a cycle and a path, respectively (see Figures \ref{fig:cylinder} and \ref{fig:tore}).
We then give a complete characterization of sets $T$ satisfying all three necessary conditions $\calP \calS \overline\calS$ such that $G$ admits an acyclic $T$-odd orientation, where $G$ is either a grid, a torus, or a cylinder.
In particular, as a second main result, we prove the following two theorems:

\begin{theorem}\label{thm:tore_carac}
Let $G = C_p \square H_q$ be a graph where $C_p$ is a cycle on $p$ vertices and $H_q$ is a path or a cycle on $q$ vertices.
Let $T \subset V(G)$ be a subset that satisfies $\calP \calS \overline\calS$.
\begin{center}
	If $H_q$ is a path or if $p,q \geq 4$, then $G$ has an acyclic $T$-odd orientation.\newline
\end{center}
\end{theorem}

\begin{theorem}\label{thm:grid_carac}
Let $G = P_p \square P_q$ be a graph where $P_p$ and $P_q$ are paths with vertex sets $\{v_1, \dots, v_p\}$ and $\{w_1, \dots, w_q\}$, respectively.
Let $T \subset V(G)$ be a subset that satisfies $\calP$.
\begin{center}
	$P_p \square P_q$ has no acyclic $T$-odd orientation if and only if $p,q \equiv_2 0$ and $V \setminus T$ satisfies that:
\end{center}
\begin{itemize}
	\item $(v_{2k}, w_j) \in V \setminus T$ if and only if $(v_{2k+1}, w_j) \in V \setminus T$ for all $k \in \{1, \dots, \frac{p}{2}-1\}$ and $j \in \{1, q\}$.
	\item $(v_i, w_{2k}) \in V \setminus T$ if and only if $(v_i, w_{2k+1}) \in V \setminus T$ for all $i \in \{1, p\}$ and $k \in \{1, \dots, \frac{q}{2}-1\}$ symmetrically.
\end{itemize}
\end{theorem}

\begin{figure}[ht]
	\begin{minipage}{6cm}
		\centering
	    \includegraphics[width=\textwidth]{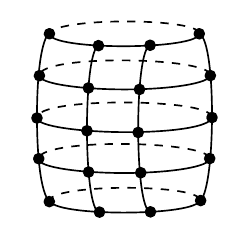}
		\caption{Cylinder graph $C_4 \times P_5$}
		\label{fig:cylinder}
	\end{minipage}
	\hfill
	\begin{minipage}{6cm}
		\centering
		\includegraphics[width=\textwidth]{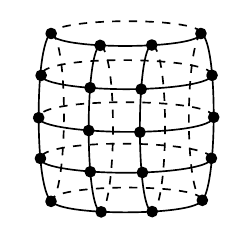}
		\caption{Torus graph $C_4 \times C_5$}
		\label{fig:tore}
	\end{minipage}
\end{figure}

\noindent As a corollary of these characterizations, we show that the inclusions between the classes $\cC_{\calP}$, $\cC_{\calP \calS}$, and $\cC_{\calP \calS \overline\calS}$ are actually strict (see Figure \ref{fig:treilli}):

\begin{corollary}\label{thm:strict_inclusion} Let $p,q,k \geq 1$ be three integers.
    \begin{itemize}
        \item [(c)'] The Cylinders $C_{2p+1}\square P_{2q}$, the Grids $P_{2p+1}\square P_q$, and Trees are in $\cC_{\calP}$.
        \item [(d)'] The Tori $C_{p}\square C_q$ with $p,q \geq 4$ and Cycles are in $\cC_{\calP\calS}\setminus\cC_{\calP}$.
        \item [(e)'] The Cylinders $C_{2p}\square P_q, C_{2p+1}\square P_{2q+1}$ are in $\cC_{\calP \calS \overline\calS}\setminus\cC_{\calP \calS}$.
	\end{itemize}
	Moreover,
	\begin{itemize}
        \item [(c)''] The Cliques $K_{4k+2} \in \{G \text{ connected and non-Eulerian} \mid |V(G)| + |E(G)| \equiv_2 1\} \setminus \cC_{\calP}$.
		\item [(d)''] The Cliques $K_{2k+1} \in \{ \circ \ \circ\} \cup \{G \mid G \text{ is Eulerian}\} \setminus \cC_{\calP \calS}$.
		\item [(e)''] The Cliques $K_{4k} \in \{G \text{ connected and non-Eulerian}  \mid |V(G)| + |E(G)| \equiv_2 0\} \setminus \cC_{\calP \calS \overline\calS}$.
    \end{itemize}
\end{corollary}

\begin{figure}[ht]
	\centering
	\begin{tikzpicture}

\coordinate (S) at (-2,5);
\coordinate (SB) at (0,5);
\coordinate (P) at (2,5);

\coordinate (SSB) at (-1,3.5);

\coordinate (PS) at (-1,1.5);
\coordinate (PSB) at (1,1.5);

\coordinate (PSSB) at (0,0);

\draw (S) -- (SSB) node[midway,below,sloped,fill=white,inner sep=1pt] {$\subsetneq$};
\draw (SB) -- (SSB);

\draw (SSB) -- (PS) node[midway,above,sloped] {$\subsetneq$};
\draw (PS) -- (P) node[midway,above,sloped,fill=white,inner sep=1pt] {$\supsetneq$};

\draw (PS) -- (PSSB) node[midway,below,sloped,fill=white,inner sep=1pt] {$\subsetneq$};
\draw (PSB) -- (PSSB);

\node[fill=white,inner sep=2pt] at (S) {$C_{\calS}$};
\node[fill=white,inner sep=2pt] at (SB) {$C_{\bar{\calS}}$};
\node[fill=white,inner sep=2pt] at (P) {$C_{\calP}$};

\node[fill=white,inner sep=2pt] at (SSB) {$C_{\calS \bar{\calS}}$};

\node[fill=white,inner sep=2pt] at (PS) {$C_{\calP \calS}$};
\node[fill=white,inner sep=2pt] at (PSB) {$C_{\calP \bar{\calS}}$};
\node[fill=white,inner sep=2pt] at (PSSB) {$C_{\calP \calS \bar{\calS}}$};

\node at (-1,5) {$=$};
\node at (0,1.5) {$=$};

\end{tikzpicture}
	\caption{Theorem \ref{thm:treilli}. }\label{fig:treilli}
\end{figure}

\noindent The paper is structured in three major sections. Section \ref{sec:classes} is dedicated to the proof of Theorem \ref{thm:treilli} on graph classes. In Section \ref{sec:cartesian_product} we prove Theorem \ref{thm:tore_carac} and \ref{thm:grid_carac} on graphs obtained from the Cartesian product of paths and cycles. Finally, Section \ref{sec:class_inclusion} is for the proof of this last corollary.

\section{Definitions and preliminaries}\label{sec:preliminary}

A \textbf{graph} $G$ is a pair $(V(G),E(G))$ where $V(G)$ is a non-empty set of vertices and $E(G)$ is a set of edges. We only consider simple graphs i.e., loopless and without multi-edges.
For a subset $X \subseteq V(G)$, the \textbf{subgraph} $G[X]$ is the graph whose vertex set is $X$ and whose edge set is $E(G[X]) = \{ (u, v) \mid u \in X, v \in X, (u, v) \in E(G) \}$.
For subsets $X, Y\subseteq V(G)$, the \textbf{cut set} $\delta(X, Y)$ is the set of edges with one endpoint in $X$ and the other in $Y$.
When $Y = V(G) \setminus X$, we simply write $\delta(X)$.
The \textbf{degree} of a vertex $v$, denoted $d(v)$, is the number of edges incident to $v$, i.e. $|\delta(v)|$.\newline

\noindent Let $G$ and $H$ be two graphs.
The \textbf{Cartesian product} of $G$ and $H$, denoted by $G \square H$, is the graph whose vertex set is $V(G \square H) = V(G) \times V(H) = \{ (u, v) \mid u \in V(G), v \in V(H) \}$.
Two vertices $(u, v)$ and $(u', v')$ in $V(G \square H)$ are adjacent if and only if [$u = u'$ and $(v, v') \in E(H)$] or [$v = v'$ and $(u, u') \in E(G)$].\newline

\noindent A \textbf{directed graph} $D$ consists of a non-empty set of vertices $V(D)$ and a set of arcs $A(D)$, where an arc is an ordered pair of vertices.
An arc is denoted $\dir{uv}$ or $\dir{vu}$ according to the ordering.
For subsets $X,Y\subseteq V(D)$, the \textbf{cut set} $\delta_D(X,Y)$ is the set of arcs with one endpoint in $X$ and the other in $Y$.
For a vertex $v \in V(D)$, we define the \textbf{out-degree} $d_D^+(v)$ and the \textbf{in-degree} $d_D^-(v)$ to be the number of arcs directed outward and inward from $v$, respectively.
The degree $d_D(v)$ is $d_D^+(v) + d_D^-(v)$.
The subscript can be omitted when there is no ambiguity on the concerned graph.\newline

\noindent Given a graph $G$ and a set $S\subset E(G)$, an \textbf{orientation} $O$ of $S$ is a set of arcs such that for every edge $uv \in S$, either $\dir{uv} \in O$ or $\dir{vu} \in O$, but not both.
An {\bfseries orientation $O$ of a graph $G$} is an orientation of $E(G)$, and we denote by $G_O = (V(G), O)$ the associated directed graph.
A set of arcs $A$ is said to be \textbf{acyclic} if the subgraph induced by $A$ does not contain any directed cycle.
By extension, a directed graph $D$ is acyclic if its arc set $A(D)$ is acyclic.\newline

\noindent For the figures, given $(G,T)$ an instance of Problem \ref{prob:AOP}, we refer to vertices in $T$ as \textbf{black} vertices and vertices in $V(G) \setminus T$ as \textbf{white} vertices.
The \textbf{gray} vertices symbolize vertices that can be either black or white.
The \textbf{blue sets} are sets of vertices of the same color; hence, a blue set of gray vertices is a set whose vertices are either all white or all black.\newline

\section{Classes of necessary conditions: proof of Theorem \ref{thm:treilli}}\label{sec:classes}

\noindent First, we prove the following useful lemma.	Let $\cG_{\leq 3}$ be the set of graphs with at most $3$ vertices.

\begin{lemma}\label{lem:smallcases}

	The following holds:
	\begin{align*}
		\cC_{\calS} \cap \cG_{\leq 3} = \cC_{\overline\calS} \cap \cG_{\leq 3} &= \{\OneCirc, \TwoCirc\}\\
        \cC_{\calS\overline\calS} \cap \cG_{\leq 3} &= \{\OneCirc, \TwoCirc, \TwoCircLinked\}\\
        \cC_{\calP} \cap \cG_{\leq 3} &= \{\OneCirc, \TwoCircLinked, \ThreeCircPath\}\\
        \cC_{\calP\calS} \cap \cG_{\leq 3} = \cC_{\calP \overline\calS} \cap \cG_{\leq 3} &= \{\OneCirc, \TwoCirc, \TwoCircLinked, \ThreeCircPath, \ThreeCircTriangle, \ThreeCircIndependant\}\\
		\cC_{\calP \calS \overline\calS} \cap \cG_{\leq 3} &= \{\OneCirc, \TwoCirc, \TwoCircLinked, \ThreeCircPath, \ThreeCircTriangle, \ThreeCircIndependant, \ThreeCirc\}
	\end{align*}
\end{lemma}
\begin{proof}
	\noindent (\OneCirc) The one-vertex graph admits an acyclic $T$-odd orientation if and only if $T = \emptyset$; this is already imposed if either $\calP$, $\calS$, or $\overline\calS$ is satisfied.
	Thus, it belongs to $\cC_{\calP}$, $\cC_{\calS}$, $\cC_{\overline\calS}$ and by inclusion, to $\cC_{\calS  \overline\calS}$, $\cC_{\calP \calS}$, $\cC_{\calP \overline\calS}$, and $\cC_{\calP \calS \overline\calS}$.\newline

	\noindent (\TwoCircLinked) This graph admits an acyclic $T$-odd orientation if and only if $|T| = 1$.
	Hence, it does not belong to $\cC_{\calS}$ or $\cC_{\overline\calS}$ since a set $T'$ with $|T'| = 0$ would satisfy $\calS$ but not $\overline\calS$, and $|T'|=2$ would satisfy $\overline\calS$ but not $\calS$.
	But it indeed belongs to $\cC_{\calS  \overline\calS}$ as any set $T'$ satisfying both $\calS$ and $\overline\calS$ has to be of size $1$.
	It also belongs to $\cC_{\calP}$ as any set $T'$ satisfying $\calP$ verifies $|T'| \equiv_2 |E| \equiv_2 1$.
	Hence, by inclusion, it also belongs to $\cC_{\calP \calS}$, $\cC_{\calP \overline\calS}$, and $\cC_{\calP \calS \overline\calS}$.\newline

	\noindent (\TwoCirc) This graph admits an acyclic $T$-odd orientation if and only if $|T| = 0$.
	Hence, it does not belong to $\cC_{\calP}$ since a set $T'$ with $|T'| = 2$ would satisfy $\calP$ but not $\calS$.
	But it indeed belongs to $\cC_{\calS}$ and $\cC_{\overline\calS}$ as $T=\emptyset$ satisfies both conditions, but the set $T'$ where $|T'|\geq 1$ would not satisfy $\calS$ nor $\overline\calS$ as $|V \setminus T' \cup \{v \in T', d(v) \equiv_2 1\}| > 1$ would not be satisfied.
	By inclusion, it also belongs to $\cC_{\calS  \overline\calS}$, $\cC_{\calP \calS}$, $\cC_{\calP \overline\calS}$, and $\cC_{\calP \calS \overline\calS}$.\newline

	\noindent (\ThreeCircPath) This graph admits an acyclic $T$-odd orientation if and only if $|T|$ is even.
	This graph does not belong to $\cC_{\calS}$ nor $\cC_{\overline\calS}$ since a set $T'$ containing only a vertex of degree $1$ would satisfy both $\calS$ and $\overline\calS$ but not $\calP$.
	However, it does belong to $\cC_{\calP}$ as any set $T'$ satisfying $\calP$ precisely satisfies $|T'| \equiv_2 |E| \equiv_2 0$.
	By inclusion, it also belongs to $\cC_{\calP \calS}$, $\cC_{\calP \overline\calS}$, and $\cC_{\calP \calS \overline\calS}$.\newline

	\noindent (\ThreeCircTriangle) This graph admits an acyclic $T$-odd orientation if and only if $|T| = 1$.
	This graph does not belong to $\cC_{\calS}$ nor $\cC_{\overline\calS}$ since a set $T'$ with $|T'| = 0$ would satisfy both $\calS$ and $\overline\calS$.
	It also does not belong to $\cC_{\calP}$ since a set $T'$ with $|T'| = 3$ would satisfy $\calP$ but not $\calS$ nor $\overline \calS$.
	However, it does belong to $\cC_{\calP \calS}$ and $\cC_{\calP \overline\calS}$ because if a set $T'$ satisfies $\calP$, it verifies $|T'| \equiv_2 1$, and if it satisfies either $\calS$ or $\overline\calS$, it verifies that $|V \setminus T'| > 1$.
	Thus, it imposes that $|T'| = 1$.
	By inclusion, it also belongs to $\cC_{\calP \calS \overline\calS}$.\newline

	\noindent (\ThreeCircIndependant) This graph admits an acyclic $T$-odd orientation if and only if $|T| = 0$.
	This graph does not belong to $\cC_{\calS \overline\calS}$ since the set $T'$ with $|T'|=1$, would satisfy both $\calS$ and $\overline\calS$.
	It also does not belong to $\cC_{\calP}$ since a set $T'$ with $|T'| = 2$ would satisfy $\calP$.
	However, it does belong to $\cC_{\calP \calS}$ and $\cC_{\calP \overline\calS}$ because if a set $T'$ satisfies $\calP$, it verifies $|T'| \equiv_2 0$, and if it satisfies either $\calS$ or $\overline\calS$, it verifies that $|V \setminus T'| > 1$.
	Thus, it imposes that $|T'| = 0$.
	By inclusion, it also belongs to $\cC_{\calP \calS \overline\calS}$.\newline

	\noindent (\ThreeCirc) This graph does not belong to $\cC_{\calS \overline\calS}$ since any set $T'$ with $|T'|=2$ would satisfy both $\calS$ and $\overline\calS$ but not $\calP$.
	It also does not belong to $\cC_{\calP \overline\calS}$ since a set $T'$ with $|T'| = 3$ would satisfy $\calP$ and $\overline\calS$ but not satisfy $\calS$.
	It does not belong to $\cC_{\calP \calS}$ since the set $T'$ composed only of the isolated vertex would satisfy both $\calP$ and $\calS$ but not $\overline\calS$.
	However, it does belong to $\cC_{\calP \calS \overline\calS}$ as the only possible sets $T'$ satisfying $\calP \calS \overline\calS$ are the ones containing only a single vertex of degree one.
	Such sets always admit an acyclic $T'$-odd orientation.
\end{proof}

\begin{lemma}\label{lem:connexe}
	Let $G\in  \cC_{\calP\calS\overline\calS}$ with $|V(G)|\geq 4$.
	Then $G$ is connected.
\end{lemma}

\begin{proof}
	\noindent By way of contradiction, suppose that $G$ is not connected.

	\noindent \textbf{a.~} Assume first that every vertex in $G$ has even degree.
	Let $H$ be the smallest connected component, either it is of size one, or it has as least $3$ vertices since the graph is simple and all degrees are even. By minimality of $H$ and because $|V(G)|\geq 4$, we have $\vert V\setminus H\vert \geq 3$. 
	Select any $A\subset V\setminus H$ with $\vert A\vert =2$ if $\vert V(G)\vert +\vert E(G)\vert \equiv_2 0$ else $\vert A\vert =3$. This choice of $A$ implies that $|A| \equiv_2 |V(G)| + |E(G)|$ and in particular that the set $T=V(G)\setminus A$ satisfies $\calP$.
	Moreover $T$ satisfies $\calS$ and $\overline\calS$ since $A \subseteq Source(T) \cap Sink(T)$ and $|A|\geq 2$.
	Nevertheless, $G$ has no acyclic $T$-odd orientation since all vertices of $H$ are in $T$, hence $T(H)$ violates $\calS$ in $G[H]$.\newline

	\noindent \textbf{b.~} Assume now that there exist two vertices $u,v\in V(G)$ with odd degree.\newline
	\noindent\textbf{-~} If $|V(G)|+|E(G)|$ is odd, then choose $T = V(G)\setminus \{x\}$ for any vertex $x \neq u,v$.
	Now, $T$ satisfies $\calP$.
	It also satisfies $\calS$ and $\overline\calS$ since $\{x\} = Source(T)$, $u,v \in Sink(T)$ and $Source(T)\neq Sink(T)$.
	However, there is at least one other connected component than the one containing $x$, and all its vertices are in $T$, hence it fails to satisfy $\calS$, and $G$ has no acyclic $T$-odd orientation.\newline

	\noindent\textbf{-~} If $|V(G)|+|E(G)|$ is even, choose $T = V(G)\setminus \{u, x\}$ with $x$ in the same connected component as $u$ ($x$ is guaranteed to exist since $u$ has odd degree).
	Note that $T$ satisfies $\calP$.
	If $x$ is different from $v$, then $T$ also satisfies $\calS$ and $\overline\calS$ since $\{u, x\} = Source(T)$, $\{v\} \in Sink(T)$ and $Source(T) \neq Sink(T)$.
	But any other connected component different from the one containing $u$ has all its vertices in $T$ and contains no potential source, so $G$ has no acyclic $T$-odd orientation.\newline
	Otherwise, if $x=v$, then we may finally assume that $u$ and $v$ form by themselves a connected component of two vertices.
	We now take $T = V(G)\setminus \{y, z\}$ with both $y,z$ different from $u,v$ (they are guaranteed to exist since $|V(G)| \geq 4$).
	This choice of $T$ ensures that it satisfies $\calP \calS \overline\calS$ since $\{y,z\} \in Source(T)$, $\{u,v\} \in Sink(T)$.
	Nevertheless, $G$ has no acyclic $T$-odd orientation since $u,v \in T$ and the connected component composed of $u,v$ fails to satisfy $\calS$.
\end{proof}

\noindent The following subsections are dedicated to prove the different items of Theorem \ref{thm:treilli}.

\subsection{Proof of (a) and (b)}

From Lemma \ref{lem:smallcases}, we have $\{\OneCirc, \TwoCirc\} = \cC_{\calS} \cap \cG_{\leq 3} = \cC_{\overline\calS} \cap \cG_{\leq 3}$ and $\{\OneCirc, \TwoCirc, \TwoCircLinked\} = \cC_{\calS\overline\calS} \cap \cG_{\leq 3}$. In fact, we show those are the only graphs in $\cC_{\calS}, \cC_{\overline\calS}, \cC_{\calS\overline\calS}$.



\noindent Let $G$ be any graph with $|V(G)| \geq 4$. We find $T$ that satisfies $\calS$ and $\overline \calS$ but not $\calP$. Consider $T \subsetneq V(G)$ such that $|T| = 2$ if $|E(G)|\equiv_2 1$, and $|T| = 1$ otherwise. By construction, $T$ does not satisfy $\calP$. Since $T \neq V(G)$, and because $|V(G)| \geq 4$, we have $|Source(T)| = |V(G)\setminus T| \geq 2$, therefore $T$ satisfies $\calS$. If all vertices are of even degree, then $Source(T) = Sink(T)$ and $\overline\calS$ is also satisfied. Otherwise, $T$ can be chosen to contain a vertex of odd degree, so that $\overline\calS$ is satisfied. Hence:
$$\cC_{\calS} = \cC_{\overline\calS} = \{\OneCirc, \TwoCirc\} \subsetneq \cC_{\calS\overline\calS} = \{\OneCirc, \TwoCirc, \TwoCircLinked\}.\hbox{\qed}$$

\subsection{Proof of (c)}
This part of Theorem \ref{thm:treilli} could be proved using Theorem \ref{thm:cp_carac}, instead, we give an independent proof.\newline

\noindent We first prove:
\begin{equation*}
    \cC_{\calP} \subseteq \{\OneCirc\}\cup \{G\in \cC_{\calP \calS \overline\calS} \ | \ G \text{ is connected, non-Eulerian and } |V(G)| + |E(G)| \equiv_2 1\}
\end{equation*}
\noindent Let $G\in\cC_{\calP}$. By Lemmas \ref{lem:smallcases} and \ref{lem:connexe}, $G$ is connected.\newline

\noindent Assume that $|V(G)| + |E(G)| \equiv_2 0$.
Let $T = V(G)$.
Since $|E(G)| + |T| \equiv_2 0$, then $T$ satisfies $\calP$.
Moreover, $T$ does not satisfy $\calS$.
Hence $G$ satisfy $|V(G)| + |E(G)| \equiv_2 1$.\newline

\noindent By contradiction, suppose that $G$ is an Eulerian graph with $|V(G)|>1$.
Set $T = V(G)\setminus \{v\}$ for some vertex $v\in V(G)$.
This choice ensures that $T$ satisfies $\calP$.

\noindent But notice that $Source(T) = Sink(T) = \{v\}$ in this case, but since $|V(G)|>1$ neither $\calS$ nor $\overline \calS$ are satisfied.
Thus, graphs $G\in \cC_\calP$ cannot be Eulerian.

\noindent We now show the other inclusion.  $$\cC_{\calP} \supseteq \{\OneCirc\}\cup \{G\in \cC_{\calP \calS \overline\calS} \ | \ G \text{ is connected, non-Eulerian and } |V(G)| + |E(G)| \equiv_2 1\}.$$

\noindent Clearly, $\OneCirc\in \cC_{\calP}$. Let $G \in \cC_{\calP\calS\overline\calS}$ be a connected non-Eulerian graph with $|V(G)| + |E(G)| \equiv_2 1$ and $|V(G)|\geq 1$. We will show that $G \in \cC_{\calP}$.
\noindent Consider any set $T \subseteq V(G)$ satisfying $\calP$, i.e. $|E(G)| + |T| \equiv_2 0$. This implies $|Source(T)| = |V(G)\setminus T| \equiv_2  1$. Note that, $G$ being non-Eulerian, $G$ has a non-null even number of vertices of odd degree, therefore, $Source(T) \neq Sink(T)$. If $T$ contains at least one vertex of odd degree, then $Sink(T)\neq \emptyset$. Otherwise, $Source(T)$ contains all vertices of odd degree (non empty set), and at least one vertex of even degree, which belongs to $Sink(T)$. In both cases, $Source(T)\neq \emptyset$, $Sink(T)\neq \emptyset$ and $Source(T) \neq Sink(T)$ proving that $\calS$ and $\overline \calS$ are satisfied.
 
\noindent Since $G\in \cC_{\calP\calS\overline\calS}$, then $G$ admits an acyclic $T$-odd orientation.
Hence $G\in \cC_{\calP}$.\qed

\subsection {Proof of (d)}
\noindent We show:

\begin{equation*}
 \cC_{\calP\calS}\setminus \cC_\calP = \{\TwoCirc, \ThreeCircIndependant\}\cup \{G\in  \cC_{\calP\calS\overline\calS}\setminus \cC_\calP\ | \ G \hbox{ is Eulerian}\} =  \cC_{\calP\overline\calS}\setminus \cC_\calP.
\end{equation*}

\begin{proof}
	From Lemma \ref{lem:smallcases}, we have $\cC_{\calP\calS} \cap \cG_{\leq 3} = \cC_{\calP \overline\calS} \cap \cG_{\leq 3} = \{\OneCirc, \TwoCirc, \TwoCircLinked, \ThreeCircPath, \ThreeCircTriangle, \ThreeCircIndependant\}$. Furthermore, by $(c)$, $\{\OneCirc, \TwoCircLinked, \ThreeCircPath \} \subseteq \cC_{\calP}$. Hence $\cG_{\leq 3} \cap  \cC_{\calP\calS}\setminus \cC_\calP = \cG_{\leq 3} \cap \cC_{\calP\overline\calS}\setminus \cC_\calP =  \{\TwoCirc, \ThreeCircIndependant\}$ \newline

	\noindent  Let $G\in  \cC_{\calP\calS}\setminus \cC_\calP$ with $|V(G)|\geq 4$. Clearly, $G\in  \cC_{\calP\calS\overline{\calS}}\setminus \cC_\calP$. \newline
	\noindent By contradiction, assume that $G$ is non-Eulerian. By Lemma \ref{lem:connexe}, $G$ is connected. So, $G$ has a non-null even number of vertices of odd degree. Let us denote by $O(G)$ the set of vertices of odd degree. We show that $G$ does not belong neither to $\cC_{\calP\calS}$ nor to $\cC_{\calP \overline\calS}$.
	Set $T_1=V(G)\setminus O(G)$ and $T_2=V(G)$.\newline 
	The set $T_1$ satisfies $\calS$ but not $\overline \calS$ as $Source(T_1) = O(G)$ but $Sink(T_1) = \emptyset$, while $T_2$ satisfies $\overline \calS$ but not $\calS$ as $Source(T_2) = \emptyset$ but $Sink(T_2) = O(G)$ . Since $|O(G)|\equiv_2 0$, $|T_1| \equiv_2 |T_2| \equiv_2 |V(G)|$. Moreover, since $G \notin \cC_\calP$, and keeping in mind that $G$ is supposed to be connected and non-Eulerian, by item (c) of Theorem \ref{thm:treilli}, we have $|V(G)| +|E(G)| \equiv_2 0$. So $|T_1| \equiv_2 |T_2| \equiv_2 |E(G)|$ and $\calP$ is satisfied. Hence $T_1$ satisfies $\calP$ and $\calS$ but not $\overline \calS$ and $T_2$ satisfies $\calP$ and $\overline \calS$ but not $\calS$, reaching a contradiction. \newline
	
	\noindent For the other inclusion, we show that any graph $G\in \cC_{\calP\calS\overline\calS}$ that has only vertices of even degree (not necessarily connected) belongs to $\cC_{\calP \calS}\setminus \cC_{\calP}$ and $\cC_{\calP \overline{\calS}}\setminus \cC_{\calP}$ . Let $T\subseteq V(G)$.
	Since all vertices are of even degree, we have that $Source(T) = Sink(T)$ and thus:
	\begin{equation*}
		T \hbox{ satisfies } \calS \iff T \hbox{ satisfies } \overline\calS.
	\end{equation*}

	\noindent Since Eulerian graphs have only vertices of even degree, we conclude that $\{\TwoCirc, \ThreeCircIndependant\} \cup \{G\in  \cC_{\calP\calS\overline\calS}\setminus \cC_\calP\ | \ G \hbox{ is Eulerian}\} \subseteq \cC_{\calP\calS}\setminus\cC_\calP\cap   \cC_{\calP\overline\calS}\setminus \cC_\calP$.
\end{proof}

\subsection {Proof of (e)}

\noindent All graphs with at most three vertices belong to $\cC_{\calP\calS \overline\calS}$ by Lemma \ref{lem:smallcases}.
Let $G\in \cC_{\calP\calS\overline\calS}\setminus \cC_{\calP\calS}$ with $|V(G)|\geq 4$. By Lemma \ref{lem:connexe}, $G$ is connected.
Moreover, by point \textbf{(d)} of Theorem \ref{thm:treilli}, $G$ is non-Eulerian and by point \textbf{(c)} of Theorem \ref{thm:treilli}, we have $|V(G)| + |E(G)| \equiv_2 0$.
Hence
$$\cC_{\calP\calS\overline\calS}\setminus \cC_{\calP\calS}\subseteq \{\ThreeCirc\}\cup\{G \hbox{ connected and non-Eulerian}  \ | \ |V(G)| + |E(G)| \equiv_2 0\}.$$ \qed

\section{Cartesian product of path and/or cycle}\label{sec:cartesian_product}

This section is dedicated to the proofs of Theorems \ref{thm:tore_carac} and \ref{thm:grid_carac}. In the first subsection, we present a general method, called $T$-decomposition, used in all the proofs. Then, we study basic cases (namely trees and cycles), prior to addressing grids, cylinders, and finally torii.

\subsection{Proofs methods and overview}
We present here the main tools used in the proofs of the next section subsections.
Let us first introduce an equivalent formulation of Problem \ref{prob:AOP}.

\begin{problem}[Game]\label{prob:Jeu_AOP}~\newline
	\textbf{Instance:} A pair $(G,T)$ where $G = (V,E)$ is a graph and $T \subset V$ is a subset of vertices colored black, the other vertices are colored white.\newline
	\textbf{Rules:} This is a single-player game.
At each turn, the player must removes a white vertex; upon doing so, the player flips the color of the vertex' neighbors: a black neighbor becomes white, and a white neighbor becomes black.\newline
	\textbf{Question:} Is it possible to remove all vertices of $G$ while respecting the rules of the game?\newline
\end{problem}

\begin{lemma}\label{lemme:equivalence_Jeu_AOP}
	Problem \ref{prob:Jeu_AOP} is equivalent to Problem \ref{prob:AOP}.
\end{lemma}

\begin{proof}
	Let $G = (V,E)$ be a graph and $T \subset V$ a set of vertices.
\newline
	\noindent $\Rightarrow$~) Let $v_1 > \dots > v_n$ be an elimination order of the vertices that solves Problem \ref{prob:Jeu_AOP}.
Define an acyclic orientation $o$ of the edges of $G$ as follows: for every $v_i > v_j$, orient the edge $v_iv_j$ from $v_i$ to $v_j$.
This orientation is clearly acyclic, and moreover, it is $T$-odd.
Indeed, each vertex is white when it is removed; thus, a vertex initially black (resp. white) has changed color an odd (resp. even) number of times, meaning that an odd (resp. even) number of its neighbors were removed before it.
That is, it has an odd (resp. even) in-degree in $o$.
Therefore, from an elimination order of the vertices, one can construct an acyclic $T$-odd orientation of $G$ in linear time.
\newline

\noindent $\Leftarrow$~) Conversely, let $o$ be an acyclic $T$-odd orientation of $G$.
Since $o$ is acyclic, there exists a total order $v_1 > \dots > v_n$ on the vertices such that for every $v_i > v_j$, the edge $v_iv_j$ is oriented from $v_i$ to $v_j$.
We consider this elimination order in Problem \ref{prob:Jeu_AOP} and notice that a vertex color is flipped as many times as its in-degree in $o$.
Hence, each vertex is white when it is removed and it is a valid elimination order for Problem \ref{prob:Jeu_AOP}.
Thus, from an acyclic $T$-odd orientation of $G$, one can construct an elimination order of the vertices of $G$ in linear time via a topological sort of the orientation.
\end{proof}

\noindent This game understanding of Problem \ref{prob:AOP} is very useful in the proofs, mainly because it allows focusing on the structure of the graph, and on the set $T$ without considering the orientation.
It also provides a constructive way to build an acyclic $T$-odd orientation by giving an elimination order of the vertices.
We often recognize some subgraphs that can be removed from $G$ in the sense of Problem \ref{prob:Jeu_AOP}, usually thanks to an induction argument, and this leaves a smaller graph that can be handled more easily.\newline

\noindent Now, more formally, a \textbf{$T$-decomposition} of a graph $G$ with $T\subseteq V(G)$ is an ordered decomposition $\langle V_0, \dots, V_k \rangle$ of $V(G)$. To each $V_i$, we associate a set $T_i= Z_i \triangle T(V_i)$ where the set $Z_i$ is defined as the set of vertices of $V_i$ that have an odd number of neighbors in $V_0 \cup \dots \cup V_{i-1}$, and $\triangle$ denotes the \textbf{symmetric difference} between two sets (i.e. $X \triangle Y = (X \setminus Y) \cup (Y \setminus X)$ for sets $X$ and $Y$). \newline

\noindent We say that a $T$-decomposition \textbf{satisfies $\calP$} if $T_i$ satisfies $\calP$ in $G[V_i]$ for all $i \in \{ 0, \dots, k\}$, meaning that $|T_i| \equiv_2 |T(V_i)| - |Z_i| \equiv_2 |E(G[V_i])|$.
We say that the decomposition is \textbf{good} if for each $V_i$ the induced subgraph $G[V_i]$ admits an acyclic $T_i$-odd orientation.
In particular, every good $T$-decomposition satisfies $\calP$.\newline

\noindent Good $T$-decompositions provide a structured way to construct acyclic $T$-odd orientations of $G$ by orienting each $G[V_i]$ sequentially.
In the sense of Problem \ref{prob:Jeu_AOP}, this decomposition corresponds to an elimination order where all vertices in $V_1$ are removed first, then all vertices in $V_2$, and so on.
More generally, we claim that:

\begin{theorem}\label{thm:T-odd decomposition}
        $G$ has an acyclic $T$-odd orientation if and only if $G$ admits a good $T$-decomposition.
\end{theorem}

\begin{proof}
    First, suppose that $G$ admits a good $T$-decomposition $\langle V_0, \dots, V_k \rangle$.
For all $0 \leq i \leq k$, let $O_i$ be an acyclic $T_i$-odd orientation of $G[V_i]$ where $T_i = Z_i \triangle T(V_i)$.
Consider the orientation $O$ of $G$ defined as follows. For any edge $uv \in E(G)$ with $u\in V_i, v\in V_j$ and $j\geq i$: 
    \begin{itemize}
        \item {if $j=i$ then $\dir{uv}\in O$ if and only if $\dir{uv}\in O_i$,}
        \item {else $\dir{uv}\in O$.}
    \end{itemize}

\noindent We claim that $O$ is an acyclic $T$-odd orientation of $G$.

\noindent Combining the facts that, in the orientation $O$, all arcs are directed from $V_i$ to $V_j$ for all $i,j \in \{0, \dots, k\}$, and each $O_i$ is acyclic, we obtain that $O$ is acyclic.
Now, let $v\in V_i$, for some $i\in \{0, ..., k\}$.
By definition of $O$, we have:

\begin{itemize}
        \item {if $v\in Z_i$ then $d_{G_O}^-(v)\equiv d_{G[V_i]_{O_i}}^-(v)+1\bmod{2}$,}
        \item {else $d^-_{G_O}(v)\equiv d_{G[V_i]_{O_i}}^-(v)\bmod{2}$.}
    \end{itemize}

\noindent Moreover, by definition of $T_i$, we have

\begin{itemize}
        \item {if $v\in Z_i$, then $v\in T_i$ if and only if $v\not \in T$,}
        \item {else $v\in T_i$ if and only if $v \in T$.}
    \end{itemize}

    \noindent From these last remarks and since each $O_i$ is a $T_i$-odd orientation of $G[V_i]$, $O$ is a $T$-odd orientation of $G$.\newline

\noindent Second, let us show that any acyclic $T$-odd orientation induces a good $T$-decomposition.
Let $O$ be an acyclic $T$-odd orientation of a graph $G$.
Since $O$ is acyclic, there exists a total ordering of the vertices $v_1, \dots, v_n$ such that for any edge $v_i v_j \in E(G)$, $O$ contains the arc $\dir{v_i v_j}$ if $i < j$ and $\dir{v_j v_i}$ otherwise.
Moreover, since $O$ is a $T$-odd orientation, $v_i \in T$ if and only if $d^-_{G_O}(v_i)$ is odd.

\noindent We define the trivial decomposition where $V_i = \{v_i\}$ for all $i \in \{0, \dots, |V(G)|-1\}$.
For each $i$, note that $d^-_{G_O}(v_i) = |\{v_1, \dots, v_{i-1}\} \cap N(v_i)|$.
Thus, $v_i \in T$ if and only if $v_i \in Z_i$.
This implies $T_i = \emptyset$ for all $i$, as $T_i = (Z_i \triangle T) \cap V_i$.
Since each $V_i$ is a single vertex, it trivially admits an $\emptyset$-odd orientation.
Therefore, the decomposition $\{V_0, \dots, V_{n-1}\}$ is a good $T$-decomposition.
\end{proof}

\noindent In the proofs of the following sections, we will often want to prove that some chosen $T$-decomposition $\langle V_0, \dots, V_k \rangle$ is good.
To make things easier, we will try to ensure as often as possible that, for each $V_i$ in the decomposition, $G[V_i]$ belongs to $\cC_\calP$, $\cC_{\calP \calS}$ or $\cC_{\calP \calS \overline\calS}$.
Hence, to show that a $T$-decomposition is good we always prove first that it satisfies $\calP$ and then prove that $T_i$ satisfies $\calP \calS \overline\calS$ in $G[V_i]$ for every $i \in \{ 0, ..., k\}$.
The following lemma will be helpful in this task:

\begin{lemma}\label{lem:decomposition_P}
	Let $G$ be a graph and let $T \subseteq V(G)$ be a subset of vertices satisfying $\calP$.
Let $\cA = \langle V_0, \dots, V_k \rangle$ be a $T$-decomposition of $G$ and let $i$ be any integer in $\{0, \dots, k\}$:
	\begin{center}
		If $T_j$ satisfies $\calP$ in $G[V_j]$ for every $j \in \{0, \dots, k\}\setminus \{i\}$ then $\cA$ satisfies $\calP$.
	\end{center}
\end{lemma}

\begin{proof}
	Let $G$ be a graph and let $T \subseteq V(G)$ be a subset of vertices satisfying $\calP$.
Let $\cA = \langle V_0, \dots, V_k \rangle$ be a $T$-decomposition of $G$ and let $i$ be any integer $\{0, \dots, k\}$.
Assume $T_j$ satisfies $\calP$ in $G[V_j]$ for every $j \in \{0, \dots, k\}\setminus \{i\}$, then:
	\begin{align*}
		|T_i| &\equiv_2 |T(V_i)| - |Z_i| & \text{(by definition)}\\
		&\equiv_2 \left(|T| - \sum_{j \in S_i} |T(V_j)|\right) - |Z_i|\\
		&\equiv_2 \left(|E(G)| - \sum_{j \in S_i} (|E(G[V_j])|+|Z_j|)\right) - |Z_i|  & \text{(by $\calP$)}\\
		&\equiv_2 |E(G)| - \sum_{j \in S_i} |E(G[V_j])| - \sum_{j=0}^{k} |Z_j|\\
		&\equiv_2 |E(G[V_i])| + \sum_{j=0}^{k} |\delta(V_0\cup \dots \cup V_{j-1}, V_j)| - \sum_{j=0}^{k} |Z_j|\\
		&\equiv_2 |E(G[V_i])| + \sum_{j=0}^{k} (|\delta(V_0\cup \dots \cup V_{j-1}, V_j)| - |Z_j|) & \text{(by definition of $Z_j$)}\\
		&\equiv_2 |E(G[V_i])|.
	\end{align*}
	\noindent Therefore, $T_i$ satisfies $\calP$ in $G[V_i]$, thus $T_j$ satisfies $\calP$ in $G[V_j]$ for every $j \in \{0, \dots, k\}$ and $\cA$ satisfies $\calP$.
\end{proof}

\subsection{Basic cases}

\begin{lemma}\label{lem:tree}
	Trees belong to $\cC_{\calP}$.
\end{lemma}

\begin{proof}
\noindent From \cite{chevalierOddRootedOrientations1983}, we know that a tree admits a $T$-odd orientation for a subset $T \subseteq V(G)$ if and only if $|E(G)| + |T| \equiv_2 0$.
\end{proof}

\begin{remark}
	Because trees (and in particular paths) verify $|V(G)| + |E(G)|\equiv_2 1$, the condition $|E(G)| + |T| \equiv_2 0$  is equivalent to $|V(G)\setminus T| \equiv_2 1$, which is consequently also necessary and sufficient for the existence of an acyclic $T$-odd orientation.
\end{remark}

\begin{lemma}\label{lem:cycle}
	Cycles belong to $\cC_{\calP \calS}\setminus \cC_\calP$.
\end{lemma}

\begin{proof}
\noindent Let $C$ be a cycle.
Since $|V(C)|=|E(C)|$, then, by point \textbf{(c)} of Theorem \ref{thm:treilli}, $C\not\in \cC_\calP$.

\noindent Let $T$ be a subset of $V(C)$ satisfying $\calP$ and $\calS$.
By $\calS$, there exists a vertex $u\in V(C)\setminus T$.
Let $v$ be a neighbor of $u$ in $C$.
The graph $P=C-uv$ is a path with $d(u)=1$.
Let $T'=\{v\}\Delta T$, and note that $T'$ satisfies $\calP$ in $P$.
Moreover $P$ is a tree, then, by Lemma \ref{lem:tree}, $P$ admits an acyclic $T'$-odd orientation $O$.
Observe that  $d_{P_O}^-(u)=0$, since $u\in V(P)-T'$ and $d_P(u)=1$.
One can get an acyclic $T$-odd orientation of $C$ by completing $O$ with the arc $\dir{uv}$.
\end{proof}

\begin{remark}
	In fact, this proves that for any vertex $v \not \in T$, there exists an acyclic $T$-odd orientation of $C$ where $v$ is a source.
\end{remark}

\FloatBarrier
\subsection{Grids}

\noindent In this section, we provide a characterization of the instances $I = (G,T)$ of Problem \ref{prob:AOP}, where $G$ is a grid, for which $G$ admits an acyclic $T$-odd orientation.\newline

\noindent We prove Theorem \ref{thm:grid_carac} recalled here:

\noindent \textbf{Theorem \ref{thm:grid_carac}}:\newline
\textit{
Let $G = P_p \square P_q$ be a graph where $P_p$ and $P_q$ are paths with vertex set $\{u_0, \dots, u_{p-1}\}$ and $\{v_0, \dots, v_{q-1}\}$ respectively.
Let be a subset $T \subset V(P_p \square P_q)$ that satisfies $\calP \calS \overline\calS$.
\begin{center}
	$P_p \square P_q$ has no acyclic $T$-odd orientation if and only if $p,q \equiv_2 0$ and $V \setminus T$ satisfies that:
\end{center}
\begin{itemize}
	\item $(u_{2k-1}, v_j) \in V \setminus T$ if and only if $(u_{2k}, v_j) \in V \setminus T$ for all $k \in \{1, \dots, \frac{p-2}{2}\}$ and $j \in \{0, q-1\}$.
	\item $(u_i, v_{2k-1}) \in V \setminus T$ if and only if $(u_i, v_{2k}) \in V \setminus T$ for all $i \in \{0, p-1\}$ and $k \in \{1, \dots, \frac{q-2}{2}\}$ symmetrically.
\end{itemize}
}

\noindent \textbf{Preliminaries}~\newline

\noindent A \textbf{path} $H$ of length $h$ is a graph with vertex set $V(H) = \{w_0, \dots, w_{h-1}\}$ and edge set $E(H) = \{(w_i, w_{i+1}) \mid 0 \leq i < h-1\}$.
We define for any path $H$, the set $Pairs(H) = \{[w_{2k-1}, w_{2k}] \subset V(H) \mid k = 1, \dots, \frac{p-2}{2}\}$.
Let $H^{[\geq k]}$ and $H^{[\leq k]}$ be the graph induced by $\{w_k, \dots, w_{h-1}\}$ and $\{w_0, \dots, w_{k}\}$ respectively.
\newline

\noindent A \textbf{bad path instance} of Problem \ref{prob:AOP} is an instance $(H,T)$ where $H$ is a path of even length $h$, and $T \subseteq V(H)$ satisfies the following conditions: first, both endpoints of the path, $w_0$ and $w_{h-1}$, belong to $T$; and second, for any pair of vertices $[w,w'] \in Pairs(H)$, $w \in T$ if and only if $w' \in T$.
An example of this construction is illustrated in Figure \ref{fig:grid:bad_path_instance}.\newline

\begin{figure}[ht]
	\begin{center}
	\includegraphics[width=0.6\textwidth]{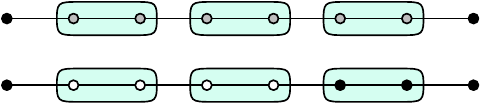}
	\caption{(Top) Generic bad path instance, along with an example (bottom)}
	\label{fig:grid:bad_path_instance}
	\end{center}
\end{figure}

\noindent A notable property of path instances is that, as long as it is not a bad path instance, it is always possible to find a subpath of odd length including one of the endpoints such that it has an even number of $T$'s vertices.

In particular we claim:
\begin{claim}\label{cla:grid:badPath}
	Let $(H,T)$ be an instance of Problem \ref{prob:AOP} such that $H$ is a path of even length $h$.\newline

	\noindent If $(H,T)$ is not a bad path instance of Problem \ref{prob:AOP}, then there exists $k \in \{0,\dots, \frac{h-2}{2}\}$ and \\ $L \in \{H^{[\leq 2k]}, H^{[\geq 2k+1]}\}$ such that $|T(L)|$ is even.
\end{claim}

\begin{proof}
	Let $(H,T)$ be an instance of Problem \ref{prob:AOP} where $H$ is a path.\newline
	Suppose $(H,T)$ is not a bad path instance.
If $w_0 \notin T$ or $w_{h-1} \notin T$, choose $L = H[w_0]$ or $L = H[w_{p-1}]$ respectively.
In either case, $|T(L)|$ is even.

	\noindent Otherwise, if $w_0, w_{p-1} \in T$, then since $(H,T)$ is not a bad path instance, there exists a pair of vertices $(w_{2k-1}, w_{2k}) \in Pairs(V(H))$ of smallest index $k \in \{1, \dots, \frac{h-2}{2}\}$ such that $w_{2k-1} \in T$ if and only if $w_{2k} \notin T$.
Then, $|T(H^{[\leq 2k]})|$ is even.
\end{proof}

\noindent Note in fact that the equivalence holds.
Indeed, for any bad path instance $(H,T)$, we have that for all $k \in \{0, \dots, \frac{h-2}{2}\}$, both $|T(H^{[\leq 2k]})|$ and $|T(H^{[\geq 2k+1]})|$ are odd.\newline

\noindent We now focus on \textbf{grids} $G = P_p \square P_q$, where $P_p$ and $P_q$ are paths of length $p$ and $q$, respectively, with vertex sets $V(P_p) = \{u_0, \dots, u_{p-1}\}$ and $V(P_q) = \{v_0, \dots, v_{q-1}\}$.
Define also $X_j = \{(u_i,v_j) \mid i = 1, \dots, p-1\}$ for all $j \in \{ 0, \dots, q-1\}$, and $Y_i = \{(u_i,v_j) \mid j = 1, \dots, q-1\}$ for all $i \in \{ 0, \dots, p-1\}$.
Additionally, we note $X_{\leq k} = \bigcup_{j=0}^k X_j$ and $X_{\geq k} = \bigcup_{j=k}^{q-1} X_j$ and similarly for $Y$.
A \textbf{grid instance} is an instance $(G, T)$ of Problem \ref{prob:AOP} where $G$ is a grid.
It is worth noticing that for any grid:
\begin{equation*}
	|V(G)| \equiv_2 p \times q, \quad |E(G)| \equiv_2 p + q \tag{G0}\label{decomposition:G0}\\
\end{equation*}

\noindent For any grid $G$, we note $Pairs(G) = Pairs(G[X_0]) \cup Pairs(G[X_{q-1}]) \cup Pairs(G[Y_0]) \cup Pairs(G[Y_{p-1}])$.
A \textbf{bad grid instance} of Problem \ref{prob:AOP} is an instance $(G,T)$ where $G = P_p \square P_q$ is a grid with $p \equiv_2 q \equiv_2 0$, and $T$ satisfies the following conditions:
\begin{enumerate}
	\item $(G[U], T(U))$ is a bad path instance of Problem \ref{prob:AOP} for all $U \in \{X_0 , X_{q-1}, Y_0, Y_{p-1}\}$ .
	\item For all $i \in \{1, \dots, p-2\}$ and $j \in \{1, \dots, q-2\}$ we have that $(u_i,v_j) \in T$.
\end{enumerate}
Note that condition \textbf{1.} is equivalent to say that for any pair $[(u,v),(u',v')] \in Pairs(G)$, we have that $(u,v) \in T$ if and only if $(u',v') \in T$.
So in essence, Theorem \ref{thm:grid_carac} states that for any grid instance $(G,T)$ with $T$ satisfying the parity condition $\calP$, $G$ admits an acyclic $T$-odd orientation if and only if $(G,T)$ is not a bad grid instance.
\newline

\begin{figure}[ht]
	\begin{center}
	\includegraphics[width=0.8\textwidth]{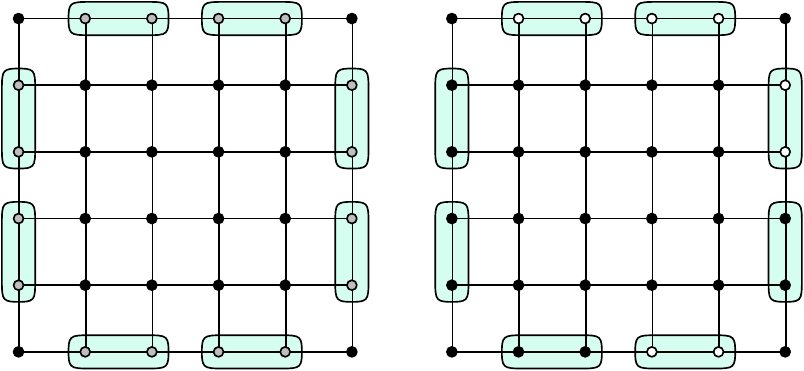}
	\caption{Generic bad grid instances along with an example}
	\label{fig:grid:bad_grid_instance}
	\end{center}
\end{figure}

\noindent It is also worth noticing the following:
\begin{claim}
	Let $(G,T)$ be a grid instance of Problem \ref{prob:AOP} such that $T$ satisfies the parity condition $\calP$.\newline
	\noindent If $(G,T)$ is not a bad grid instance of Problem \ref{prob:AOP}, then $T$ also satisfies $\calS$ and $\overline\calS$.
\end{claim}
\begin{proof}
	Let $(G,T)$ be a grid instance of Problem \ref{prob:AOP} such that $T$ satisfies the parity condition $\calP$.
	
	\noindent If $G$ has an odd dimension, then, by \textbf{(a)}, $|V(G)\setminus T| \equiv_2 |V(G)\setminus E(G)| \equiv_2 1$.
Furthermore, any grid with an odd dimension admits at least two vertices of odd degree.
Hence, $T$ does satisfy the conditions $\calS$ and $\overline\calS$.

	\noindent Suppose now that both $p$ and $q$ are even.
If $T$ does not satisfy condition $\calS$, then $T = V(G)$.
And if it does not satisfy condition $\overline\calS$, then $V(G) \setminus T = \{v \in V(G) \mid d(v) \equiv_2 1\}$.
In both cases, $(G,T)$ is a bad grid instance of Problem \ref{prob:AOP}.
\end{proof}

\noindent Thanks to this claim, when considering grid instances, we can focus on the parity condition $\calP$ and on checking if the instance is a bad grid instance or not.
Hence, we drop further reference to the conditions $\calS$ and $\overline\calS$ in this section.
\newline

\noindent In the next proofs, we use multiple $T$-decompositions of a grid instance $(G,T)$ that we shall describe carefully here.
We always assume that $T$ satisfies condition $\calP$.\newline

\noindent \textbf{(i,j)-corner T-decomposition} $\langle V_0,V_1,V_2 \rangle $ is defined as follows (see figure \ref{fig:grid:ij_corner_decomposition}):
\begin{align*}
	& V_0 = \{(u_a, v_b) \mid 0 \leq a < i, 0 \leq b < j\} = \bigcup_{k=0}^{i-1}Y_k^{[<j]} \neq \emptyset\\
	& V_1 = \{(u_a, v_b) \mid i \leq a < p, 0 \leq b < q\} = Y_{\geq i} \\
	& V_2 = \{(u_a, v_b) \mid 0 \leq a < i, j \leq b < q\} = \bigcup_{k=0}^{i-1}Y_k^{[\geq j]}
\end{align*}

\noindent  By Lemma \ref{lem:decomposition_P}, this $T$-decomposition satisfies $\calP$ if and only if: \hfill{(G1)}\label{decomposition:G1}
\begin{align*} 
	&|T_0| \equiv_2 |E(G[V_0])| \text{ i.e. } |T(V_0)| \equiv_2 |E(G[V_0])| \equiv_2 i + j\\
	&|T_1| \equiv_2 |E(G[V_1])|  \text{ i.e. }  |T(V_1)| \equiv_2 |E(G[V_1])| - |Z_1| \equiv_2 (p - i + q) - j
\end{align*}

\noindent \textbf{Inner-grid T-decompositions}: We define here two different $T$-decompositions $\langle V_{out},V_{in} \rangle$ and $\langle V_{in},V_{out} \rangle$ of a grid instance $(G,T)$ (see figure \ref{fig:grid:inner_grid_decomposition}), where $V_{out},V_{in}$ are as follows:
\begin{align*}
	& V_{out} = X_0 \cup Y_{0} \cup X_{q-1} \cup Y_{p-1}\\
    & V_{in} = V \setminus V_{out}
\end{align*}

\noindent By Lemma \ref{lem:decomposition_P}, $\langle V_{out},V_{in} \rangle$ satisfies $\calP$ if and only if	$|T_0| \equiv_2 |E(G[V_{out}])| \text{ i.e. } |T(V_{out})| \equiv_2 0$.\newline

\noindent And by Lemma \ref{lem:decomposition_P}, $\langle V_{in},V_{out} \rangle$ satisfies $\calP$ if and only if $|T_1| \equiv_2 |E(G[V_{out}])| \text{ i.e. } |T(V_{out}) \equiv_2 |E(G[V_{out}])| - |Z_1| \equiv_2 0$.\newline

\noindent Hence, for both inner-grid $T$-decompositions we only have to check whether $|T(V_{out})|\equiv_2 0$ to check if the decomposition satisfies $\calP$.
\hfill{(G2)}\label{decomposition:G2}

\begin{figure}[ht]
\centering
	\begin{subfigure}{0.49\textwidth}
		\centering
		\includegraphics[width=1\textwidth]{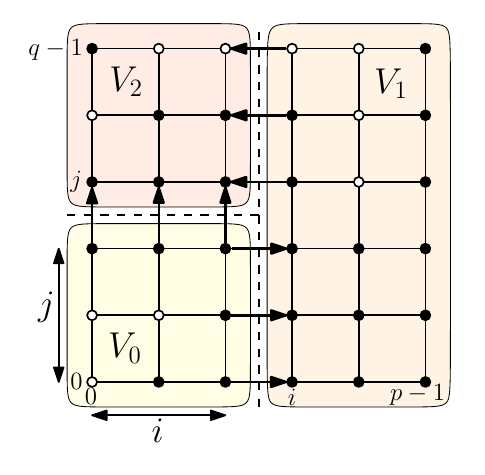}
		\caption{$(i,j)$-corner $T$-decomposition satisfying $\calP$}
		\label{fig:grid:ij_corner_decomposition}
	\end{subfigure}
	\begin{subfigure}{0.49\textwidth}
		\centering
		\includegraphics[width=1\textwidth]{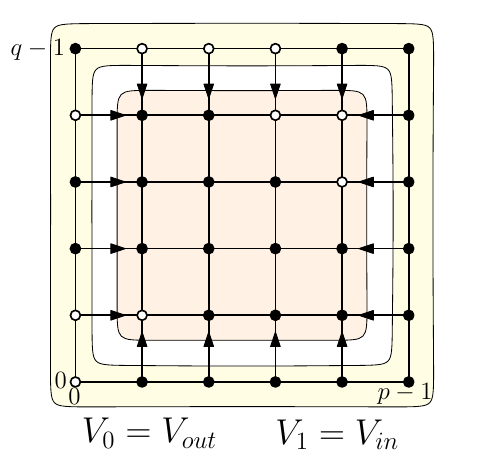}
		\caption{An inner-grid $T$-decomposition satisfying $\calP$}
		\label{fig:grid:inner_grid_decomposition}
	\end{subfigure}
	\caption{}
\end{figure}

\FloatBarrier

\noindent \textbf{Proof of Theorem \ref{thm:grid_carac}}~\newline

\noindent The proof is split in three lemmas, the first two (Lemma \ref{lem:grid:notBadPath} and Lemma \ref{lem:grid:notBadGrid}) proving the sufficient condition, and the last one (Lemma \ref{lem:grid:badGrid}) the necessary condition.

\begin{lemma}\label{lem:grid:notBadPath}
	Let $(G,T)$ be a grid instance of Problem \ref{prob:AOP} such that $T$ satisfies the parity condition $\calP$.\newline
	If there exists $U \in \{X_0 , X_{q-1}, Y_0, Y_{p-1}\}$ such that $(G[U], T(U))$ is not a bad path instance of Problem \ref{prob:AOP} then $G$ admits an acyclic $T$-odd orientation.
\end{lemma}
\begin{proof}
	Let $(G,T)$ be a grid instance of Problem \ref{prob:AOP} such that $T$ satisfies the parity condition $\calP$.

	\phantomsection
	\label{proof:notBadPath:case:a}
	\noindent \textbf{a.~} Suppose $G$ has an odd dimension, say $q \equiv_2 1$.
In this case, both $(G[Y_0],T(Y_0))$ and $(G[Y_{p-1}],T(Y_{p-1}))$ are not bad path instances.
We proceed by induction on $p$.
For $p = 1$, $G$ is a path and the result follows from Lemma \ref{lem:tree}.
For $p > 1$, since $T$ satisfies $\calP$ and $|E(G)| \equiv_2 p + 1$, we have $|T| \equiv_2 p + 1$.
Thus, $|V(G) \setminus T| = pq - (p+1)$ is odd, so there exists $k \in \{0, \dots, p-1\}$ such that $|Y_k\setminus T(Y_k)|$ is odd.
Let $k$ be the smallest such index and up to symmetry, we can suppose $k< p-1$.
We now consider the $(k+1, q)$-corner $T$-decomposition $\langle V_0, V_1, \emptyset \rangle = \langle Y_{< k+1}, Y_{\geq k+1}, \emptyset \rangle $ (see Figure \ref{fig:grid:not_bad_path_case1}).
We prove that this decomposition satisfies $\calP$.
Thanks to Lemma \ref{lem:decomposition_P} we only get to check wether $|T_0| = |T(V_0)| \equiv_2 |E(G[V_0])|$.
Recalling that $q \equiv_2 1$, we have:
	$$|T(V_0)| = \sum_{0\leq i < k+1}|T(Y_i)| \equiv_2 qk + q+1 \equiv_2 k+1 + q$$

	\noindent Moreover, each $V_i$ induces a grid with an odd dimension, so by induction, each $G[V_i]$ admits an acyclic $T_i$-odd orientation.
Therefore,  $\langle V_0, V_1, \emptyset \rangle $ is a good $T$-decomposition and $G$ admits an acyclic $T$-odd orientation by Theorem \ref{thm:T-odd decomposition}.\newline

	\noindent\textbf{b.~} Suppose both $p$ and $q$ are even, and there exists $U \in \{X_0 , X_{q-1}, Y_0, Y_{p-1}\}$ such that $(G[U], T(U))$ is not a bad path instance.
Without loss of generality, let $U = Y_0$.
Since $T$ satisfies $\calP$, $|T|$ is even.

	\phantomsection
	\label{proof:notBadPath:case:b1}
	\noindent\hspace*{2ex}\textbf{b.1.~} If $|T(Y_0)|$ is odd, consider the $(1, q)$-corner $T$-decomposition $\langle V_0, V_1, \emptyset \rangle = \langle Y_{0}, Y_{\geq 1}, \emptyset \rangle $ (see Figure \ref{fig:grid:not_bad_path_case2a}).\newline We prove that this decomposition satisfies $\calP$ and thanks to Lemma \ref{lem:decomposition_P} we only get to check wether $|T_0| = |T(V_0)| \equiv_2 |E(G[V_0])|$.
But by hypothesis:
	$$|T(V_0)| = |T(Y_0)| \equiv_2 1 + q$$

	\begin{adjustwidth}{2ex}{0pt}
	\phantomsection
	\label{proof:notBadPath:case:b2}
	\noindent\textbf{b.2.~} If $|T(Y_0)|$ is even, because $(Y_0, T(Y_0))$ is not a bad path instance of problem \ref{prob:AOP}, by Claim \ref{cla:grid:badPath}, there exists $k \in \{0, \dots, \frac{q-2}{2}\}$ such that $|T(Y_0^{[\leq 2k]})|$ is even.
In this case, consider the $(1, 2k+1)$-corner $T$-decomposition $\langle V_0, V_1, V_2 \rangle = \langle Y_0^{[< 2k+1]}, Y_{\geq 1}, Y_0^{[\geq 2k+1]} \rangle $ (see Figure \ref{fig:grid:not_bad_path_case2b}).
To prove that this decomposition satisfies $\calP$ we refer to \hyperref[decomposition:G1]{(G1)}:
	\begin{align*}
		|T(V_0)| &= |T(Y_0^{[< 2k+1]})| \equiv_2 1 + 2k+1\\
		|T(V_1)| &= |T(Y_{\geq 1})| \equiv_2 |T| - |T(Y_0)| \equiv_2 p+q - 0 \equiv_2 p - 1 + q - (2k + 1) &\text{(by \hyperref[decomposition:G0]{(G0)})}
	\end{align*}
	In both cases, each $V_i$ induces a grid with an odd dimension, so by Case 1, each $G[V_i]$ admits an acyclic $T_i$-odd orientation.
Thus, $\langle V_0, V_1, V_2 \rangle $ is a good $T$-decomposition and $G$ admits an acyclic $T$-odd orientation by Theorem \ref{thm:T-odd decomposition}.
	\end{adjustwidth}
\end{proof}

\begin{figure}[ht]
\centering
	\begin{subfigure}{0.3\textwidth}
		\includegraphics[width=1\textwidth]{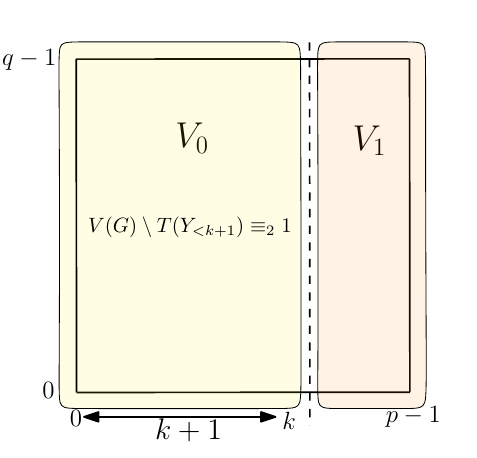}
		\caption{\hyperref[proof:notBadPath:case:a]{Case a}: $q$ odd}
		\label{fig:grid:not_bad_path_case1}
	\end{subfigure}
	\hspace*{0.3cm}
	\begin{subfigure}{0.3\textwidth}
		\includegraphics[width=1\textwidth]{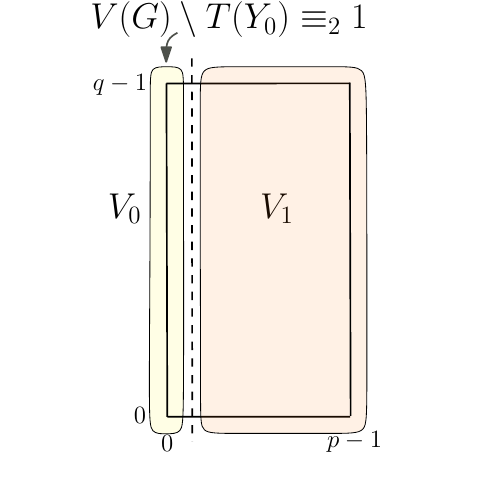}
		\caption{\hyperref[proof:notBadPath:case:b1]{Case b.1}: $p,q$ even, $(G[Y_0], T(Y_0))$ is not a bad path instance and $|T(Y_0)|$ is odd}
		\label{fig:grid:not_bad_path_case2a}
	\end{subfigure}
	\hspace*{0.3cm}
	\begin{subfigure}{0.3\textwidth}
			\includegraphics[width=1\textwidth]{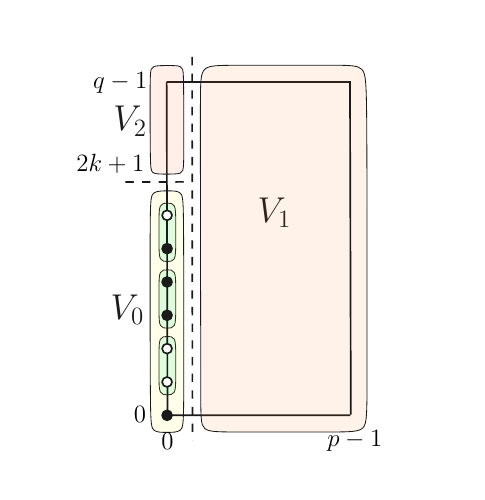}
			\caption{\hyperref[proof:notBadPath:case:b2]{Case b.2}: $p,q$ even, $(G[Y_0], T(Y_0))$ is not a bad path instance and $|T(Y_0)|$ is even}
			\label{fig:grid:not_bad_path_case2b}
	\end{subfigure}
	\caption{Illustration of the three decompositions used in the proof of Lemma \ref{lem:grid:notBadPath}.}
\end{figure}

\begin{lemma}\label{lem:grid:notBadGrid}
	Let $(G,T)$ be a grid instance of Problem \ref{prob:AOP} such that $T$ satisfies the parity condition $\calP$.\newline
	If $(G,T)$ is not a bad grid instance of Problem \ref{prob:AOP}, then $G$ admits an acyclic $T$-odd orientation.
\end{lemma}

\begin{proof}
	Let $(G,T)$ be a grid instance of Problem \ref{prob:AOP} such that $T$ satisfies the parity condition $\calP$.
	Lemma \ref{lem:grid:notBadPath} covers the cases where $G$ has an odd dimension, thus we assume $p,q \equiv_2 0$ and since $T$ satisfies $\calP$, we have that $|T| \equiv_2 0$.\newline

	\noindent We proceed by induction on the size of the grid.\newline
	\textbf{Init case:} $p=2$ and $q\geq 2$.
	If $(G,T)$ is not a bad grid instance, then there exists $i \in \{0,1\}$ such that $(G[Y_i], T(Y_i))$ is not a bad path instance.
	Thus, Lemma \ref{lem:grid:notBadPath} applies and $G$ admits an acyclic $T$-odd orientation.
	We take care of the case $p\geq 2$ and $q = 2$ by symmetry.\newline

	\noindent \textbf{Induction:} We suppose $p \geq 4$ and $q \geq 4$ and every non-bad grid instance $(G' = P_{p'} \square P_{q'}, T')$ of Problem \ref{prob:AOP} such that $p' < p$ or $q' < q$ and $T'$ satisfies $\calP$ in $G'$ admits an acyclic $T'$-odd orientation.

	\noindent From Lemma \ref{lem:grid:notBadPath}, if there exists $U \in \{X_0 , X_{q-1}, Y_0, Y_{p-1}\}$ such that $(G[U], T(U))$ is not a bad path instance of Problem \ref{prob:AOP}, then $G$ admits an acyclic $T$-odd orientation.
	So, in what follows, we can assume such a $U$ does not exist.
	Note that this implies $|T(U)| \equiv_2 0$ for all $U \in \{X_0 , X_{q-1}, Y_0, Y_{p-1}\}$.\newline

	\noindent Define $V_{out} = X_0 \cup Y_0 \cup X_{q-1} \cup Y_{p-1}$ and $V_{in} = V(G) \setminus V_{out}$.
	The rest of the proof is split into two parts in which we prove that $G$ admits an acyclic $T$-odd orientation.
	In the first part, we suppose $(G[V_{in}],T(V_{in}))$ is not a bad grid instance of Problem \ref{prob:AOP}.
	In the second part, we suppose $(G[V_{in}],T(V_{in}))$ is a bad grid instance of Problem \ref{prob:AOP} such that $T(V_{in}) \neq V_{in}$.
	The only remaining case not considered here corresponds to $T(V_{in}) = V_{in}$, which is precisely the case where $(G,T)$ is a bad grid instance.
	\newline

	\noindent \textbf{Part a.~} We suppose here that $(G[V_{in}],T(V_{in}))$ is not a bad grid instance of Problem \ref{prob:AOP}.
	Recall that, in a $T$-decomposition, $T_i = Z_i \triangle T(V_i)$ where $Z_i$ is the set of vertices of $V_i$ that have an odd number of neighbors in $V_0 \cup \dots \cup V_{i-1}$.\newline

	\noindent We distinguish two cases illustrated in Figures \ref{fig:grid:bad_grid_Part1_case1}, \ref{fig:grid:bad_grid_Part1_case2a} and \ref{fig:grid:bad_grid_Part1_case2b}. \newline

    \begin{figure}[ht]
	\centering
		\begin{subfigure}{0.3\textwidth}
			\includegraphics[width=1\textwidth]{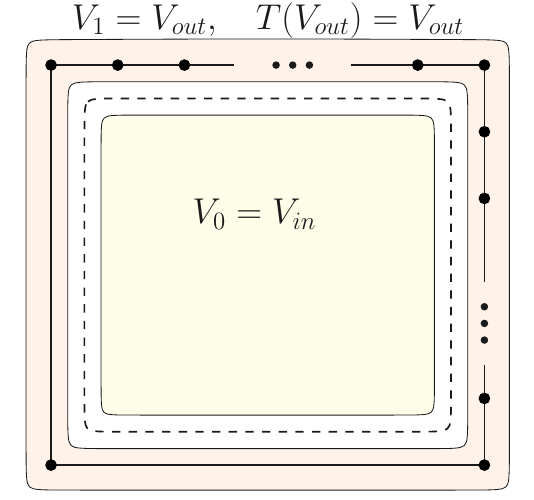}
			\caption{\hyperref[proof:notBadGrid:case:a1]{Case a.1}: $T(V_{out}) = V_{out}$}
			\label{fig:grid:bad_grid_Part1_case1}
		\end{subfigure}
		\begin{subfigure}{0.3\textwidth}
			\includegraphics[width=1\textwidth]{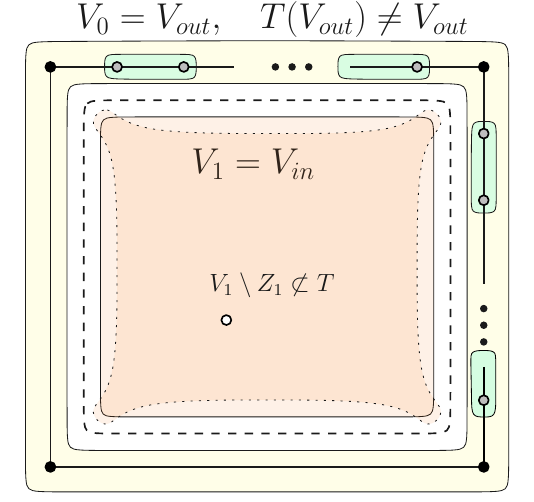}
			\caption{\hyperref[proof:notBadGrid:case:a2i]{Case a.2.i}: $T(V_{out}) \neq V_{out}$ and $V_1 \setminus Z_1 \not\subset T$}
			\label{fig:grid:bad_grid_Part1_case2a}
		\end{subfigure}
		\begin{subfigure}{0.3\textwidth}
			\includegraphics[width=1\textwidth]{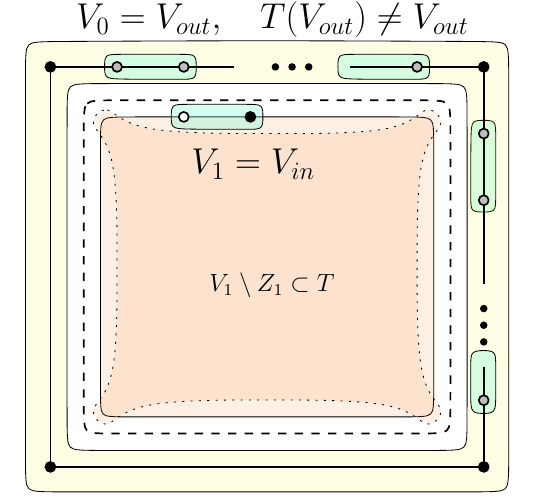}
			\caption{\hyperref[proof:notBadGrid:case:a2ii]{Case a.2.ii}: $T(V_{out}) \neq V_{out}$ and $V_1 \setminus Z_1 \subset T$}
			\label{fig:grid:bad_grid_Part1_case2b}
		\end{subfigure}
		\caption{Illustration of the three decompositions used in the proof of Lemma \ref{lem:grid:notBadGrid} Part a.}
    \end{figure}

	\phantomsection
	\label{proof:notBadGrid:case:a1}
	\noindent \textbf{a.1.~}If $T(V_{out}) = V_{out}$, we prove that the inner-grid $T$-decomposition $\langle V_0, V_1 \rangle  = \langle V_{in},V_{out} \rangle $ is good (see figure \ref{fig:grid:bad_grid_Part1_case1}) and apply Theorem \ref{thm:T-odd decomposition}. Since $|T(V_{out})| = |V_{out}| \equiv_2 0$, the decomposition satisfies $\calP$ by \hyperref[decomposition:G2]{(G2)}.
	By hypothesis $(G[V_0],T_0)$ is not a bad grid instance of Problem \ref{prob:AOP}, so $G[V_0]$ admits an acyclic $T_0$-odd orientation by induction.
	And $G[V_1]$ is a cycle with $T_1 = V_{out} \triangle Z_1 \neq  V_{out}$ since $Z_1 \neq \emptyset$.
	So $T_1$ also satisfies $\calS$ in $G[V_1]$ and $G[V_1$ admits an acyclic $T_1$-odd orientation by Lemma \ref{lem:cycle}.
	\newline

	\noindent \textbf{a.2.~}Otherwise, we prove that the inner-grid $T$-decomposition $\langle V_0, V_1 \rangle  = \langle V_{out},V_{in} \rangle $ is good (see figure \ref{fig:grid:bad_grid_Part1_case2a} and \ref{fig:grid:bad_grid_Part1_case2b}) and then apply Theorem \ref{thm:T-odd decomposition}.
	First, the decomposition satisfies $\calP$: indeed, since $(G[U], T(U))$ is a bad path instance of Problem \ref{prob:AOP} for all $U \in \{X_0 , X_{q-1}, Y_0, Y_{p-1}\}$, we have $|T(V_{out})| \equiv_2 0$ and \hyperref[decomposition:G2]{(G2)} ensures that $\calP$ is thus satisfied. 
	\noindent Moreover, $G[V_0]$ is a cycle with $T_0 = T(V_{out}) \neq V_{out}$, so $G[V_0]$ admits an acyclic $T_0$-odd orientation by Lemma \ref{lem:cycle}.
	It remains to check that $(G[V_1],T_1)$ is not a bad grid instance of Problem \ref{prob:AOP}, so that it admits an acyclic $T_1$-odd orientation by induction.
	Recall $V_1 = V_{in}$ and $T_1 = T(V_{in}) \triangle Z_1$.\newline
	\begin{adjustwidth}{2ex}{0pt}
		\phantomsection
	\label{proof:notBadGrid:case:a2i}
	\noindent \textbf{a.2.i}~If there exists a vertex $(u, v) \in V_1\setminus T$ that is not in $Z_1$, then $(u,v)$ is not in $T_1$, which ensures that $(G[V_1], T_1)$ is not a bad grid instance (see Figure \ref{fig:grid:bad_grid_Part1_case2a}).
	\end{adjustwidth}

	\begin{adjustwidth}{2ex}{0pt}
	\phantomsection
	\label{proof:notBadGrid:case:a2ii}
	\noindent \textbf{a.2.ii}~Else, since $(G[V_{in}], T(V_{in}))$ is not a bad grid instance, there must exist $[(u,v), (u',v')] \in Pairs(G[V_{in}])$ such that $(u,v) \in T(V_{in})$ if and only if $(u',v') \notin T(V_{in})$ (see Figure \ref{fig:grid:bad_grid_Part1_case2b}).
	But since both $(u,v)$ and $(u',v')$ belong to $Z_1$, they also satisfy $(u,v) \in T_1$ if and only if $(u',v') \notin T_1$.
	 Thus, $(G[V_1], T_1)$ is not a bad grid instance.\newline
	\end{adjustwidth}

	\noindent \textbf{Part b.~} Suppose now that $(G[V_{in}],T(V_{in}))$ is a bad grid instance of Problem \ref{prob:AOP} such that $T(V_{in}) \neq V_{in}$.
	Note that it implies that $p\geq 6$ or $q\geq 6$ and that it must exist $[x, x'] \in Pairs(G[V_{in}])$ such that $x, x' \not\in T(V_{in})$.
	Up to symmetry, we can assume that $q\geq 6$ and $[x, x'] = [(u_1,v_{2k}), (u_1,v_{2k+1})]$ for some $k \in \{1, \ldots, \frac{q-4}{2}\}$.\newline

	\noindent We distinguish three cases illustrated in Figures \ref{fig:grid:bad_grid_Part2_case1}, \ref{fig:grid:bad_grid_Part2_case2} and \ref{fig:grid:bad_grid_Part2_case3} (gray vertices in the same green set are either all black or all white). \newline

	\begin{figure}[ht]
		\centering
  		\begin{subfigure}[b]{0.3\textwidth}
				\includegraphics[width=1\textwidth]{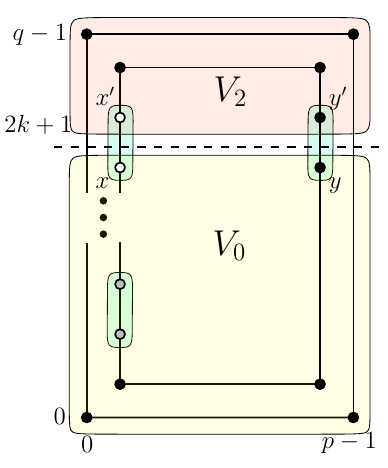}
				\caption{\hyperref[proof:notBadGrid:case:b1]{Case b.1}: $y,y' \in T$}
				\label{fig:grid:bad_grid_Part2_case1}
		\end{subfigure}
		\hfill
		\begin{subfigure}{0.3\textwidth}
			\includegraphics[width=1\textwidth]{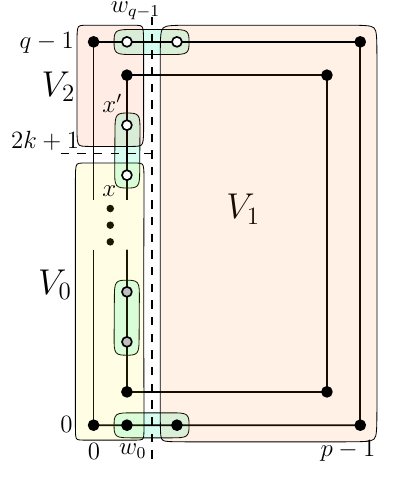}
			\caption{\hyperref[proof:notBadGrid:case:b2]{Case b.2}: $w_0 \in T$, $w_{q-1} \not \in T$}
			\label{fig:grid:bad_grid_Part2_case2}
		\end{subfigure}
		\hfill
		\begin{subfigure}{0.3\textwidth}
			\includegraphics[width=1\textwidth]{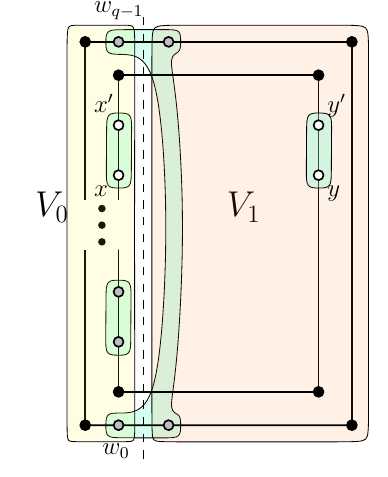}
			\caption{\hyperref[proof:notBadGrid:case:b3]{Case b.3}: $y,y' \not \in T$ and either both $w_0, w_{q-1} \in T$ or both $w_0, w_{q-1} \not \in T$}
			\label{fig:grid:bad_grid_Part2_case3}
		\end{subfigure}
		\caption{Illustration of the three decompositions used in the proof of Lemma \ref{lem:grid:notBadGrid} Part b.}
	\end{figure}

	\phantomsection
	\label{proof:notBadGrid:case:b1}
	\noindent \textbf{b.1.~}Assume that the pair $[y, y'] = [(u_{p-2},v_{2k}), (u_{p-2},v_{2k+1})]$ is such that $y,y' \in T$.\newline
		We prove that the $(p, 2k+1)$-corner $T$-decomposition $\langle V_0, \emptyset, V_2 \rangle = \langle X_{< 2k+1}, \emptyset, X_{\geq 2k+1} \rangle$ is good (see figure \ref{fig:grid:bad_grid_Part2_case1}) to apply Theorem \ref{thm:T-odd decomposition}.
		From the current hypothesis and \hyperref[decomposition:G1]{(G1)}, the decomposition satisfies $\calP$: indeed, $T(V_0) \equiv_2 1 \equiv_2 p + 2k+1$.

		\noindent Moreover, $V_0$ and $V_2$ induce grids with an odd dimension, so Lemma \ref{lem:grid:notBadPath} ensures that $G[V_0]$ and $G[V_2]$ admit an acyclic $T_0$-odd and $T_2$-odd orientation respectively.\newline

	\phantomsection
	\label{proof:notBadGrid:case:b2}
	\noindent \textbf{b.2.~} Assume that $w_0 = (u_1,v_0) \in T$ and $w_{q-1} = (u_1,v_{q-1}) \not \in T$.
\newline
		We now prove that the $(2, 2k+1)$-corner $T$-decomposition $\langle V_0, V_1, V_2 \rangle = \langle Y_0^{[< 2k+1]} \cup Y_1^{[< 2k+1]}, Y_{\geq 2}, Y_0^{[\geq 2k+1]} \cup Y_1^{[\geq 2k+1]} \rangle$ of $G$ is good (see Figure \ref{fig:grid:bad_grid_Part2_case2}).
		Similarly, our current hypothesis and \hyperref[decomposition:G1]{(G1)} certify that the decomposition satisfies $\calP$:
		\begin{align*}
			T(V_0) &\equiv_2 1 \equiv_2 2 + 2k+1 \\
			T(V_1) &\equiv_2 |T(Y_{\geq 2})| \equiv_2 1 \equiv_2 p - 2 + q - (2k+1)
		\end{align*}

		\noindent Notice also that $G[V_0]$ and $G[V_2]$ are grids with an odd dimension.
		Hence, by Lemma \ref{lem:grid:notBadPath} they admit an acyclic $T_0$-odd and $T_2$-odd orientation respectively.
		To show that the $T$-decomposition is good, it is now enough to check that $G[V_1]$ is not a bad grid instance of Problem \ref{prob:AOP} and apply induction.
		Since $(u_1,v_0) \in T$ and $[(u_1,v_0), (u_2,v_0)] \in Pairs(G[V_1]]$, we deduce that $(u_2,v_0) \in T$.
		Notice also that $(u_2,v_0) \in Z_2$.
		Hence, $(u_2,v_0) \not \in T_2$, which implies that $(G[V_1], T_1)$ is not a bad grid instance of Problem \ref{prob:AOP}.
		So by Theorem \ref{thm:T-odd decomposition}, $G$ admits an acyclic $T$-odd orientation.\newline

		\noindent The case where $(u_1,v_{q-1}) \in T$ and $(u_1,v_0)\not \in T$ is symmetrical.\newline

	\phantomsection
	\label{proof:notBadGrid:case:b3}
	\noindent \textbf{b.3.~} Finally, assume that $y,y'$ are not in $T$ and either both $w_0, w_{q-1} \in T$ or both $w_0, w_{q-1} \not \in T$. We prove that the $(2, q)$-corner $T$-decomposition $\langle V_0, V_1, \emptyset \rangle = \langle Y_{<2}, Y_{\geq 2}, \emptyset\rangle$ of $G$ is good (see figure \ref{fig:grid:bad_grid_Part2_case3}) and conclude by applying Theorem \ref{thm:T-odd decomposition}.
	To show that the decomposition satisfies $\calP$, we only need to check that $|T(V_0)| \equiv_2 |E(G[V_0])|$ thanks to Lemma \ref{lem:decomposition_P}:
	\begin{equation*}
	    T(V_0) \equiv_2 0 \equiv_2 2 + q\equiv_2 |E(G[V_0])| \quad \text{by \hyperref[decomposition:G0]{(G0)}}
	\end{equation*}

	\noindent Now, to show that the $T$-decomposition is good, it is enough to check that $(G[V_0], T_0)$ and $(G[V_1], T_1)$ are not bad grid instances of Problem \ref{prob:AOP} and apply induction:\newline
		\noindent \textbf{- }Let us show that  $(G[V_0], T_0)$ (with $V_0 = Y_{<2}$) is not a bad grid instance.\newline
		Since $(G[V_{in}],T(V_{in}))$ is a bad grid instance of Problem \ref{prob:AOP}, note that $(u_1, v_1) \in T$ and $[(u_1, v_{2j})$, $(u_2, v_{2j+1})] \in Pairs(G[V_{in}])$ for all $j \in \{1, \dots, \frac{q-4}{2}\}$.
		However, since $(u_1, v_{2k}) \not \in T$, there must exist $1 \leq k' \leq k$ such that $(u_1, v_{2k'-1}) \in T$ if and only if $(u_1, v_{2k'}) \not \in T$.
		Thus, $(G[Y_1], T(Y_1))$ is not a bad path instance and by Lemma \ref{lem:grid:notBadPath}, $G[V_0]$ admits an acyclic $T_0$-odd orientation.\newline

		\noindent \textbf{- }Let us show that  $(G[V_1], T_1)$ (with $V_1 = Y_{\geq 2}$) is not a bad grid instance.
		We distinguish two cases: \newline
		If $p = 4$, since $Z_1 = Y_2$, we deduce that $(u_2, v_1) \not \in T_1$ and $(u_{p-2},v_{2k-1}), (u_{p-2},v_{2k}) = (u_2, v_{2k}), (u_2, v_{2k+1}) \in T_1$.
		Moreover, since $(G[V_{in}],T(V_{in}))$ is a bad grid instance of Problem \ref{prob:AOP}, we have that $[(u_2, v_{2j})$, $(u_2, v_{2j+1})] \in Pairs(G[Y_{\geq 2}])$ for all $1 \leq j \leq \frac{q-4}{2}$.
		Hence, there must exist $1 \leq k' \leq k$ such that $(u_2, v_{2k'-1}) \in T_1$ if and only if $(u_2, v_{2k'}) \not \in T_1$.
		Thus, $(G[Y_2], T_1(Y_2))$ is not a bad path instance of Problem \ref{prob:AOP} and by Lemma \ref{lem:grid:notBadPath}, $G[V_1]$ admits an acyclic $T_1$-odd orientation.\newline

		\noindent If $p \geq 6$, note that $(u_{p-2},v_{2k-1}), (u_{p-2},v_{2k}) \not \in T_1$, so $(G[V_1], T_1)$ is not a bad grid instance of Problem \ref{prob:AOP} and by induction, $G[V_1]$ admits an acyclic $T_1$-odd orientation.

	\noindent Combining parts $1$ and $2$ concludes the induction.
	Thus, if $(G,T)$ is not a bad grid instance of Problem \ref{prob:AOP} such that $T$ satisfies the parity condition $\calP$, then it admits an acyclic $T$-odd orientation.\newline
\end{proof}

\begin{lemma}\label{lem:grid:badGrid}
	If $(G,T)$ is a bad grid instance of Problem \ref{prob:AOP}, then $G$ does not admit an acyclic $T$-odd orientation.
\end{lemma}

\begin{proof}
	Let $(G,T)$ be a bad grid instance of Problem \ref{prob:AOP} with $G = P_p \square P_q$.
	We show that $G$ does not admit any acyclic $T$-odd orientation.\newline

	\noindent The proof proceeds as follows.
	We construct a new instance $(\cG, \cT)$, where $|\cT| = |V(\cG)| - 1$, and we show that any acyclic $T$-odd orientation of $G$ would imply the existence of an acyclic $\cT$-odd orientation of $\cG$.
	\noindent Next, we recall the following theorem of Kiraly and Kisfaludi-Bak (Theorem 1.40 in \cite{kiralyDualCriticalGraphsNotes}).
	It provides a key characterization for acyclic $T$-odd orientations in planar graphs.
	To cite this theorem, we give some quick intuitive definitions; for formal ones, we refer to \cite{diestelGraphTheory2017}.
	A \textbf{planar embedding} is a drawing of a graph in the plane such that no edges cross each other.
	A graph that admits a planar embedding is called \textbf{planar}.
	The \textbf{dual} of a planar embedding of $G$ is the graph where the faces of this embedding are its vertex set, and two faces are adjacent whenever they share an edge of $G$.
	A graph is \textbf{factor-critical} if, after removing any vertex, the remaining graph has a perfect matching.
	\begin{quote}
		\textbf{(Theorem 1.40 in \cite{kiralyDualCriticalGraphsNotes})} The dual of a planar graph $G = (V,E)$ having an acyclic $V\setminus \{v\}$-odd orientation for some vertex $v\in V$ is always factor-critical.
		(So if there are multiple dual graphs depending on the planar embedding of $G$, then all of them are factor-critical.)
	\end{quote}
	\noindent From the contrapositive of this theorem, to show that $G$ does not admit any acyclic $T$-odd orientation, it is enough to show that the dual of some planar embedding of $\cG$ is not factor-critical.
	To do so, we show that some planar embedding $\phi(\cG)$ of $\cG$ contains $\lceil \frac{f}{2} \rceil$ pairwise independent faces, where $f$ is the number of faces in $\phi(\cG)$ and two faces are independent if they do not share an edge.
	Indeed, this implies that the dual of $\phi(\cG)$ contains an independent set of size $\lceil \frac{n}{2} \rceil$ and hence that it is not factor-critical.\newline

	\noindent \textbf{Construction of $(\cG, \cT)$:}\newline
	Define $\cG$ from $G$ by adding a vertex $s$ connected to every vertex in $V\setminus T$ as follows:
	\begin{itemize}
		\item $V(\cG) = V(G) \cup \{s_{root}\}$
		\item $E(\cG) = E(G) \cup \{ws_{root} \mid w\in V(G)\setminus T\}$
	\end{itemize}
	Set $\cT = V(\cG)\setminus \{s_{root}\}$.\newline

	\noindent \textbf{Correctness of $(\cG, \cT)$:}\newline
	Suppose that $G$ admits an acyclic $T$-odd orientation $O$.
	We construct an orientation $O'$ of $\cG$ as follows:
	\begin{itemize}
		\item $a \in O'$ for all arcs $a \in O$.
		\item For each $w \in V(G) \setminus T$, $\dir{s_{root}w} \in O'$.
	\end{itemize}

	\noindent By construction, $O'$ is a $\cT$-odd orientation of $\cG$.
	Furthermore, $O'$ remains acyclic $s_{root}$ is a source in $\cG_{O'}$.
	Hence, any acyclic $T$-odd orientation of $G$ implies an acyclic $\cT$-odd orientation of $\cG$.\newline

	\noindent \textbf{Planar Embedding $\phi(\cG)$:}\newline
	We now describe a planar embedding $\phi(\cG)$ of $\cG$.
	We first map the vertex $(u_i, v_j)\in V(G)$ to the point $(i, j) \in \mathbb Z^2$ for all $i \in \{0, \dots, p-1\}$ and $j \in \{0, \dots, q-1\}$.
	Each edge is drawn as a straight line segment between its endpoints.
	We say that a vertex $(v_i, v_j)$ is on the border of the grid if $i \in \{0, p-1\}$ or $j \in \{0, q-1\}$.
	\noindent The additional vertex $s_{root}$ is adjacent only to vertices of  $V(G) \setminus T$. Because $(G,T)$ is a bad grid instance, all those vertices lie on the border of the grid.
	Hence, we embed $s_{root}$ in the external face of the grid, at the point $(p, q)$, and for each vertex $w \in V(G) \setminus T$, we draw a curve from $s_{root}$ to $w$. These curves can be chosen so that they do not cross each other, nor any edge of $G$, except at their endpoints (see Figure \ref{fig:grid:bad_grid_embedding}).\newline

	\noindent \textbf{Independent faces of $\phi(\cG)$:}\newline
	First, note that $|V(\cG)| = pq + 1$ and $|E(\cG)| = q(p-1) + p(q-1) + |V(G)\setminus T|$, so by Euler's formula, $\phi(\cG)$ contains $f = (p-1)(q-1) + |V(G)\setminus T|$ faces.
	\newline

	\noindent We select the following faces of $\phi(\cG)$ defined by the cycles:
	\begin{enumerate}
		\item $(u_{2i}, v_{2j})(u_{2i}, v_{2j+1})(u_{2i+1}, v_{2j+1})(u_{2i+1}, v_{2j})$ for all $i \in \{0, \dots, \frac{p-2}{2}\}$ and $j \in \{0, \dots, \frac{q-2}{2}\}$
		\item $(u,v)(u',v')s_{root}$ for all $[(u,v), (u',v')] \in \text{Pairs}(G)$ such that $(u,v),(u',v') \in V(G)\setminus T$.
	\end{enumerate}

	\noindent Note that all selected faces of $\phi(\cG)$ are independent.
	Also, we selected precisely $\lceil \frac{(p-1)(q-1)}{2} \rceil$ faces in \textbf{1.} and $\frac{|V(G)\setminus T|}{2}$ faces in \textbf{2.}.
	This makes a total of $\lceil \frac{(p-1)(q-1)}{2} \rceil + \frac{|V(G)\setminus T|}{2} = \lceil \frac{f}{2} \rceil$ independent faces (see Figure \ref{fig:grid:bad_grid_embedding}).\newline

	\begin{figure}[ht]
		\centering
		\includegraphics[width=0.5\textwidth]{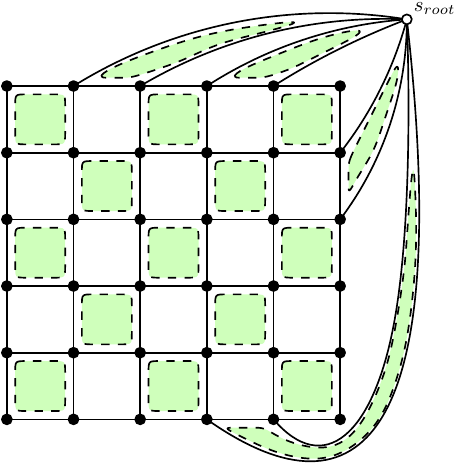}
		\caption{\centering Chosen faces on the planar embedding of $\cG$
		associated with the bad grid instance of Figure \ref{fig:grid:bad_grid_instance}.}
		\label{fig:grid:bad_grid_embedding}
	\end{figure}

\end{proof}

\noindent Combining Lemma \ref{lem:grid:notBadGrid} and Lemma \ref{lem:grid:badGrid}, we obtain the following theorem:

\begin{theorem}\label{thm:grid:carac}
	Let $(G,T)$ be a grid instance of Problem \ref{prob:AOP} such that $T$ satisfies the parity condition $\calP$.\newline
	$G$ admits an acyclic $T$-odd orientation if and only if $(G,T)$ is not a bad grid instance of Problem \ref{prob:AOP}.
	\qed
\end{theorem}

\subsection{Cylinders and Tori} \label{sec:cyclinder_torus}

\noindent In this section, we prove Theorem \ref{thm:tore_carac} which solves Problem \ref{prob:AOP} for graphs $G\square H$ when $G$ and $H$ are a path or a cycle.
Let us recall this second result on cylinders and tori:

\begin{theorem*}[\ref{thm:tore_carac}]
	Let $G = C_p \square H_q$ be a graph where $C_p$ is a cycle on $p$ vertices and $H_q$ is either a path or a cycle on $q$ vertices.
	Let $T \subseteq V(G)$ be a subset that satisfies $\calP \calS \overline\calS$.
	\begin{center}
		If $H_q$ is a path, or if $p,q \geq 4$, then $G$ has an acyclic $T$-odd orientation.\newline
	\end{center}
\end{theorem*}

\noindent As for grids, the proof of Theorem \ref{thm:tore_carac} works by construction and induction.
We decompose the proof into three lemmas (Lemmas \ref{lem:cyl:odd}, \ref{lem:cyl:even} on cylinders and \ref{lem:Torus} on tori) from which this theorem is a direct consequence.
We now give some preliminary results and definitions.\newline

\noindent\textbf{Preliminaries}~\newline

\noindent Let $p\geq 3$ and $q\geq 1$.
Let $G$ be $C_p\square H_q$ with $H_q\in \{P_q, C_q\}$.
If $H_q=C_q$, then $q\geq 3$.

\noindent We denote the vertex sets by $V(C_p)=\{u_0, \dots, u_{p-1}\}$ and $V(H_q)=\{v_0, \dots, v_{q-1}\}$.
We fix the edge sets as follows: $E(C_p)=\{u_iu_{i+1\bmod p} \mid i=0, \dots, p-1\}$ and $E(H_q)=\{v_jv_{j+1} \mid j=0, \dots, q-2\}$ if $H_q$ is a path; else $E(H_q)=\{v_jv_{j+1\bmod q} \mid j=0, \dots, q-1\}$.

\noindent For every $j\in \{0, \dots, q-1\}$, let $X_j^{[\geq i]}=\{(u_k, v_j) \mid k=i, \dots, p-1\}$ and $X_j^{[\leq i]}=\{(u_k, v_j) \mid k=0, \dots, i\}$.
We drop the superscript $[\geq i]$ and $[\leq i]$ when $i=0$ or $p-1$ respectively, and in those cases, the graph induced by $X_j$ is $C_p$; we call $X_j$ a \textbf{row}.
By extension, we denote $X_{\leq j}=\cup_{k=0}^j X_k, X_{< j}=X_{\leq j}\setminus X_j$ and $X_{\geq j}=\cup_{k=j}^q X_k, X_{> j}=X_{\geq j}\setminus X_j$.
In particular, the graph induced by $X_{\leq j}$ is a cylinder.
\newline
\noindent For every $i\in \{0, \dots, p-1\}$, let $Y_i^{[\geq j]} = \{(u_i, v_k) \mid j\leq k\leq q-1\}$ and $Y_i^{[\leq j]} = \{(u_i, v_k) \mid 0\leq k\leq j\}$.
We drop the superscript $[\geq j]$ and $[\leq j]$ when $j=0$ or $q-1$ respectively, and in those cases, the graph induced by $Y_i$ is $H_q$; we call $Y_i$ a \textbf{column}.
We also denote similarly $Y_{\leq i}=\cup_{k=0}^i Y_k, Y_{< i}=Y_{\leq i}\setminus Y_i$ and $Y_{\geq i}=\cup_{k=i}^p Y_k, Y_{> i}=Y_{\geq i}\setminus Y_i$.
In particular, the graph induced by $Y_{\leq i}$ is a grid whenever $H_q=P_q$ but a cylinder whenever $H_q=C_q$.\newline

\noindent \textbf{Cylinders}~\newline

\noindent This part is dedicated to the proof of Theorem \ref{thm:tore_carac} for cylinders.
For any cylinder $G = C_p\square P_q$, it is useful to remind that:
\begin{equation*}
	|V(G)| \equiv_2 p \times q, \quad |E(G)| \equiv_2 p \tag{C0}\label{decomposition:C0}
\end{equation*}

\noindent Let be a subset $T\subseteq V(G)$. Recall that $Source(T) = V(G)\setminus T$ and $Sink(T)$ is the set of vertices in $T$ that have an odd degree and vertices in $V(G)\setminus T$ that have an even degree. Hence, in cylinders, $Source(T) \neq \emptyset$ if and only if $T \neq V(G)$ and $Sink(T) \neq \emptyset$ if and only if $T \neq V(G)\setminus \{X_0,X_{q-1}\}$. Moreover there is no need to check whether $Source(T) \neq Sink(T)$ as it is always the case for non Eulerian graphs. \newline

\noindent Note that $T$ does not satisfy $\calS$ if and only if $T = V(G)$ and does not satisfy $\overline\calS$ if and only if $T = V(G)\setminus \{X_0,X_{q-1}\}$ (see Figure \ref{fig:bad_cylindre}). In the proof, to show that a cylinder instance $(G,T)$ satisfies $\calS$ or $\overline\calS$ we show that $Source(T)$ and $Sink(T)$ are non empty. 

\begin{figure}[ht]
	\centering
	\includegraphics[width = 0.5 \textwidth, keepaspectratio]{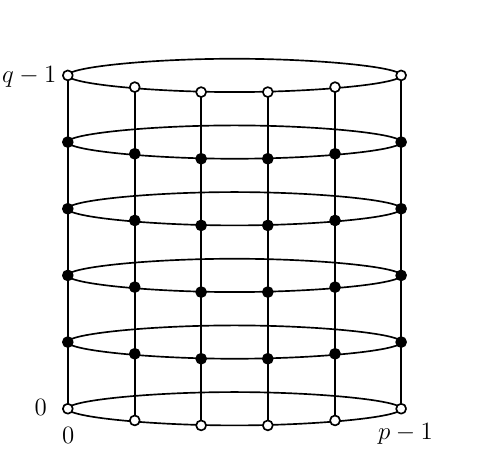}
	\caption{A cylinder with $T$ satisfying $\calP$ and $\calS$, but violating $\overline\calS$}
	\label{fig:bad_cylindre}
\end{figure}

\noindent In the next proofs, we use multiple $T$-decompositions of a cylinder instance $(G,T)$ that we shall describe carefully here.
Using Lemma \ref{lem:decomposition_P}, we give the necessary and sufficient conditions we use to prove that these $T$-decompositions satisfy the parity condition $\calP$.\newline

\noindent \textbf{first-row T-decompositions.}  We define here two different $T$-decompositions $\langle V_0, V_1 \rangle  = \langle X_{0},X_{\geq 1} \rangle $ or $\langle X_{\geq 1}, X_{0} \rangle $ of $(G,T)$ (see Figure \ref{fig:first_row_decomposition}).
By Lemma \ref{lem:decomposition_P}, these two $T$-decompositions satisfy $\calP$ if and only if:
\begin{align*}
	& |T(V_0)| \equiv_2 |E(X_{0})| \equiv_2 |E(X_{\geq 1})| \equiv_2 p \tag{C1}\label{decomposition:C1}
\end{align*}

\noindent \textbf{i\textsuperscript{th}-column T-decomposition.} It is defined as $\langle V_0,V_1 \rangle = \langle Y_i, V(G) \setminus Y_i \rangle$ for $i \in \{0 \dots p-1\}$ (see Figure \ref{fig:i_column_decomposition}):
$$ V_0 = Y_i, \quad V_1 = V(G) \setminus Y_i$$

\noindent  By Lemma \ref{lem:decomposition_P}, this $T$-decomposition satisfies $\calP$ if and only if: 
\begin{align*}
	& |T_0| \equiv_2 |E(G[V_0])| \text{ i.e. } |T(V_0)| \equiv_2 |E(Y_i)| \equiv_2 q-1 \tag{C2}\label{decomposition:C2}
\end{align*}

\noindent \textbf{i\textsuperscript{th}-column \& first-row T-decompositions.} We define here two different $T$-decompositions (see Figure \ref{fig:i_column_first_row_decomposition}).\newline
\noindent The first one is $\langle V_0,V_1, V_2 \rangle  = \langle Y_i^{[\geq 1]}, X_0, V(G)\setminus (Y_i^{[\geq 1]} \cup X_0) \rangle$.
By Lemma \ref{lem:decomposition_P}, it satisfies $\calP$ if and only if:
\begin{align*}
	&|T_0| \equiv_2 |E(G[V_0])| \text{ i.e. } |T(V_0)| \equiv_2 |E(Y_i^{[\geq 1]})| \equiv_2 q \tag{C3a}\label{decomposition:C3a}\\
	&|T_1| \equiv_2 |E(G[V_1])|  \text{ i.e. }  |T(V_1)| \equiv_2 |E(X_0)| - |Z_1| \equiv_2 p-1
\end{align*}

\noindent The second one is $\langle V_0,V_1, V_2 \rangle  = \langle V\setminus (Y_i^{[\geq 1]} \cup X_0), X_0, Y_i^{[\geq 1]}\rangle$. By Lemma \ref{lem:decomposition_P}, it satisfies $\calP$ if and only if:
\begin{align*}
	&|T_1| \equiv_2 |E(G[V_1])|  \text{ i.e. }  |T(V_1)| \equiv_2 |E(X_0)| - |Z_1| \equiv_2 p - (p-1) \equiv_2 1 \tag{C3b}\label{decomposition:C3b}\\
	&|T_2| \equiv_2 |E(G[V_2])|  \text{ i.e. }  |T(V_2)| \equiv_2 |E(Y_i^{[\geq 1]})| - |Z_2| \equiv_2 q-1
\end{align*}

\begin{figure}[htbp]
  \centering

  \begin{subfigure}[b]{0.3\textwidth}
    \includegraphics[width=\textwidth]{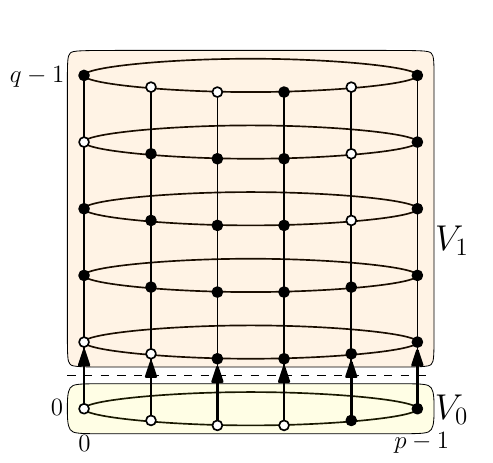}
    \caption{A first-row $T$-decomposition satisfying $\calP$}
    \label{fig:first_row_decomposition}
  \end{subfigure}
  \hfill
  \begin{subfigure}[b]{0.3\textwidth}
    \includegraphics[width=\textwidth]{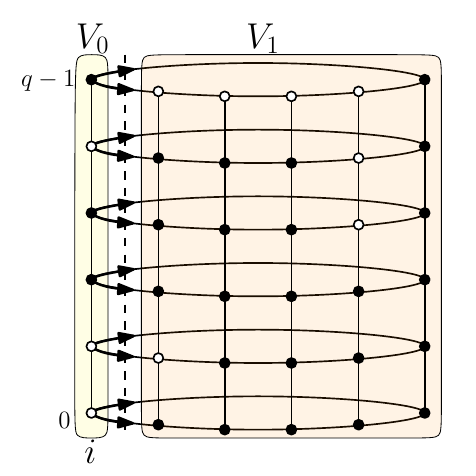}
    \caption{An i-column $T$-decomposition satisfying $\calP$}
	\label{fig:i_column_decomposition}
  \end{subfigure}
  \hfill
  \begin{subfigure}[b]{0.3\textwidth}
    \includegraphics[width=\textwidth]{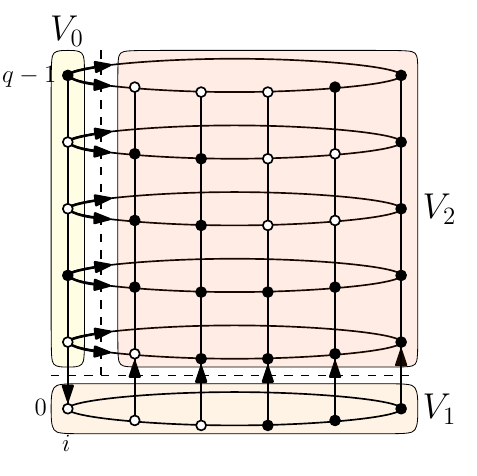}
	\caption{An i-column \& first-row $T$-decomposition satisfying $\calP$}
	\label{fig:i_column_first_row_decomposition}
  \end{subfigure}
  \caption{$T$-decompositions used in the proofs of Lemmas \ref{lem:cyl:even} and \ref{lem:cyl:odd}.}
\end{figure}

\begin{lemma}\label{lem:cyl:odd}
Let $G$ be a cylinder $C_p\square P_q$ with $p\geq 3$ odd and $q\geq 1$.
The graph $G$ admits an acyclic $T$-odd orientation for any subset $T\subseteq V(G)$ satisfying:
\begin{itemize}
	\item $\calP$ if $q$ is even.
	\item $\calP \calS \overline\calS$ if $q$ is odd.
\end{itemize}
\end{lemma}

\begin{proof}
	Let $G=C_p\square P_q$ with $p\geq 3$ odd and $q > 1$ (for $q = 1$, refer to Lemma \ref{lem:cycle}).\newline
	\phantomsection
	\label{proof:cyl:odd:case:a}
	\textbf{a.~}Suppose that $q\equiv_2 0$, so that $|V(G)|\equiv_2 0$. Let $T$ be a subset of $V(G)$ satisfying $\calP$ and recall we suppose $p$ odd. By \hyperref[decomposition:C0]{(C0)}, this implies $|T| \equiv_2 1$ and $|V(G)\setminus T| \equiv_2  1$. This implies that $T$ also satisfies $\calS$ and $\overline \calS$ ($T \neq V(G)$ and $T \neq V(G) \setminus \{X_0,X_{q-1}\}$ since $|V(G) \setminus \{X_0,X_{q-1}\} \equiv_2 0$).
	By \hyperref[decomposition:C0]{(C0)}, we have that $|V(G)| + |E(G)| \equiv_2 1$ so $T$ must verify whether $|V(G)\setminus T| \equiv_2 1$ and hence satisfies $\calS$ and $\overline\calS$ as well.
	Also, since $|V(G)\setminus T| \equiv_2 1$, there must exist
	$i \in \{1,\dots, p\}$ such that $|Y_i\setminus T(Y_i)| \equiv_2 1$.
	Let us consider the $i$\textsuperscript{th}-column $T$-decomposition $\cA_1 = \langle V_0, V_1 \rangle  = \langle Y_i, V(G)\setminus Y_i \rangle$ (see Figure \ref{fig:odd_case_a}).
	It satisfies $\calP$ (see \hyperref[decomposition:C2]{(C2)}):
	$$|T(V_0)| \equiv_2 q-1 \equiv_2 |E(G[V_0])|$$

	\noindent Now, $G[V_0]$ is a path and has an acyclic $T_0$-odd orientation by Lemma \ref{lem:tree}.
	$G[V_1]$ is a grid with even dimensions, hence, if $(G[V_1],T_1)$ is not a bad grid instance, it has an acyclic $T_1$-odd orientation by Lemma \ref{lem:grid:notBadGrid} and $\cA_1$ is good.\newline
	\noindent Otherwise, if $(G[V_1],T_{1})$ is a bad grid instance in $\cA_1$, we have in particular that $(u_{i+1},v_0)$ and $(u_{i+1},v_{q-1})$ are not in $T$.
	In this case, we define $\cA_2 = \langle V_0, V_1 \rangle  = \langle V(G)\setminus Y_i, Y_i  \rangle$.
	To show it satisfies $\calP$ it is enough by Lemma \ref{lem:decomposition_P} to check that $|T_1| \equiv_2 |E(G[V_1])|$ or equivalently that $|T(V_1)| \equiv_2 |E(G[V_1])| - |Z_1|$.
	But $|Z_1| \equiv_2 0$ hence:
	$$|T(V_1)| \equiv_2 |T(Y_i)| \equiv_2 q-1 \equiv_2 |E(G[V_1])| - |Z_1|$$

	\noindent Now, the subgraph induced by $V_0$ in $\cA_2$ is a grid with even dimensions but in this case $(G[V_0],T_0)$ is guaranteed to not be a bad grid instance in $\cA_2$ as its corner vertices $(u_{i+1},v_0)$ and $(u_{i+1},v_{q-1})$ are not in $T_0$.
	Hence, $G[V_0]$ has an acyclic $T_0$-odd orientation by Lemma \ref{lem:grid:notBadGrid}.
	Furthermore, $G[V_1]$ is a path and has an acyclic $T_1$-odd orientation by Lemma \ref{lem:tree}, achieving to prove that $\cA_2$ is good.\newline

	\noindent \textbf{b.~}From now on, suppose that $q \equiv_2 1$ and let $T$ be a subset of $V(G)$ satisfying $\calP \calS \overline\calS$.
	\begin{adjustwidth}{2ex}{0pt}
	\phantomsection
	\label{proof:cyl:odd:case:b1}
	\noindent \textbf{b.1.~} Suppose that there is a column $Y_i$ such that $T(Y_i)$ satisfies $\calP$ in $G[Y_i]$ for some $i\in \{0,\dots, p-1\}$.
	Then the $i$\textsuperscript{th}-column $T$-decomposition $\cB = \langle V_0, V_1 \rangle  = \langle Y_i, V(G)\setminus Y_i \rangle$ satisfies $\calP$ (refer to \hyperref[decomposition:C2]{(C2)} and Figure \ref{fig:odd_case_b1}):
	$$|T(V_0)| \equiv_2 |T(Y_i)| \equiv_2 q-1$$
	We now show it is a good decomposition.
	The subgraph $G[V_0]$ is a path and has an acyclic $T_0$-odd orientation by Lemma \ref{lem:tree}.
	The subgraph $G[V_1]$ is a grid with an odd dimension, hence it has an acyclic $T_1$-odd orientation by Lemma \ref{lem:grid:notBadPath}.\newline

	\noindent \textbf{b.2.~} Assume that for all $i \in \{0,\dots, p-1\}$, $T(Y_i)$ does not satisfy $\calP$ in $G[Y_i]$, i.e. $|T(Y_i)| \equiv_2 1$.
	\begin{adjustwidth}{2ex}{0pt}
	\phantomsection
	\label{proof:cyl:odd:case:b2i}
	\noindent \textbf{b.2.i.~} Suppose that $T(X_0) \neq X_0$ and $T(X_0) \neq \emptyset$.
	If $T(X_0)$ satisfies $\calP$ in $G[X_0]$, we consider the first-row $T$-decomposition $\cC_1 = \langle V_0, V_1 \rangle  = \langle X_0, X_{\geq 1} \rangle$; otherwise if $T(X_0)$ does not satisfy $\calP$ in $G[X_0]$, then $T(X_{\geq 1})$ satisfies $\calP$ in $G[X_{\geq 1}]$.
	Indeed :
	\begin{align*}
		|T| &= |T(X_0)| + |T(X_{\geq 1})|\\
		&\equiv_2 |E(G[X_0])| + |E(G[X_{\geq 1}])| + |\delta(X_0, X_{\geq 1})|\\
		&\equiv_2 |E(G[X_0])| + |E(G[X_{\geq 1}])| + 1
	\end{align*} 
	where $\delta(X_0, X_{\geq 1})$ denotes the set of edges with one endpoint in $X_0$ and the other in $X_{\geq 1}$.
	
	\noindent In this case, we consider the first-row $T$-decomposition $\cC_2 = \langle V_0, V_1 \rangle  = \langle X_{\geq 1}, X_0 \rangle$ (see Figure \ref{fig:odd_case_b2i}) and either $\cC_1$ or $\cC_2$ satisfies $\calP$ by Lemma \ref{lem:decomposition_P}.

	\noindent For $i \in \{0,1\}$, if the subgraph $G[V_i]$ is either a cylinder with $p$ odd and $q$ even or a cycle (with $T_i$ verifying $\calS$ since $T_i \neq V_i$), then it has an acyclic $T_i$-odd orientation by case \textbf{a.} or Lemma \ref{lem:cycle}, respectively.
	Hence, either $\cC_1$ or $\cC_2$ is good.
	We take care of the case where $T(X_{q-1}) \neq X_{q-1}$ and $T(X_{q-1}) \neq \emptyset$ by symmetry.\newline

	\phantomsection
	\label{proof:cyl:odd:case:b2ii}
	\noindent \textbf{b.2.ii.~} Suppose that $T(X_0) = X_0$ and $T(X_{q-1}) = X_{q-1}$.
	We consider an $i$\textsuperscript{th}-column \& first-row $T$-decomposition $\cD_1 = \langle V_0, V_1, V_2 \rangle = \langle V\setminus (Y_i^{[\geq 1]} \cup X_0), X_0, Y_i^{[\geq 1]} \rangle$ for some $i\in \{0,\dots, p-1\}$ such that $T(V\setminus (Y_i^{[\geq 1]} \cup X_0)) \neq V\setminus (Y_i^{[\geq 1]} \cup X_0)$ (see Figure \ref{fig:odd_case_b2ii}).
	This decomposition satisfies $\calP$ (refer to \hyperref[decomposition:C3b]{(C3b)}):
	$$|T(V_1)| \equiv_2 |T(X_0)| \equiv_2 p, \quad |T(V_2)| \equiv_2 |T(Y_i^{[\geq 1]})| \equiv_2 0 \equiv_2 q-1$$
	The subgraph $G[V_0]$ is a grid with even dimensions.
	If $(G[V_0],T_0)$ is a bad grid instance of Problem \ref{prob:AOP}, we consider symmetrically the $i$\textsuperscript{th}-column \& first-row $T$-decomposition $\cD_2 = \langle V_0, V_1, V_2 \rangle = \langle V\setminus (Y_i^{[\leq q-2]} \cup X_{q-2}), X_{q-1}, Y_i^{[\leq q-2]}\rangle$ (see Figure \ref{fig:odd_case_b2ii}) which satisfies $\calP$ as well.
	Now, $(G[V_0],T_0)$ cannot be a bad grid instance in both $\cD_1$ and $\cD_2$.
	Indeed, this would imply that $T(V\setminus Y_i) = V\setminus Y_i$, which is not possible because  $V\setminus Y_i \supset V\setminus (Y_i^{[\geq 1]} \cup X_0)$ and we supposed $T(V\setminus (Y_i^{[\geq 1]} \cup X_0)) \neq V\setminus (Y_i^{[\geq 1]} \cup X_0)$.
	By Lemma \ref{lem:grid:notBadGrid}, $G[V_0]$ has an acyclic $T_0$-odd orientation.
	The subgraph $G[V_1]$ is a cycle with $T_1$ satisfying $\calP$ and $\calS$ since $T_1 \neq V_1$, so it has an acyclic $T_1$-odd orientation by Lemma \ref{lem:cycle}.
	The subgraph $G[V_2]$ is a path and has an acyclic $T_2$-odd orientation by Lemma \ref{lem:tree}.
	We deduce that at least one decomposition among $\cD_1$ or $\cD_2$ is good.\newline

	\phantomsection
	\label{proof:cyl:odd:case:b2iii}
	\noindent \textbf{b.2.iii.~} Suppose that $T(X_0) = \emptyset$ and $T(X_{q-1}) = X_{q-1}$.
	We consider the $i$\textsuperscript{th}-column \& first-row $T$-decomposition $\cE = \langle V_0, V_1, V_2 \rangle  = \langle Y_i^{[\geq 1]}, X_0, V\setminus (Y_i^{[\geq 1]} \cup X_0) \rangle$ for some $i\in \{0, \dots, p-1\}$ (see Figure \ref{fig:odd_case_b2iii}).
	This decomposition satisfies $\calP$ (refer to \hyperref[decomposition:C3a]{(C3a)}):
	$$|T(V_0)| \equiv_2 |T(Y_i^{[\geq 1]})| \equiv_2 1 \equiv_2 q, \quad |T(V_1)| \equiv_2 |T(X_0)| \equiv_2 0 \equiv_2 p-1 $$
	We now show that $\cE$ is good.
	The subgraph $G[V_0]$ is a path and has an acyclic $T_0$-odd orientation by Lemma \ref{lem:tree}.
	The subgraph $G[V_1]$ is a cycle and $T_1$ satisfies $\calS$ since $T_1 \neq V_1$ from the hypothesis, so it has an acyclic $T_1$-odd orientation by Lemma \ref{lem:cycle}.
	The subgraph $G[V_2]$ is a grid with even dimensions, but $(G[V_2], T_2)$ is not a bad grid instance of Problem \ref{prob:AOP} as the corner vertex $(u_{i+1},v_{q-1}) \in T_2$, so we are done using Lemma \ref{lem:grid:notBadGrid}.\newline

	\phantomsection
	\label{proof:cyl:odd:case:b2iv}
	\noindent \textbf{b.2.iv.~} Suppose that $T(X_0) = \emptyset$ and $T(X_{q-1}) = \emptyset$, so that all vertices of odd degree are not in $T$.
	Thus, since $T$ satisfies $\overline \calS$, we have $T \neq V(G)\setminus \{X_0, X_{q-1}\}$ and there exists a vertex in $u \in V(G)\setminus \{X_0, X_{q-1}\}$ that does not belong to $T$ (otherwise it would precisely be the instance described in Figure \ref{fig:bad_cylindre}).
	But because $T$ satisfies $\calP$ we also have that $|V(G)\setminus T| \equiv_2 p\times q - |E(G)| \equiv_2 0$.
	Thus, there exist in fact two distinct vertices $u,v \in V(G)\setminus \{X_0, X_{q-1}\}$ such that $u,v \not \in T$.
	Moreover, because we assumed that for all $i \in \{0, \dots, p-1\}$, $|T(Y_i)| \equiv_2 p$, there must exist a $k \in \{0, \dots, p-1\}$ such that $|Y_k \setminus T(Y_k)| \geq 4$.
	We now consider an $i$\textsuperscript{th}-column \& first-row $T$-decomposition $\cF_1 = \langle V_0, V_1, V_2 \rangle  = \langle Y_i^{[\geq 1]}, X_0, V \setminus (Y_i^{[\geq 1]} \cup X_0) \rangle$ for some $i\in \{0, \dots, p-1\} \setminus \{k\}$ (see Figure \ref{fig:odd_case_b2iv}).
	It is essentially the same decomposition as in case \textbf{b.2.iii.} and it satisfies $\calP$ for the same reason.
	The subgraph $G[V_2]$ is a grid with even dimensions and $T_2$ verifies both $\calS$ and $\overline\calS$ in $G[V_2]$: indeed, since $Y_i \neq Y_k$ it is enough to look at vertices in $Y_k$, to find either two vertices of degree $4$ not in $T_2$ (i.e. in $Source(T_2)\cap Sink(T_2)$ of $G[V_2]$) whenever $k\neq i-1$ and $k \neq i+1$, or two vertices of degree $3$ in $T_2$ (i.e. in $Sink(T_2)\setminus Source(T_2)$ of $G[V_2]$) and a vertex of degree $3$ not in $T_2$ (i.e. in $Source(T_2)\setminus Sink(T_2)$ of $G[V_2]$) whenever $k = i-1$ or $k=i+1$.
	Now, if $(G[V_2],T_2)$ is a bad grid instance, we consider symmetrically the decomposition $\cF_2 = \langle V_0, V_1, V_2 \rangle = \langle Y_i^{[\leq q-2]}, X_{q-1}, V\setminus (Y_i^{[\leq q-2]} \cup X_{q-1})\rangle $ which satisfies $\calP$ as well (see Figure \ref{fig:odd_case_b2iv}).

	\noindent We now show that either $\cF_1$ or $\cF_2$ is good.
	The subgraph $G[V_0]$ is a path and has an acyclic $T_0$-odd orientation by Lemma \ref{lem:tree}.
	The subgraph $G[V_1]$ is a cycle and $T_1$ satisfies $\calS$ since $T_1 \neq V_1$ from the hypothesis, so it has an acyclic $T_1$-odd orientation by Lemma \ref{lem:cycle}.

	\noindent We now notice that $(G[V_2],T_2)$ cannot be a bad grid instance in both $\cF_1$ and $\cF_2$.
	Indeed, this would imply that $Y_j\setminus T(Y_j) = \{(u_j, v_0), (u_j, v_{q-1})\}$ for all $j \in \{0, \dots, p-1\}\setminus \{i\}$, which is a contradiction as $|Y_k \setminus T(Y_k)| \geq 4$ with $k \in \{0, \dots, p-1\}$ and $i \neq k$ by hypothesis.
	From this observation, we deduce that $G[V_2]$ has an acyclic $T_2$-odd orientation by Lemma \ref{lem:grid:notBadGrid} in either $\cF_1$ or $\cF_2$.
	Consequently, either $\cF_1$ or $\cF_2$ is good which ends the proof.
	

	\end{adjustwidth}
	\end{adjustwidth}
\end{proof}

\newpage
\begin{figure*}[htbp]
	\centering
	\begin{subfigure}[b]{0.66\textwidth}
		\centering
		\includegraphics[height=0.22\textheight, keepaspectratio]{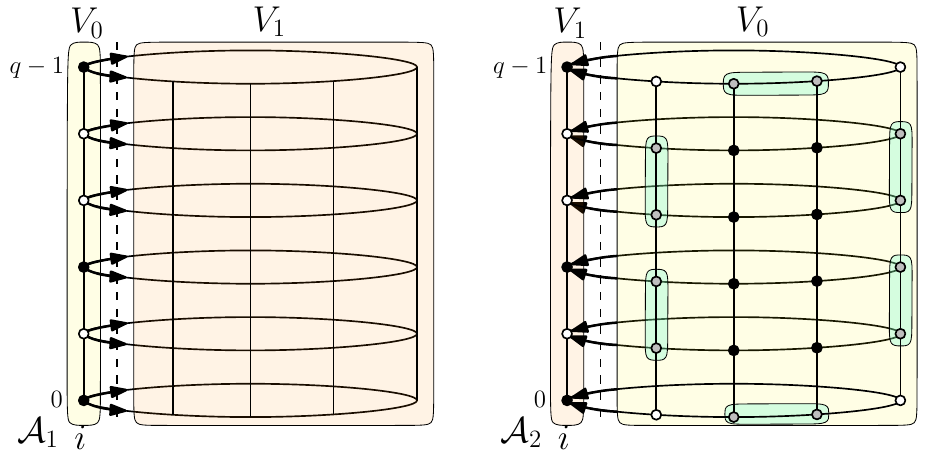}
		\caption{\hyperref[proof:cyl:odd:case:a]{Case a}: $q$ is even}
		\label{fig:odd_case_a}
	\end{subfigure}
	\hfill
	\begin{subfigure}[b]{0.33\textwidth}
		\centering
		\includegraphics[width = \textwidth, keepaspectratio]{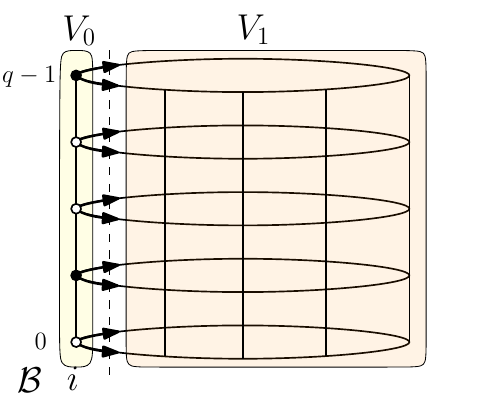}
		\caption{\hyperref[proof:cyl:odd:case:b1]{Case b.1}: $q$ is odd and $T(Y_i)$ satisfies $\calP$ in $G[Y_i]$}
		\label{fig:odd_case_b1}
	\end{subfigure}
	\hfill
	\begin{subfigure}[b]{\textwidth}
		\centering
		\includegraphics[height=0.22\textheight, keepaspectratio]{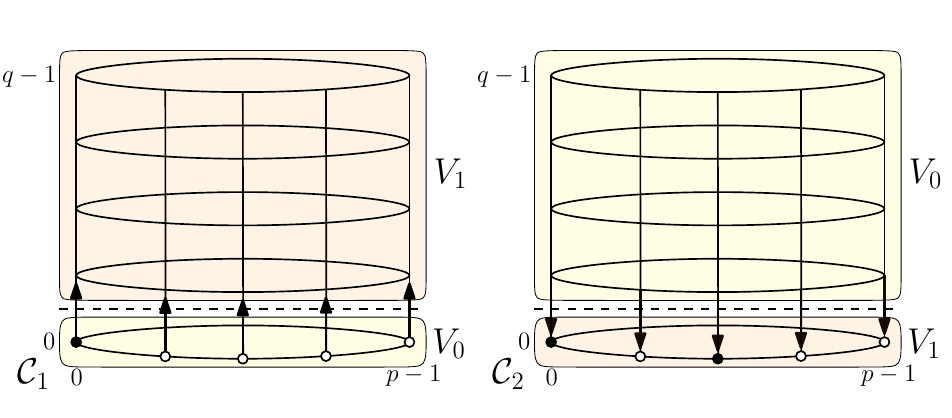}
		\caption{\hyperref[proof:cyl:odd:case:b2i]{Case b.2.i}: $q$ is odd, $T(X_0) \neq X_0$ and $T(X_0) \neq \emptyset$}
		\label{fig:odd_case_b2i}
	\end{subfigure}
	\hfill
	\begin{subfigure}[b]{\textwidth}
		\centering
		\includegraphics[ height=0.22\textheight, keepaspectratio]{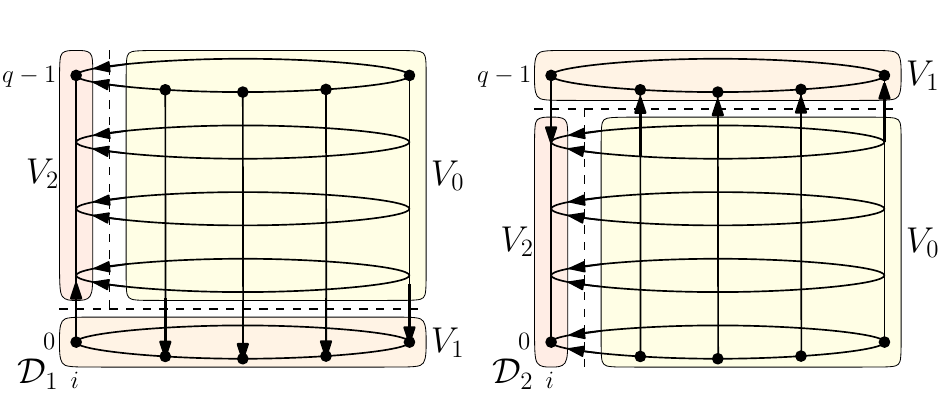}
		\caption{\hyperref[proof:cyl:odd:case:b2ii]{Case b.2.ii}: $q$ is odd, $T(X_0) = X_0$ and $T(X_{q-1}) = X_{q-1}$}
		\label{fig:odd_case_b2ii}
	\end{subfigure}
	\hfill
	\begin{subfigure}[b]{0.32\textwidth}
		\centering
		\includegraphics[width = \textwidth, keepaspectratio]{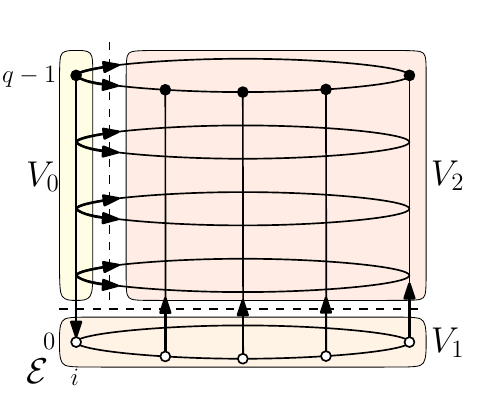}
		\caption{\hyperref[proof:cyl:odd:case:b2iii]{Case b.2.iii}: $q$ is odd, $T(X_0) = \emptyset$ and $T(X_{q-1}) = X_{q-1}$}
		\label{fig:odd_case_b2iii}
	\end{subfigure}
	\hfill
	\begin{subfigure}[b]{0.65\textwidth}
		\centering
		\includegraphics[width = \textwidth, keepaspectratio]{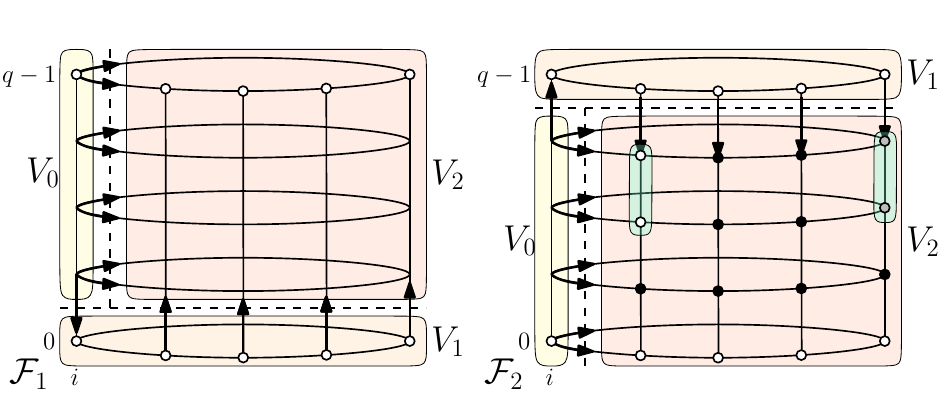}
		\caption{\hyperref[proof:cyl:odd:case:b2iv]{Case b.2.iv}: $q$ is odd, $T(X_0) = \emptyset$ and $T(X_{q-1}) = \emptyset$}
		\label{fig:odd_case_b2iv}
	\end{subfigure}
	\caption{$T$-decompositions used for the proof of Lemma \ref{lem:cyl:odd} in which $p$ is odd by hypothesis.}
\end{figure*}

\FloatBarrier

\begin{lemma}\label{lem:cyl:even}
Let $G$ be a cylinder $C_p\square P_q$ with $p\geq 4$ even and $q\geq 1$.\newline
Then $G$ admits an acyclic $T$-odd orientation for any subset $T\subseteq V(G)$ satisfying $\calP \calS \overline\calS$.
\end{lemma}

\begin{proof}
Let $G$ be a cylinder $C_p\square P_q$ with $p\geq 4$ even and $q > 1$ (for $q = 1$, refer to Lemma \ref{lem:cycle}) and let $T \subseteq V(G)$ such that it satisfies $\calP \calS \overline\calS$.
We go through multiple cases:

\phantomsection
\label{proof:cyl:even:case:a}
\noindent \textbf{a.~} Suppose that $T(X_0)$ does not satisfy $\calP$ in $G[X_0]$ i.e.
$|T(X_0)| \equiv_2 1 \not \equiv_2 |E(G)|$ by \hyperref[decomposition:C0]{(C0)}.
Then there must exist a subpath of $3$ adjacent vertices $w_1, w_2, w_3$ in $X_0$ such that $|T(\{w_1, w_2, w_3\})| \equiv_2 0$.
Indeed, suppose it is not the case, then:
$$
	3|T(X_0)| = \sum_{i=0}^{p-1} |T(\{(u_i,v_0), (u_{i+1}, v_0), (u_{i+2}, v_0) \})| \equiv_2 p
$$
But $p$ is even, so $|T(X_0)| \equiv_2 0$ and it satisfies $\calP$ in $G[X_0]$, a contradiction.
Let $w \in \{w_1, w_2, w_3\}$ such that $w \not \in T$.
We consider the decomposition $\cA_1 = \langle V_0, V_1, V_2 \rangle = \langle \{w\}, X_{\geq 1}, X_0\setminus \{w\} \rangle$ (see Figure \ref{fig:cyl:even_case_a}).
$T(V_0)$ trivially satisfies $\calP$ in $G[V_0]$, and we have $|T_1| = |T(V_1)|+ 1 = |T| - |T(V_0)| + 1 \equiv_2 0 \equiv_ 2 |E(G[V_1])|$, so using Lemma \ref{lem:decomposition_P} it shows that $\cA_1$ satisfies $\calP$.
If $T_1$ does not satisfy either $\calS$ or $\overline\calS$ in $G[V_1]$, then, instead of $\cA_1$, we consider the decomposition $\cA_2 = \langle V_0, V_1, V_2 \rangle = \langle \{w_1, w_2, w_3\}, X_{\geq 1}, X_0\setminus \{w_1, w_2, w_3\} \rangle$ (see Figure \ref{fig:cyl:even_case_a}) which also satisfies $\calP$.
Indeed, by Lemma \ref{lem:decomposition_P}, it is enough to check the following:
\begin{align*}
	&|T_0| \equiv_2 |T(V_0)| \equiv_2 0 \equiv_2 |E(G[V_0])| \\
	&|T_2| \equiv_2 |T(V_2)| - |Z_2| \equiv_2 1 - 1 \equiv_2 0 \equiv_2 |E(G[V_2])|
\end{align*}

\noindent Now, in both partitions, $V_0$ and $V_2$ are paths and admit an acyclic $T_0$-odd and $T_2$-odd orientation respectively by Lemma \ref{lem:tree}.
Moreover, $V_1$ induces a non empty cylinder (as $q>1$) and we show that it admits an acyclic $T_1$-odd orientation by induction on Lemma \ref{lem:cyl:even} in at least one of the partitions.
First, if we consider $\cA_1$ then $T_1$ satisfies both $\calS$ or $\overline\calS$ in $G[V_1]$ by hypothesis. Second, if we consider $\cA_2$, then we define $w_1', w_2', w_3' \in V_1$ as the three neighbors of $w_1, w_2, w_3$ in $V_1$, and we denote $w'\in \{w_1', w_2', w_3'\}$ the neighbor of $w$. By hypothesis, $T(V_1)\triangle \{w'\}$ does not satisfy either $\calS$ or $\overline\calS$ in $G[V_1]$ i.e. $T(V_1)\triangle \{w'\} = V_1$ or  $T(V_1)\triangle \{w'\} = V_1\setminus (X_1 \cup X_{q-1})$.
But then $T(V_1)\triangle \{w_1', w_2', w_3'\} = T_1$ is different from both $V_1$ and $V_1\setminus (X_1 \cup X_{q-1})$ and as such satisfies $\calS$ and $\overline\calS$ in $G[V_1]$ (see Figure \ref{fig:bad_cylindre}).
Hence, either $\cA_1$ or $\cA_2$ is good.\newline

\noindent \textbf{b.~}From now on, we assume that $T(X_0)$ satisfies $\calP$ in $G[X_0]$.
This implies that $|T(X_{\geq 1})| \equiv_2 p$, so that $T(X_{\geq 1}$ also satisfies $\calP$ in $G[X_{\geq 1}]$.
And up to symmetry, we can assume that $T(X_{q-1})$ and $T(X_{\leq q-2})$ satisfy $\calP$ as well in $G[X_{q-1}]$ and $G[X_{\leq q-2}]$ respectively.

\begin{adjustwidth}{2ex}{0pt}
\phantomsection
\label{proof:cyl:even:case:b1}
\noindent \textbf{b.1.~}Suppose that neither $T(X_{\geq 1})$ nor $T(X_{\leq q-2})$ satisfies both $\calS$ and $\overline\calS$ in $G[X_{\geq 1}]$ and $G[X_{\leq q-2}]$ respectively.
Then either $T = V(G)$ in which case $T$ does not satisfy $\calS$ which is a contradiction, or $q=3$ and $T = \emptyset$. In this case, we define the $i$\textsuperscript{th}-column $T$-decomposition $\cB = \langle V_0, V_1 \rangle = \langle Y_i, V(G)\setminus Y_i \rangle$. It satisfies $\calP$ for any $0\leq i\leq p-1$ (see Figure \ref{fig:cyl:even_case_b1}): indeed, we refer to \hyperref[decomposition:C1]{(C1)}, and remark that $|T(V_0)| \equiv_2 |T(Y_i)| \equiv_2 p$.
We now show that the decomposition is good.
$V_0$ induces a path so it has an acyclic $T_0$-odd orientation by Lemma \ref{lem:tree}.
The graph $G[V_1]$ is a grid with an odd dimension, hence it has an acyclic $T_1$-odd orientation by Lemma \ref{lem:grid:notBadPath}.\newline

\noindent \textbf{b.2.~}Assume without loss of generality that $T(X_{\geq 1})$ satisfies both $\calS$ and $\overline\calS$ in $G[X_{\geq 1}]$.
\begin{adjustwidth}{2ex}{0pt}
\phantomsection
\label{proof:cyl:even:case:b2i}
\noindent \textbf{b.2.i.~} Suppose that $T(X_0) \neq \emptyset$.
We consider the first-row $T$-decomposition $\cC = \langle V_0, V_1 \rangle = \langle X_{\geq 1}, X_0 \rangle$ (see Figure \ref{fig:cyl:even_case_b2i}).
This decomposition satisfies $\calP$ as $T(X_0)$ satisfies $\calP$ in $G[X_0]$ by hypothesis and $T_1$ satisfies $\calP$ in $G[V_1]$ by Lemma \ref{lem:decomposition_P}.

\noindent Furthermore, it is good.
Indeed, $G[V_0]$ is a cylinder with an even dimension and has an acyclic $T_0$-odd orientation by induction, and $G[V_1]$ is a cycle in which $T_1$ satisfies $\calS$ since  $T_1 = X_0 \triangle T(X_0) \neq X_0 = V_1$ from the hypothesis, so it has an acyclic $T_1$-odd orientation by Lemma \ref{lem:cycle}.\newline

\phantomsection
\label{proof:cyl:even:case:b2ii}
\noindent \textbf{b.2.ii.~} Suppose that $T(X_0) = \emptyset$.
We consider the first-row $T$-decomposition $\cD_1 = \langle V_0, V_1 \rangle = \langle X_0, X_{\geq 1} \rangle$ (see Figure \ref{fig:cyl:even_case_b2ii}).
The decomposition $\cD_1$ satisfies $\calP$: indeed, we refer to \hyperref[decomposition:C1]{(C1)} and remark that $|T(V_0)| \equiv_2 |T(X_0)| \equiv_2 p$.
Note also that $G[V_0]$ has an acyclic $T_0$-odd orientation by Lemma \ref{lem:cycle}. We now have three cases to consider:
\newline
\noindent \textbf{-~}If $T_1$ satisfies both $\calS$ and $\overline\calS$ in $G[V_1]$ then $G[V_1]$ has an acyclic $T_1$-odd orientation by induction and $\cD_1$ is good.\newline
\noindent \textbf{-~}If $T_1$ does not satisfy $\overline\calS$ in $G[V_1]$, then $T_1 = X_{\geq 1}\setminus \{X_1,X_{q-1}\}$ and $T = V(G)\setminus \{X_0,X_{q-1}\}$, which is precisely saying that $T$ does not satisfy $\overline\calS$ (see Figure \ref{fig:bad_cylindre}), a contradiction.\newline
\noindent \textbf{-~}Finally, if $T_1$ does not satisfy $\calS$ in $G[V_1]$, then $T_1 = X_{\geq 1}$, which implies that $T = V(G)\setminus \{X_0,X_1\}$.
If $q=2$ then $T$ does not satisfy $\overline\calS$ as every vertex of $G$ would be of degree $3$ but none of them would belong to $T$. Hence, we can assume that $q\geq 3$. We define $C_4 = \{(u_0,v_0), (u_0, v_1),(u_1,v_0), (u_1,v_1)\}$ and consider the $T$-decomposition $\cD_2 = \langle V_0,V_1,V_2 \rangle = \langle C_4, X_{\geq 2}, X_{< 2}\setminus C_4 \rangle$ (see Figure \ref{fig:cyl:even_case_b2ii}).
By Lemma \ref{lem:decomposition_P}, it satisfies $\calP$. Indeed, we have:
\begin{align*}
	&|T_0| \equiv_2 |T(V_0)| \equiv_2 0 \equiv_2 |E(G[V_0])| \\
	&|T_2| \equiv_2 |T(V_2)| - |Z_2| \equiv_2 0 - 0 \equiv_2 |E(G[V_2])| &\text{by \hyperref[decomposition:G0]{(G0)}}
\end{align*}

\noindent We now show that it is good.
$G[V_0]$ is a cycle and admits an acyclic $T_0$-odd orientation by Lemma \ref{lem:cycle}. $G[V_1]$ is a cylinder with $T_1$ satisfying $\calP\calS\overline\calS$ since $ (u_0, v_2), (u_1, v_2) \not \in T_1$ (i.e. in $Source(T_1)$ of $G[V_1]$), $(u_2,v_2) \in T_1$ is of degree $3$ in $G[V_1]$ (i.e. in $Sink(T_1)\setminus Source(T_1)$ of $G[V_1]$) .
So $G[V_1]$ has an acyclic $T_1$-odd orientation by induction.
Finally, $(G[V_2],T_2)$ is not a bad grid instance since both corner vertices $(u_3,v_1)$ and $(u_{p-1},v_{1})$ are not in $T_2$. Hence it has an acyclic $T_2$-odd orientation by Lemma \ref{lem:grid:notBadPath}.
\end{adjustwidth}
\end{adjustwidth}
\end{proof}

\begin{figure*}[htbp]
	\centering
	\begin{subfigure}[b]{\textwidth}
		\centering
		\includegraphics[height=0.25 \textheight, keepaspectratio]{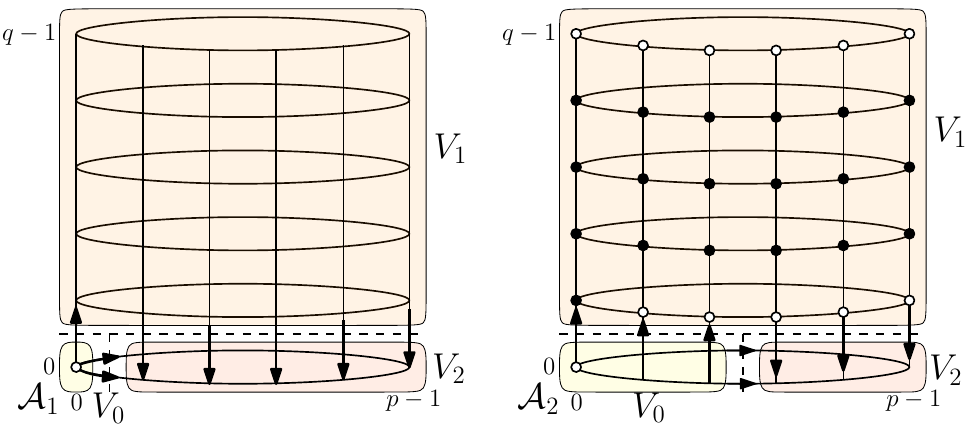}
		\caption{\hyperref[proof:cyl:even:case:a]{Case a}: $T(X_0)$ does not satisfy $\calP$ in $G[X_0]$.}
		\label{fig:cyl:even_case_a}
	\end{subfigure}
	\hfill
	\begin{subfigure}[b]{0.49\textwidth}
		\centering
		\includegraphics[height=0.25\textheight, keepaspectratio]{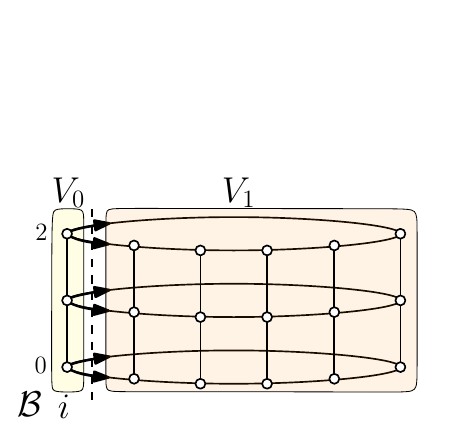}
		\caption{\hyperref[proof:cyl:even:case:b1]{Case b.1}: neither $T(X_{\geq 1})$ nor $T(X_{\leq q-2})$ satisfies both $\calS$ and $\overline\calS$.}
		\label{fig:cyl:even_case_b1}
	\end{subfigure}
	\hfill
	\centering
	\begin{subfigure}[b]{0.49\textwidth}
		\centering
		\includegraphics[height=0.25\textheight, keepaspectratio]{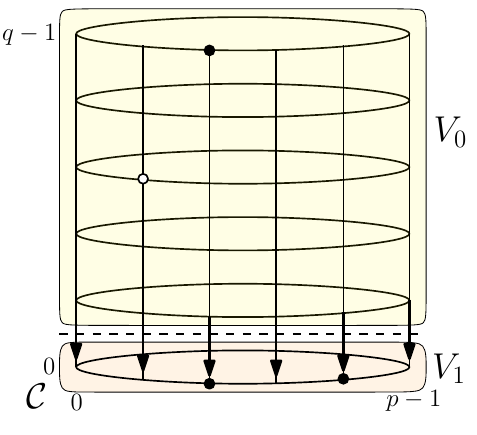}
		\caption{\hyperref[proof:cyl:even:case:b2i]{Case b.2.i}: $T(X_{\geq 1})$ satisfies both $\calS$ and $\overline\calS$ and $T(X_0) \neq \emptyset$.}
		\label{fig:cyl:even_case_b2i}
	\end{subfigure}
	\hfill
	\begin{subfigure}[b]{\textwidth}
		\centering
		\includegraphics[ height=0.25\textheight, keepaspectratio]{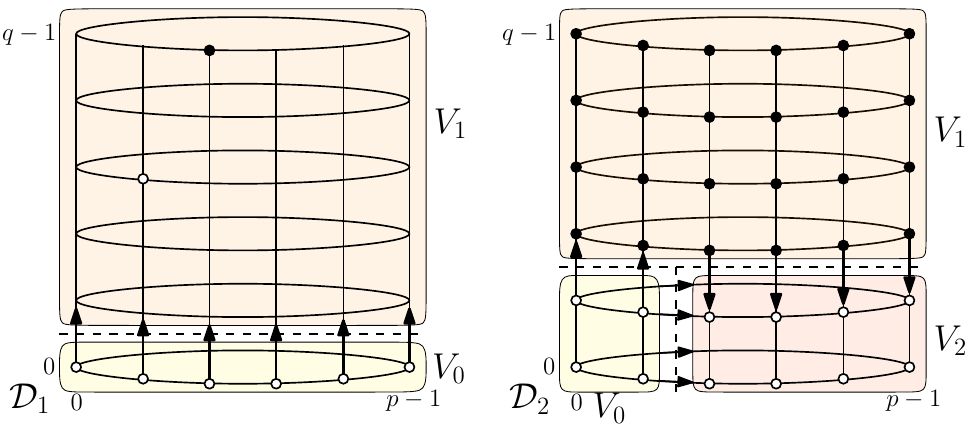}
		\caption{\hyperref[proof:cyl:even:case:b2ii]{Case b.2.ii}: $T(X_{\geq 1})$ satisfies both $\calS$ and $\overline\calS$ and $T(X_0) = \emptyset$}
		\label{fig:cyl:even_case_b2ii}
	\end{subfigure}
	\caption{$T$-decompositions used for the proof of Lemma \ref{lem:cyl:even} in which $p$ is even by hypothesis.}
\end{figure*}

\FloatBarrier

\noindent \textbf{Quasi 2-cylinders}\newline

\noindent In this subsection, we take care of special instances that prove to be useful in the proof of Theorem \ref{thm:tore_carac} for tori.

\noindent We define a quasi 2-cylinder $Q_p$ as a cylinder $C_p\square P_2$ for some integer $p \geq 4$ from which the vertex $x = (u_0,v_1)$ is removed.
We abusively use the notations defined for cylinders and tori, and as such, we note the set $\{(u_1, v_1), \dots, (u_{p-1}, v_1)\}$ as $X_1$ and $\{(u_0, v_0)\}$ as $Y_0$.
It is useful to keep in mind that:
\begin{equation*}
	|V(Q_p)| \equiv_2 1 \text{  and  } |E(Q_p)| \equiv_2 p-1 \tag{QC}\label{decomposition:QC}
\end{equation*}

\noindent Also, in a quasi 2-cylinder $Q_p$, a subset $T$ of the vertex set does not satisfy $\calS$ if and only if $T = V(Q_p)$ and does not satisfy $\overline\calS$ if and only if $T = \{(u_0, v_0), (u_1, v_1), (u_{p-1}, v_1)\}$ (see Figure \ref{fig:bad_quasi2cylindre}).

\begin{figure}[ht]
	\centering
	\includegraphics[width = 0.5 \textwidth, keepaspectratio]{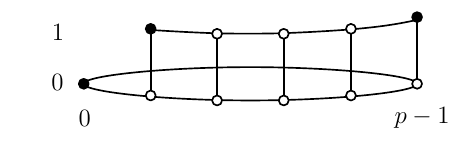}
	\caption{A quasi 2-cylinder $Q_p$ and a subset $T$ satisfying $\calP$ and $\calS$ but not $\overline\calS$.}
	\label{fig:bad_quasi2cylindre}
\end{figure}

\FloatBarrier

\begin{lemma}\label{lem:quasi-2cylindre}
Let $p \geq 4$.
The quasi 2-cylinder $G = (C_p \square P_2) \setminus \{x\}$ admits an acyclic $T$-odd orientation for any subset $T\subseteq V(G)$ satisfying:
\begin{itemize}
	\item $\calP$ if $p$ is odd.
	\item $\calP \calS \overline\calS$ if $p$ is even.
\end{itemize}

\end{lemma}

\begin{proof}
Let $p \geq 4$ and $G=(C_p \square P_2) \setminus \{x\}$.

\noindent \textbf{a.~} Suppose that $p \equiv_2 1$ and let $T\subseteq V(G)$ be a subset satisfying $\calP$.
Then by \hyperref[decomposition:QC]{(QC)}, we have that $|T| \equiv_2 |E(G)| \equiv_2 0$ and $|V(G)\setminus T| \equiv_2 1$.
Hence $T$ satisfies both $\calS$ and $\overline\calS$ as well.
\begin{adjustwidth}{2ex}{0pt}
\phantomsection
\label{proof:quasi-2cylindre:case:a1}
\noindent \textbf{a.1.~} Suppose that $T(X_0)$ satisfies $\calP$ in $G[X_0]$, i.e.,
$|T(X_0)| \equiv_2 p \equiv_2 1$.
Then $T(X_1)$ satisfies $\calP$ in $G[X_1]$ as well, indeed:
\begin{align*}\tag{$*$}
	|T(X_1)| &\equiv_2 |T| - |T(X_0)|\\
	&\equiv_2 |E(G)| - |E(G[X_0])|\\
	&\equiv_2 |E(G[X_1])| + |\delta(X_0, X_1)| - 1\\
	&\equiv_2 |E(G[X_1])|
\end{align*}
\noindent \textbf{-~}If $T(X_0) \neq \{(u_0, v_0)\}$, we consider the $T$-decomposition $\cA_2 = \langle  V_0, V_1 \rangle = \langle  X_1, X_0 \rangle$ (see Figure \ref{fig:q2-cyl:case_a1}).
By Lemma \ref{lem:decomposition_P}, it satisfies $\calP$ since $|T_0| = |T(X_1)|$ satisfies $\calP$ in $G[V_0]$.
Furthermore, it is good. Indeed, the path $G[V_0]$ admits an acyclic $T_0$-odd orientation by Lemma \ref{lem:tree}. Moreover, since $T(X_0) \neq \{(u_0,v_0)\}$, there exists $(u_i,v_0) \in T(X_0)$.
Thus $(u_i,v_0) \in Source(T_1)$ in $G[V_1]$  (see Figure \ref{fig:q2-cyl:case_a1}), so $T_1$ satisfies $\calS$ in $G[V_1]$ and the cycle $G[V_1]$ admits an acyclic $T_1$-odd orientation by Lemma \ref{lem:cycle}.\newline

\noindent \textbf{-~} Else if $T(X_0) = \{(u_0, v_0)\}$, we consider the $T$-decomposition $\cA_1 = \langle  V_0, V_1 \rangle = \langle  X_0, X_1 \rangle$ (see Figure \ref{fig:q2-cyl:case_a1}).
By Lemma \ref{lem:decomposition_P}, it satisfies $\calP$ since $|T_0| = |T(X_0)|$ satisfies $\calP$ in $G[V_0]$.
Furthermore it is good, because $T(X_0)\neq X_0$, so $T_0$ satisfies $\calS$ in $G[V_0]$ and the cycle $G[V_0]$ admits an acyclic $T_0$-odd orientation by Lemma \ref{lem:cycle}.
Moreover, the path $G[V_1]$ admits an acyclic $T_1$-odd orientation by Lemma \ref{lem:tree}.\newline

\phantomsection
\label{proof:quasi-2cylindre:case:a2}
\noindent \textbf{a.2.~} Suppose now that $T(X_0)$ does not satisfy $\calP$ in $G[X_0]$, i.e., $|T(X_0)| \equiv_2 0$.
Then, by $(*)$, $T(X_1)$ does not satisfy $\calP$ either, i.e., $|T(X_1)| \equiv_2 0$.\newline
\noindent \textbf{-~} If $T(X_0) \neq \emptyset$, then there exists $i \in \{1, \dots, p-1\}$ such that $(u_i,v_0) \in T$.
We consider the $T$-decomposition $\cB_1 = \langle  V_0, V_1, V_2 \rangle = \langle  X_0\setminus \{(u_i, v_0)\}, X_{1}, \{(u_i, v_0)\} \rangle$ (see Figure \ref{fig:q2-cyl:case_a2}).
By Lemma \ref{lem:decomposition_P}, it satisfies $\calP$ since:
\begin{align*}
	&|T_0|= |T(V_0)| = |T(X_0)| - |T(\{(u_i, v_0)\})| \equiv_2 1 - 1 \equiv_2 |E(G[V_0])|\\
	&|T_2| = |T(V_2)| - |Z_2| \equiv_2 1 - 3 \equiv_2 |E(G[V_2])|
\end{align*}
Furthermore, it is good since for all $k \in \{0,1,2\}$, the path $G[V_k]$ admits an acyclic $T_k$-odd orientations by Lemma \ref{lem:tree}.\newline

\noindent \textbf{-~} Suppose now that $T(X_0) = \emptyset$.
We consider the $T$-decomposition $\cB_2 =$ $\langle  V_0, V_1, V_2 \rangle =$ \\ $\langle  \{(u_0, v_1)\}, X_{0}, X_1\setminus \{(u_0, v_1)\} \rangle$ (see Figure \ref{fig:q2-cyl:case_a2}).
By Lemma \ref{lem:decomposition_P}, it satisfies $\calP$ since:
\begin{align*}
	&|T_0| = |T(V_0)| \equiv_2 0 \equiv_2 |E(G[V_0])|\\
	&|T_1|= |T(X_0)| - |Z_1| \equiv_2 0 - 1 \equiv_2 |E(G[X_0])|
\end{align*}
Furthermore, it is good since for all $k \in \{0,1,2\}$, the path $G[V_k]$ admits an acyclic $T_k$-odd orientations by Lemma \ref{lem:tree}.\newline
\end{adjustwidth}

\noindent \textbf{b.~} Suppose that $p \equiv_2 0$, so in particular $|X_0| \equiv_2 0$ and $|X_1| \equiv_2 1$.
Let $T\subseteq V(G)$ be a subset satisfying $\calP \calS \overline\calS$. Then by \hyperref[decomposition:QC]{(QC)} we have that $|T| \equiv_2 |E(G)| \equiv_2 1$ and $|V(G)\setminus T| \equiv_2 0$.
\begin{adjustwidth}{2ex}{0pt}
\phantomsection
\label{proof:quasi-2cylindre:case:b1}
\noindent \textbf{b.1.~} Suppose first that $T(X_0)$ satisfies $\calP\calS\overline\calS$ in $G[X_0]$.
We consider the decomposition $\cC = \langle V_0, V_1\rangle = \langle X_0, X_1\rangle$ (see Figure \ref{fig:q2-cyl:case_b1}).
Since $T_0 = T(X_0)$ satisfies $\calP$ and $\calS$ in $G[V_0]$ by hypothesis, using Lemma \ref{lem:decomposition_P} is enough to show that $\cC$ satisfies $\calP$.
Furthermore, it is good since the cycle $G[V_0]$ admits an acyclic $T_0$-odd orientation by Lemma \ref{lem:cycle} and the path $G[V_1]$ admits an acyclic $T_1$-odd orientation by Lemma \ref{lem:tree}.\newline

\phantomsection
\label{proof:quasi-2cylindre:case:b2}
\noindent \textbf{b.2.~} Suppose now that $T(X_0)$ satisfies $\calP$ but does not satisfy $\calS$ (or equivalently $\overline\calS$) in $X_0$, then $T(X_0) = X_0$ and $|T(X_1)| \equiv_2 1$.
In this case, we claim that we can split $X_1$ into two paths such that one of them, say $P$, is of even length and $T(P)$ satisfies $\calP$ in $G[P]$.
More precisely, we claim that there exists $l \in [1,\frac{p-2}{2}]$ such that either $|T(X_1^{[\leq 2l]})| \equiv_2 1$ or $|T(X_1^{[\geq 2l]})| \equiv_2 1$.
Suppose by contradiction that it is not the case, then for all $k \in [1,\frac{p-2}{2}]$, we have $|T(X_1^{[\leq 2k]})|\equiv_2 |T(X_1^{[\geq 2k]})|\equiv_2 0$ and:
$$0 \equiv_2|T(X_1^{[\leq 2k]})| + |T(X_1^{[\geq 2k]})| \equiv_2 |T(X_1)| + |T(\{(u_{2k}, v_1)\})| \equiv_2 1 + |T(\{(u_{2k}, v_1)\})|$$
So for all $k \in [1,\frac{p-2}{2}]$, we get that $(u_{2k}, v_1) \in T(X_1)$.
But we also have for all $k \in [1,\frac{p-2}{2}]$, that $|T(X_1^{[\leq 2k]})|\equiv_2 0$ by hypothesis, so we deduce inductively that $(u_{2k-1}, v_1) \in T(X_1)$ for all $k\in [1,\frac{p-2}{2}]$.
This implies that $|T(X_1)| = |X_1|$ and thus that $T = V(G)$ and does not satisfy $\calS$, which contradicts our assumptions.\newline
Thus, let $l$ be such that, up to symmetry, $|T(X_1^{[\geq 2l]})| \equiv_2 1$.
We define the decomposition $\cD = \langle V_0, V_1, V_2\rangle = \langle X_1^{[\geq 2l]}, X_0, X_1^{[< 2l]}\rangle$ (see Figure \ref{fig:q2-cyl:case_b2}).
By Lemma \ref{lem:decomposition_P}, it satisfies $\calP$ since:
\begin{align*}
	&|T_0| = |T(X_1^{[\geq 2l]})| \equiv_2 1 \equiv_2 |E(G[X_0])|\\
	&|T_2| = |T(X_1^{[< 2l]})| - |Z_2| \equiv_2 |T(X_1)| - |T(X_1^{[\geq 2l]})| -  0 \equiv_2 1 - 1 \equiv_2 |E(G[X_1^{[< 2l]}])|
\end{align*}
Note that $|Z_2| \equiv_2 0$ because $|Z_2| \equiv_2 |\delta(V_0 \cup V_1, V_2)| \equiv_2 |\delta(V_0, V_2)| + |\delta(V_1, V_2)| \equiv_2 1 + 1 \equiv_2 0$.\newline
\noindent We now show that the decomposition is good.
$G[V_0]$ and $G[V_2]$ are paths and admit an acyclic $T_0$-odd and $T_2$-odd orientation respectively by Lemma \ref{lem:tree}. $G[V_1]$ is a cycle and $T_1$ satisfies $\calS$ in $G[V_1]$ since $T_1= T(V_1) \triangle Z_1 = X_0 \triangle X_0^{[\geq 2l]} \neq X_0$ (i.e. $Source(T_1)\neq \emptyset$ in $G[V_1]$), so $G[V_1]$ admits an acyclic $T_1$-odd orientation by Lemma \ref{lem:cycle}.\newline

\noindent \textbf{b.3.~} Suppose now that $T(X_0)$ does not satisfy $\calP$, i.e., $|T(X_0)|\equiv_2 1$ and $|T(X_1)| \equiv_2 0$.
\begin{adjustwidth}{2ex}{0pt}
\phantomsection
\label{proof:quasi-2cylindre:case:b3i}
\noindent \textbf{b.3.i.~} Suppose first that $T(X_0) \neq \{(u_0, v_0)\}$. We consider the decomposition $\cE = \langle V_0, V_1\rangle = \langle X_1, X_0\rangle$ (see Figure \ref{fig:q2-cyl:case_b3i}).
By Lemma \ref{lem:decomposition_P} it satisfies $\calP$ since $|T_0| = |T(X_1)| \equiv_2 |T| - |T(X_0)| \equiv_2 1 - 1 \equiv_2 |E(G[X_1])|$.
We now show it is good.
The path $G[V_0]$ admits an acyclic $T_0$-odd orientation by Lemma \ref{lem:tree}. In order to show that the cycle $G[V_1]$ admits an acyclic $T_1$-odd orientation using Lemma \ref{lem:cycle}, we have to show that $T_1 \neq V_1$. By contradiction, we have:
\begin{align*}
	T_1 = V_1 &\iff T(V_1) \triangle Z_1 = V_1\\
	&\iff T(X_0) \triangle X_0^{[\geq 1]} = X_0 \\
	&\iff T(X_0) = X_0 \triangle X_0^{[\geq 1]}\\
	&\iff T(X_0) = \{(u_0, v_0)\} &\text{ a contradiction}
\end{align*}

\phantomsection
\label{proof:quasi-2cylindre:case:b3ii}
\noindent \textbf{b.3.ii.~} Now suppose that $T(X_0) = \{(u_0, v_0)\}$.
In this case, we claim that we can split $X_1$ into two paths such that one of them, say $P$, is of even length and $T(P)$ does not satisfy $\calP$ in $G[P]$.
More precisely, we claim that there exists $l \in [1,\frac{p-2}{2}]$ such that either $|T(X_1^{[\leq 2l]})| \equiv_2 0$ or $|T(X_1^{[\geq 2l]})| \equiv_2 0$.
Suppose it is not the case, then for all $k \in [1,\frac{p-2}{2}]$, we have $|T(X_1^{[\leq 2k]})|\equiv_2 |T(X_1^{[\geq 2k]})|\equiv_2 1$ and:
$$0 \equiv_2|T(X_1^{[\leq 2k]})| + |T(X_1^{[\geq 2k]})| \equiv_2 |T(X_1)| + |T(\{(u_{2k}, v_1)\})| \equiv_2 |T(\{(u_{2k}, v_1)\})|$$
So for all $k \in [1,\frac{p-2}{2}]$, we get that $(u_{2k}, v_1) \not \in T(X_1)$.
Using that $|T(X_1^{[\leq 2]})| = |T(\{(u_1, v_1), (u_2, v_1)\})|\equiv_2 0$, we deduce that $(u_1, v_1)\in T(X_1)$.
Similarly, because $|T(X_1^{[\geq p-2]})| = |T(\{(u_{p-2}, v_1), (u_{p-1}, v_1)\})|\equiv_2 0$ we deduce that $(u_{p-1}, v_1)\in T(X_1)$.
Then using $|T(X_1^{[\leq 2k]})|\equiv_2 0$ for all $k \in [2,\frac{p-2}{2}]$ we further deduce inductively that no other vertex belongs to $T(X_1)$ and that $T(X_1) = \{(u_1,v_1), (u_{p-1}, v_1)\}$.
This implies that $T = \{(u_0, v_0), (u_1,v_1), (u_{p-1}, v_1)\}$ and does not satisfy $\overline\calS$ as it contains precisely every vertex of even degree (see Figure \ref{fig:bad_quasi2cylindre}), which contradicts our initial assumptions.\newline
Thus, let be such an $l$ and assume up to symmetry that $|T(X_1^{[\geq 2l]})| \equiv_2 0$.
We define the decomposition $\cF = \langle V_0, V_1, V_2\rangle = \langle Y_{\geq 2l}, Y_0, Y_{< 2l}\setminus Y_0 \rangle$ (see Figure \ref{fig:q2-cyl:case_b3ii}).
By Lemma \ref{lem:decomposition_P}, it satisfies $\calP$ since:
\begin{align*}
	&|T_0| = |T(Y_{\geq 2l})| = |T(X_1^{[\geq 2l]})| + |T(X_0^{[\geq 2l]})|\equiv_2 0 + 0 \equiv_2 |E(G[X_1^{[\geq 2l]}])|\\
	&|T_1| = |T(V_1)| - |Z_1| \equiv_2 1 - 1 \equiv_2 |E(G[V_1])|
\end{align*}
\noindent We now prove that the decomposition is good.
The graphs $G[V_0]$ and $G[V_2]$ are grids and admit an acyclic $T_0$-odd and $T_2$-odd orientation respectively by Theorem \ref{thm:grid:carac}.
Indeed, $(G[V_0],T_0)$ is not a bad grid instance since the corner vertex $(u_{p-1}, v_0) \not \in T_0$ and $(G[V_2],T_2)$ is not a bad grid instance either since $G[V_2]$ is a grid with an odd dimension.
\end{adjustwidth}
\end{adjustwidth}
\end{proof}

\begin{figure*}[htbp]
	\centering
	\begin{subfigure}[b]{\textwidth}
		\centering
		\includegraphics[width = \textwidth, keepaspectratio]{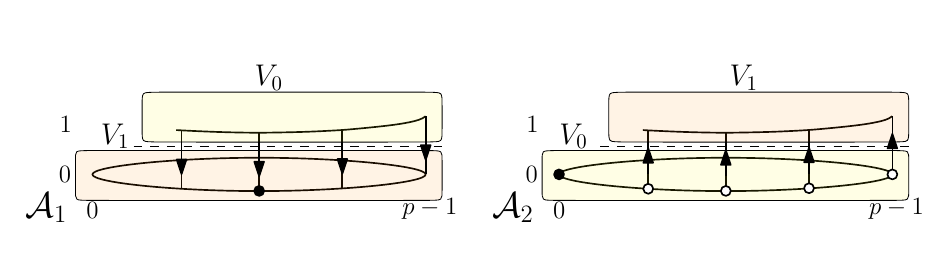}
		\caption{\hyperref[proof:quasi-2cylindre:case:a1]{Case a.1}: $p\equiv_2 1$ and $T(X_0)$ satisfies $\calP$ in $G[X_0]$}
		\label{fig:q2-cyl:case_a1}
	\end{subfigure}
	\hfill
	\begin{subfigure}[b]{\textwidth}
		\centering
		\includegraphics[width = \textwidth, keepaspectratio]{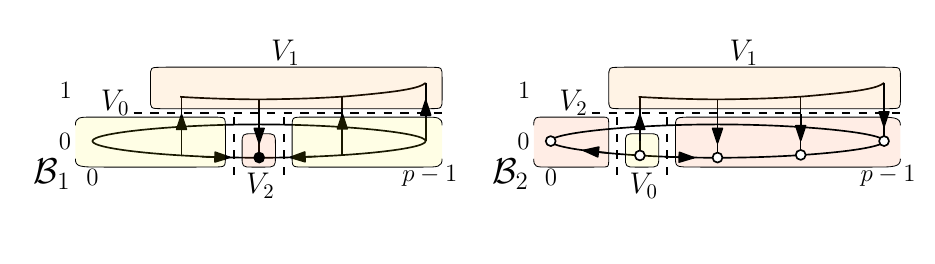}
		\caption{\hyperref[proof:quasi-2cylindre:case:a2]{Case a.2}: $p\equiv_2 1$ and $T(X_0)$ does not satisfy $\calP$ in $G[X_0]$}
		\label{fig:q2-cyl:case_a2}
	\end{subfigure}
	\begin{subfigure}[b]{0.49\textwidth}
		\centering
		\includegraphics[width=\textwidth, keepaspectratio]{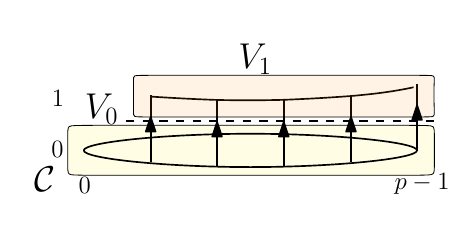}
		\caption{\hyperref[proof:quasi-2cylindre:case:b1]{Case b.1}: $p\equiv_2 0$ and $T(X_0)$ satisfies $\calP \calS \overline\calS$ in $G[X_0]$}
		\label{fig:q2-cyl:case_b1}
	\end{subfigure}
	\hfill
	\centering
	\begin{subfigure}[b]{0.49\textwidth}
		\centering
		\includegraphics[width=\textwidth, keepaspectratio]{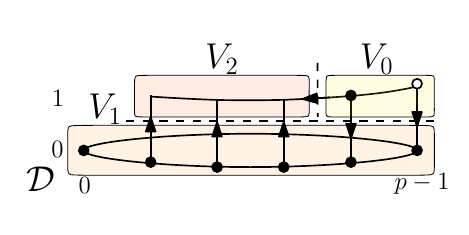}
		\caption{\hyperref[proof:quasi-2cylindre:case:b2]{Case b.2}: $p\equiv_2 0$ and $T(X_0)$ satisfies $\calP$ but not $\calS$ and $\overline\calS$ in $G[X_0]$}
		\label{fig:q2-cyl:case_b2}
	\end{subfigure}
	\hfill
	\begin{subfigure}[b]{0.49\textwidth}
		\centering
		\includegraphics[width=\textwidth, keepaspectratio]{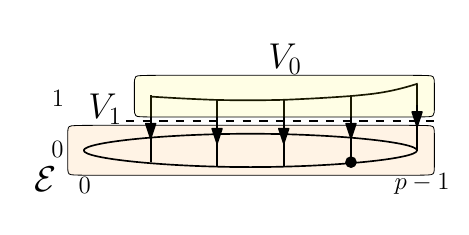}
		\caption{\hyperref[proof:quasi-2cylindre:case:b3i]{Case b.3.i}: $p\equiv_2 0$, $T(X_0)$ does not satisfy $\calP$ in $G[X_0]$ and $T(X_0) \neq \{(u_0, v_0)\}$}
		\label{fig:q2-cyl:case_b3i}
	\end{subfigure}
	\begin{subfigure}[b]{0.49\textwidth}
		\centering
		\includegraphics[width = \textwidth, keepaspectratio]{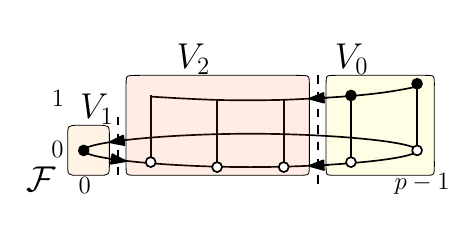}
		\caption{\hyperref[proof:quasi-2cylindre:case:b3ii]{Case b.3.ii}: $p\equiv_2 0$ and $T(X_0) = \{(u_0, v_0)\}$ in $G[X_0]$}
		\label{fig:q2-cyl:case_b3ii}
	\end{subfigure}
	\caption{$T$-decompositions used for the proof of Lemma \ref{lem:quasi-2cylindre}.}
\end{figure*}

\FloatBarrier

\noindent \textbf{Torus}\newline

\noindent This part is dedicated to proving Theorem \ref{thm:tore_carac} for tori.
As before, the proof works by construction and induction.
For any torus $G = C_p\square C_q$, it is useful to remind that:
\begin{equation*}
	|V(G)| \equiv_2 p \times q \text{  ,  } |E(G)| \equiv_2 0 \tag{T}\label{decomposition:T}
\end{equation*}

\begin{lemma}\label{lem:Torus}
Let $p, q\geq 4$.
The Torus $G=C_p\square C_q$ admits an acyclic $T$-odd orientation for any subset $T\subseteq V(G)$ satisfying $\calP \calS \overline\calS$.
\end{lemma}

\begin{proof}
Let $p, q\geq 4$ and $G=C_p\square C_q$.
Let $T\subseteq V(G)$ be a subset which satisfies $\calP \calS \overline\calS$.
By \hyperref[decomposition:T]{(T)} we have that $|T|$ is even.
\newline

\phantomsection
\label{proof:torus:case:a}
\noindent \textbf{a.~} First, suppose that $T(X_0)$ satisfies $\calP$ and $\calS$, that is, $|T(X_0)|\equiv_2 p$, $X_0\setminus T\neq \emptyset$, and consequently that $|T(X_{\geq 1})|\equiv_2 p$.
By Lemma \ref{lem:cycle}, we already have that $G[X_0]$ admits an acyclic $T(X_0)$-odd orientation.
We distinguish two cases depending on the set $T(X_{\geq 1})$.\newline
\noindent \textbf{-~} Suppose that $T(X_{\geq 1}) \neq X_{\geq 1}$ and $T(X_{\geq 1}) \neq X_{> 1}\setminus X_{q-1}$ (this last set is well defined since $q\geq 4$) and let us consider the decomposition $\cA_1= \langle V_0, V_1 \big\rangle= \langle X_0, X_{\geq 1}\big\rangle$ (see Figure \ref{fig:tore:case_a}).
By Lemma \ref{lem:decomposition_P}, this decomposition satisfies $\calP$, indeed $|T_0| = |T(X_0)|$ satisfies $\calP$ in $G[X_0]$ by hypothesis. In order to prove that $G[V_1]$ admits an acyclic $T_1$-odd orientation, and thus that $\cA_1$ is good, we use Theorem \ref{thm:tore_carac} on cylinders. To do so, we have to prove that $T_1$ satisfies $\calS$ and $\overline \calS$. Equivalently, we prove by contradiction that $T_1 \neq X_{\geq 1}$ (i.e. $Source(T_1) \neq \emptyset$ in $G[V_1]$) and $T_1 \neq X_{> 1}\setminus X_{q-1}$ (i.e. $Sink(T_1) \neq \emptyset$ in $G[V_1]$):
\begin{align*}
	T_1 = (X_{\geq 1} \text{ or } X_{> 1}\setminus X_{q-1})&\iff T(V_1) \triangle Z_1 = (X_{\geq 1} \text{ or } X_{> 1}\setminus X_{q-1})\\
	&\iff T(X_{\geq 1}) \triangle (X_1 \cup X_{q-1}) = (X_{\geq 1} \text{ or } X_{> 1}\setminus X_{q-1})\\
	&\iff T(X_{\geq 1}) = (X_{\geq 1}\triangle (X_1 \cup X_{q-1})  \text{ or } (X_{> 1}\setminus X_{q-1}) \triangle (X_1 \cup X_{q-1}))\\
	&\iff T(X_{\geq 1}) = (X_{> 1}\setminus X_{q-1} \text{ or } X_{\geq 1})  \quad \text{ a contradiction}
\end{align*}

\noindent \textbf{-~} Now, suppose that $T(X_{\geq 1})$ is either $X_{\geq 1}$ or $X_{> 1}\setminus X_{q-1}$.
Let $O$ be an acyclic $T(X_0)$-odd orientation of $G[X_0]$.
Let $w_1$ and $w_2$ be vertices of $X_0$ such that $d_{G[X_0]_O}^-(w_1)=d_{G[X_0\setminus\{w_1\}]_O}^-(w_2)=0$. Note that this directly implies that $T(\{w_1, w_2\})$ satisfies $\calP$ in $G[\{w_1,w_2\}]$.
We define the decomposition $\cA_2 = \langle V_0, V_1, V_2\big\rangle = \langle \{w_1, w_2\}, X_{\geq 1}, X_0\setminus \{w_1, w_2\}\big\rangle$ (see Figure \ref{fig:tore:case_a}). By Lemma \ref{lem:decomposition_P} this decomposition satisfies $\calP$ : indeed, $T_0 = T(\{w_1, w_2\})$ satisfies $\calP$ in $G[V_0]$ and:
\begin{align*}
	|T_1| &= |T(X_{\geq 1})| - |Z_1|\\
	&\equiv_2 |T| - |T(X_0)| - 2\\
	&\equiv_2 |E(G)| - |E(G[X_0])|\\
	&\equiv_2 |E(G[X_{\geq 1}])| + |\delta(X_0, X_{\geq 1})|\\
	&\equiv_2 |E(G[X_{\geq 1}])|
\end{align*}
proving that $T_1$ satisfies $\calP$ in $G[V_1]$.
Furthermore, $\cA_2$ is good.
Indeed, by construction, $G[V_0]$ and $G[V_2]$ admit acyclic $T_0$-odd and $T_2$-odd orientations respectively.
Finally, we use \ref{thm:tore_carac} on cylinders to prove that $G[V_1]$ admits an acyclic $T_1$-orientation. To do so, we have to show that $T_1$ satisfies $\calS$ and $\overline \calS$: call $w_1', w_1''$ and $w_2', w_2''$ the four neighbors of $w_1$ and $w_2$ in $V_1$. Then $T_1 = T(V_1)\triangle Z_1 = T(X_{\geq 1}) \triangle \{w_1', w_2', w_1'', w_2''\} \neq X_{\geq 1}$ (i.e. $Source(T_1) \neq \emptyset$ in $G[V_1]$) and $T_1 \neq X_{> 1}\setminus X_{q-1}$ (i.e. $Sink(T_1) \neq \emptyset$ in $G[V_1]$) as $T(X_{\geq 1})$ is either $X_{\geq 1}$ or $X_{> 1}\setminus X_{q-1}$ by hypothesis. Thus $(G[V_1], T_1)$ cannot be one of the instances on which $T_1$ fails $\calS$ or $\overline\calS$ in the cylinder $G[V_1]$ (see Figure \ref{fig:bad_cylindre}).\newline

\noindent \textbf{b.~}
One may assume now that for all $i$ $\in \{0, \dots, p-1\}$, either $T(X_i)$ does not verify $\calS$ in $G[X_i]$ (i.e., $X_i \subseteq T$) or does not verify $\calP$ in $G[X_i]$ (i.e., $|T(X_i)| \equiv_2 p-1$). Equivalently, for all $i \in \{0, \dots, p-1\}$, $|X_i \setminus T(X_i)|$ is either $0$ or odd.
By symmetry, assume also that for all $j \in \{0, \dots, q-1\}$, $|T(Y_j)| \equiv_2 q$ if and only if $Y_j\subseteq T$ (i.e., $|Y_j \setminus T(Y_j)|$ is either $0$ or odd).
Furthermore, since $T$ verifies $\calS$, we may suppose $T(X_0)$ satisfies $\calS$, which implies $|T(X_0)| \equiv_2 p-1$ and $X_0\setminus T\neq \emptyset$. As a consequence, and since $|T|$ is even by \hyperref[decomposition:T]{(T)}, we have $|T(X_{\geq 1})| \equiv_2 p-1$.
Let $x\in X_0\setminus T$ and by symmetry, suppose that $x=(u_0, v_0)$.

\begin{adjustwidth}{2ex}{0pt}
\phantomsection
\label{proof:torus:case:b1}
\noindent \textbf{b.1.~} If $p$ is even then $|T(X_{0})|\equiv_2 1$, so we can assume up to rotation that $(u_{p-1}, v_{0})\in T$.
Also, it implies $|T(X_{\geq 1})|\equiv_2 1$, so there is $0<i\leq q-1$ such that $|T(X_i)|\equiv_2 1$.
Let $k$ be the smallest such index.
Since $p\geq 3$, up to symmetry, we may assume that $k<q-1$.
By definition of $k$, we have that $X_j\subseteq T$ for all $0<j<k$. We now consider two cases:

\noindent \textbf{-~} Suppose there exists $j \in \{k+1, \ldots, q-1\}$ such that $X_j \not \subset T$.
We consider the decomposition $\cB_1 = \langle V_0, V_1\rangle = \langle X_{< k+1}, X_{\geq k+1}\rangle$ (see Figure \ref{fig:tore:case_b1}).
By Lemma \ref{lem:decomposition_P}, it satisfies $\calP$ as $|T_0| = |T(X_{< k+1})| \equiv_2 |T(X_0)| + |T(X_k)| \equiv_2 1 + 1 \equiv_2 0$.
Also, $\cB_1$ is good.
Indeed, $T_0$ satisfies $\calS$ and $\overline\calS$ in $G[V_0]$ as $x \in Source(T_0)$ in the cylinder $G[V_0]$ and $(u_{p-1}, v_0) \in Sink(T_0)$ in $G[V_0]$. Thus, $G[X_0]$ admits an acyclic $T_0$-odd orientation by Lemma \ref{lem:cyl:even}.
Furthermore, $T_1$ satisfies $\calS$ in the cylinder $G[V_1]$ since $T(X_j)\neq  X_j$ by hypothesis (i.e. $Source(T_1)\neq \emptyset$ in $G[V_1]$) and it satisfies $\overline\calS$ in $G[V_1]$ indeed: 
by hypothesis $T(X_{\geq k+1}) \neq X_{\geq k+1}$ and $T_1 = T(V_1) \triangle Z_1 = T(X_{\geq k+1}) \triangle (X_{k+1} \cup X_{q-1}) \neq X_{> k+1}\setminus X_{q-1}$ (i.e. $Sink(T_1)\neq \emptyset$ in $G[V_1]$).
Thus $(G[V_1], T_1)$ is not an instance for which $T_1$ fails $\overline\calS$ in $G[V_1]$ (see Figure \ref{fig:bad_cylindre}).
So, by Lemma \ref{lem:cyl:even}, the cylinder $G[V_1]$ admits an acyclic $T_1$-odd orientation.\newline

\noindent \textbf{-~} Suppose now that $T(X_j) = X_j$ for all $j \in \{k+1, \ldots, q-1\}$.
We consider the decomposition $\cB_2 = \langle V_0, V_1, V_2\rangle = \langle Y_0^{[< q-1]}, X_{q-1} \cup X_0\setminus \{x\}, V\setminus (V_0 \cup V_1) \rangle$ (see Figure \ref{fig:tore:case_b1}).
By Lemma \ref{lem:decomposition_P}, it satisfies $\calP$ since:
\begin{align*}
	|T_0| &= |T(Y_{0}^{[<q-1]})| \equiv_2 |T(Y_{0})| - 1 \equiv_2 q \equiv_2 |E(G[Y_{0}])| \equiv_2 |E(G[Y_{0}^{[<q-1]}])|\\
	|T_1| &= |T(X_{q-1} \cup X_0\setminus \{x\})| - |Z_1| \\
	&\equiv_2 |T(X_{q-1})| + |T(X_0)| - 4 \\
	&\equiv_2 0 + 1 - 4 \\
	&\equiv_2 |E(G[X_{q-1} \cup X_0\setminus \{x\}])| \quad \text{by \hyperref[decomposition:QC]{(QC)}}
\end{align*}
Furthermore, $\cB_2$ is good.
By Lemma \ref{lem:tree}, the path $G[V_0]$ admits an acyclic $T_0$-odd orientation.
By Lemma \ref{lem:quasi-2cylindre}, the quasi-cylinder $G[V_1]$ admits an acyclic $T_1$-odd orientation.
Indeed, $T_1$ verifies $\overline\calS$ because $T_1(X_{q-1}) = X_{q-1} \subset Sink(T_1)$ in $G[V_1]$ and it verifies $\calS$ since $(u_{p-1}, v_0) \in Source(T_1)$ in $G[V_1]$.
Furthermore, $V_2$ induces a grid with an odd dimension so Lemma \ref{lem:grid:notBadPath} ensures that $G[V_2]$ admits an acyclic $T_2$-odd orientation.
\newline
\end{adjustwidth}

\begin{adjustwidth}{2ex}{0pt}
\noindent \textbf{b.2.~} Suppose now that $p$ and $q$ are odd.
Recall that we supposed in \textbf{b.} that for all $i \in \{0, \ldots, p-1\}$ either $T(X_i) = X_i$ or $|T(X_i)|$ is even, and for all $j \in \{0, \ldots, q-1\}$ either $T(Y_j) = Y_j$ or $|T(Y_j)|$ is even.\newline
\end{adjustwidth}

\begin{adjustwidth}{4ex}{0pt}
\phantomsection
\label{proof:torus:case:b2i}
\noindent \textbf{b.2.i~} Assume that $X_1\subseteq T$. Then, $T(X_{<2})$ and $T(X_{\geq 2})$ verify $\calP$ in $G[X_{<2}]$ and $G[X_{\geq 2}]$ respectively.
Indeed:
\begin{align*}
 	&|T(X_{<2})|=|T(X_0)| + |T(X_1)| \equiv_2 0 + 1 \equiv_2 |E(G[X_{<2}])| &\text{by \hyperref[decomposition:C0]{(C0)}}\\
	&|T(X_{\geq 2})| = |T| - |T(X_{<2})| \equiv_2 0 - 1 \equiv_2 |E(G[X_{\geq 2}])|  &\text{by \hyperref[decomposition:C0]{(C0)}}
\end{align*}
We rely on Figure \ref{fig:tore:case_b2i} to illustrate the following cases.\newline
\noindent \textbf{-~}If $T(X_{\geq 2})$ verifies $\calS$ and $\overline\calS$ in $G[X_{\geq 2}]$ then the decomposition $\cC_1 = \langle V_0, V_1\rangle = \langle X_{\geq 2}, X_{<2}\rangle$ is good by Lemma \ref{lem:cyl:odd}.\newline
\noindent \textbf{-~}If $T(X_{\geq 2})$ does not verify $\overline\calS$ in $G[X_{\geq 2}]$ then $T(X_2)=T(X_{q-1})=\emptyset$ and $X_i\subseteq T$ for all $2<i<q-1$ (see Figure \ref{fig:bad_cylindre}).
In this case, since $p\geq 5$ and by Lemma \ref{lem:decomposition_P}, the decomposition $\cC_2 = \langle V_0, V_1\rangle =  \langle X_1\cup X_2, X_0\cup X_{\geq 3}\rangle$ satisfies $\calP$ as $|T(X_{1}\cup X_2)|=|T(X_1)| + |T(X_2)| \equiv_2 1 + 0 \equiv_2 |E(G[X_1 \cup X_2])|$.
Furthermore, the decomposition is good by Lemma \ref{lem:cyl:odd}.
Indeed, $T_0$ satisfies $\calS$ and $\overline\calS$ in $G[V_0]$ as $X_1 \in Sink(T_0)$ in $G[V_0]$ and $X_2 \in Source(T_0)$ in $G[V_0]$, and $T_1$ satisfies both $\calS$ and $\overline\calS$ in $G[V_1]$ as $X_{q-1} \in Source(T_1)\cap Sink(T_1)$ in $G[V_1]$. Note also that because we supposed that for all $i \in \{0, \ldots, p-1\}$ either $T(X_i) = X_i$ or $|T(X_i)|$ is even, it imposes that $T(X_0) = \emptyset$. 
 \newline
\noindent \textbf{-~}If $T(X_{\geq 2})$ does not verify $\calS$ in $G[X_{\geq 2}]$ then $X_i\subseteq T$ for all $1 \leq i \leq p-1$.
Suppose first that $T(X_0)\neq \emptyset$ and up to symmetry we can assume that $(u_1, v_0)\in T$. Also, since $T$ satisfies $\calS$, i.e., $V \setminus T \neq \emptyset$ and $|V \setminus T \cup \{v \in T, d(v) \equiv_2 1\}| > 1$, we know $|V\setminus T|> 1$.
Hence, there exists $y\neq (u_{0}, v_0)$ with $y\in X_0\setminus T$.
Let us now consider the decomposition $\cC_3 = \langle V_0, V_1\rangle = \langle Y_{<2}, Y_{\geq 2}\rangle$.
By Lemma \ref{lem:decomposition_P} it satisfies $\calP$ as $|T_0| = |T(Y_{<2})| = |T(Y_0)| + |T(Y_1)| \equiv_2 0 + 1 \equiv_2 |E(G[Y_{<2}])|$.
Furthermore, the decomposition is good by Lemma \ref{lem:cyl:odd}.
Indeed, $T_0$ satisfies $\calS$ and $\overline\calS$ in the cylinder $G[V_0]$ as $x \in Source(T_0)$ in $G[V_0]$ and $(u_1, v_0)\in Sink(T_0)$ in $G[V_0]$.
Furthermore, $T_1$ satisfies $\calS$ in the cylinder $G[V_1]$ as $Y_2^{[\geq 1]} \in Source(T_1)$ in $G[V_1]$ and also satisfies $\overline\calS$ in $G[V_1]$ as $y$ either belongs to $Y_2$ or $Y_{q-1}$, in which case it belongs to $T_1$ and is of degree $3$ in $G[V_1]$, or it does not belong to $Y_2$ nor $Y_{q-1}$, in which case it does not belong to $T_1$ and is of degree $4$ in $G[V_1]$. In both cases, $y \in Sink(T_1)$ in $G[V_1]$.\newline

\noindent Finally, assume that $T(X_0)=\emptyset$.
In this case, we consider the decomposition $\cC_4 = \langle V_0, V_1, V_2\rangle = \langle \{x\}, Y_{\geq 2}, Y_{<2}\setminus \{x\}\rangle$.
By Lemma \ref{lem:decomposition_P}, the decomposition satisfies $\calP$:
\begin{align*}
	&|T_0| = |T(\{x\})| \equiv_2 0\\
	&|T_1| = |T(Y_{\geq 2})| - |Z_1| \equiv_2 1 - 1 \equiv_2 |E(G[Y_{\geq 2}])|
\end{align*}
Furthermore, the decomposition is good.
Indeed, the cylinder $G[V_1]$ admits an acyclic $T_1$-odd orientation by Lemma \ref{lem:cyl:odd} since $T_1$ satisfies both $\calS$ and $\overline\calS$ in $G[V_1]$ as the vertex $(u_3,v_0) \in Sink(T_1)$ in $G[V_1]$ because $q\geq 4$ and $(u_2,v_0) \in Source(T_1)$ in $G[V_1]$.
Also the quasi-cylinder $G[V_2]$ admits an acyclic $T_2$-orientation by Lemma \ref{lem:quasi-2cylindre} since $q$ is odd.\newline
\end{adjustwidth}

\begin{adjustwidth}{4ex}{0pt}
\phantomsection
\label{proof:torus:case:b2ii}
\noindent \textbf{b.2.ii~} Suppose now that for all $i \in \{0,\dots, p-1\}$ and $j\in \{0,\dots, q-1\}$, we have $|T(X_i)| \equiv_2 |T(Y_j)| \equiv_2 0$.
We consider the decomposition $\cD = \langle \{x\}, X_{\geq 2}, X_{<2}\setminus \{x\}\rangle$ (see Figure \ref{fig:tore:case_b2ii}).
By Lemma \ref{lem:decomposition_P}, it satisfies $\calP$:
\begin{align*}
	&|T_0| = |T(\{x\})| \equiv_2 0\\
	&|T_1| = |T(X_{\geq 2})| - |Z_1| \equiv_2 (q-2)\times (p-1) - 1 \equiv_2 p \equiv_2 |E(G[X_{\geq 2}])| \quad \text{by \hyperref[decomposition:C0]{(C0)}}
\end{align*}
We now show the decomposition is good. 
From the hypothesis, $|T(X_2)| \equiv_2 p-1$ and by construction of the decomposition $Z_1 = X_2 \cup X_{q-1}$. This implies that $X_2$ contains both vertices that are and are not in $T_1 = T(V_1) \triangle Z_1$. In particular, $T_1(X_2)\subset Sink(T_1)$ in $G[V_1]$ as those are vertices of degree $3$, and $X_2\setminus T_1(X_2) \subset Source(T_1)$ in $G[V_1]$. Therefore $T_1$ satisfies both $\calS$ and $\overline\calS$ in $G[V_1]$ and the cylinder $G[V_1]$ admits an acyclic $T_1$-odd orientation by Lemma \ref{lem:cyl:odd}. Also the quasi-cylinder $G[V_2]$ admits an acyclic $T_2$-orientation by Lemma \ref{lem:quasi-2cylindre} since $q$ is odd.\newline

\end{adjustwidth}

\noindent In all cases, we get a good $T$-decomposition and one can conclude by applying Theorem \ref{thm:T-odd decomposition}.
\end{proof}

\FloatBarrier

\begin{figure*}[htbp]
	\centering
	\begin{subfigure}[b]{\textwidth}
		\centering
		\includegraphics[height=0.23 \textheight, keepaspectratio]{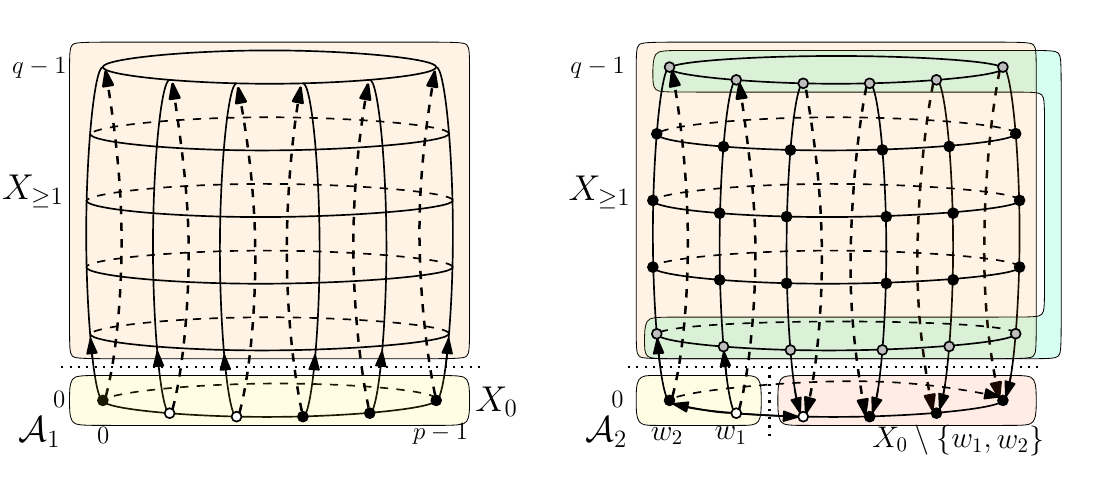}
		\caption{\hyperref[proof:torus:case:a]{Case a}: $T(X_0)$ satisfies $\calS$ and $\calP$ in $G[X_0]$}
		\label{fig:tore:case_a}
	\end{subfigure}
	\hfill
	\begin{subfigure}[b]{\textwidth}
		\centering
		\includegraphics[height=0.23 \textheight, keepaspectratio]{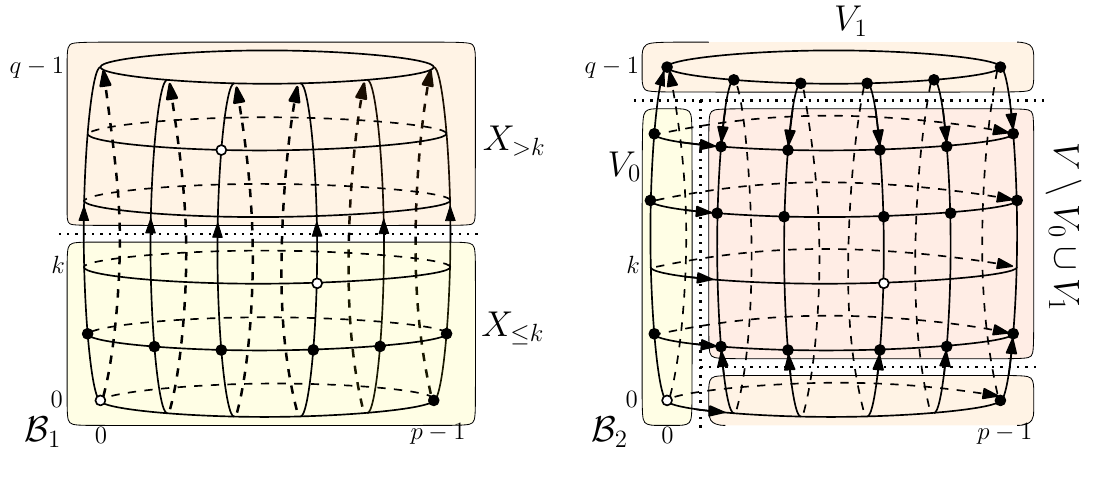}
		\caption{\hyperref[proof:torus:case:b1]{Case b.1}: $p$ is even, $|X_i\setminus T(X_i)|$ and $|Y_j\setminus T(Y_j)|$ are $0$ or odd for all $i,j$}
		\label{fig:tore:case_b1}
	\end{subfigure}
	\hfill
	\begin{subfigure}[b]{\textwidth}
		\centering
		\includegraphics[height=0.46 \textheight, keepaspectratio]{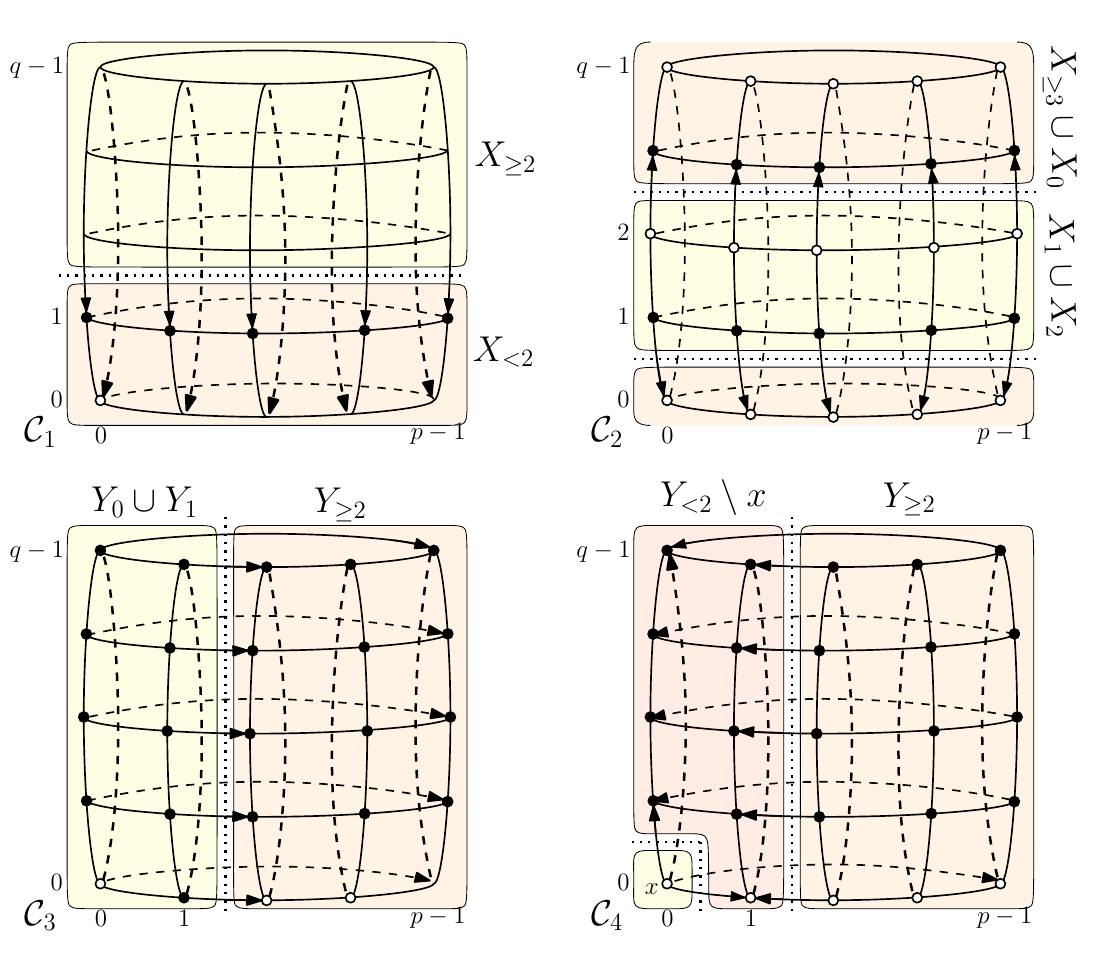}
		\caption{\hyperref[proof:torus:case:b2i]{Case b.2.i}: $p,q$ are odd, $|X_i\setminus T(X_i)|$ and $|Y_j\setminus T(Y_j)|$ are $0$ or odd for all $i,j$ and $X_1\subseteq T$}
		\label{fig:tore:case_b2i}
	\end{subfigure}
	\caption{$T$-decompositions used for the proof of Lemma \ref{lem:Torus}}
\end{figure*}

\begin{figure*}[htbp]
	\centering
	\includegraphics[height=0.25 \textheight, keepaspectratio]{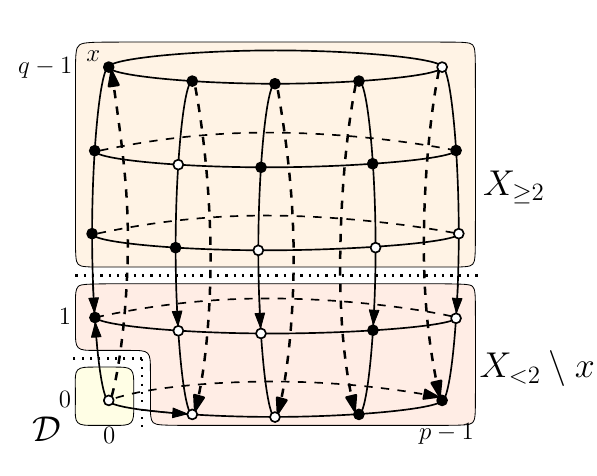}
	\caption{\hyperref[proof:torus:case:b2ii]{Case b.2.ii}: $p,q$ are odd, $|X_i\setminus T(X_i)| \equiv_2 |Y_j\setminus T(Y_j)| \equiv_2 1$ for all $i$ and $j$}
	\label{fig:tore:case_b2ii}
\end{figure*}

\subsection{Final remarks}

It seems that Lemma \ref{lem:Torus} works when $p=3$ and $q>3$.
Nevertheless, the proof would be even more technical to settle this case.
Moreover, the configuration illustrated in Figure \ref{fig:badboy} shows that with $p=q=3$ the result fails.

\begin{figure}[ht]
    \centering
    \includegraphics[width=0.3\textwidth]{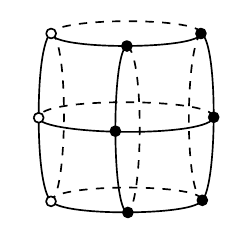}
    \caption{The Torus $C_3 \times C_3$ with a bad $T$.}
    \label{fig:badboy}
\end{figure}

\FloatBarrier

\section{Proof of Corollary \ref{thm:strict_inclusion}} \label{sec:class_inclusion}
This section is dedicated to the proof of Corollary \ref{thm:strict_inclusion}.
For readability, we recall Theorem \ref{thm:treilli} and Corollary \ref{thm:strict_inclusion} here :

\begin{theorem*}[\ref{thm:treilli}]~
    \begin{enumerate}[label= $(\alph*)$]
        \item $\cC_{\calS} = \cC_{\overline\calS}= \{\OneCirc, \TwoCirc\}$,
        \item $\cC_{\calS\overline\calS}=\{\OneCirc, \TwoCirc, \TwoCircLinked\}$,
        \item $\cC_{\calP}= \{\OneCirc\}\cup \{G \in \cC_{\calP \calS \overline\calS} \ | \ G \text{ is connected, non-Eulerian, and } |V(G)| + |E(G)| \equiv_2 1\}$,

        \item $\cC_{\calP\calS} \setminus \cC_{\calP}
        =\cC_{\calP\overline\calS} \setminus \cC_{\calP}
        = \{\TwoCirc, \ThreeCircIndependant\}\cup \{G\in \cC_{\calP \calS \overline\calS} \mid G \text{ is Eulerian}\}$,

		\item $\cC_{\calP \calS \overline\calS} \setminus  \cC_{\calP \calS} \subset \{\ThreeCirc\} \cup \{G \ | \ G \text{ is connected, non-Eulerian, and } |V(G)| + |E(G)| \equiv_2 0\}$.
    \end{enumerate}
\end{theorem*}~

\begin{corollary*}[\ref{thm:strict_inclusion}] Let $p,q,k \geq 1$ be three integers.
    \begin{itemize}
        \item [(c)'] The Cylinders $C_{2p+1}\square P_{2q}$, the Grids $P_{2p+1}\square P_q$, and Trees are in $\cC_{\calP}$.
        \item [(d)'] The Tori $C_{p}\square C_q$ with $p,q \geq 4$ and Cycles are in $\cC_{\calP\calS}\setminus\cC_{\calP}$.
        \item [(e)'] The Cylinders $C_{2p}\square P_q, C_{2p+1}\square P_{2q+1}$ are in $\cC_{\calP \calS \overline\calS}\setminus\cC_{\calP \calS}$.
	\end{itemize}
	Moreover,
	\begin{itemize}
        \item [(c)''] The Cliques $K_{4k+2} \in \{G \text{ connected and non-Eulerian} \mid |V(G)| + |E(G)| \equiv_2 1\} \setminus \cC_{\calP}$.
		\item [(d)''] The Cliques $K_{2k+1} \in \{ \circ \ \circ\} \cup \{G \mid G \text{ is Eulerian}\} \setminus \cC_{\calP \calS}$.
		\item [(e)''] The Cliques $K_{4k} \in \{G \text{ connected and non-Eulerian}  \mid |V(G)| + |E(G)| \equiv_2 0\} \setminus \cC_{\calP \calS \overline\calS}$.
    \end{itemize}
\end{corollary*}

\subsection{Proof of (c'), (d') and (e')}

\noindent Proof of \textbf{(c)'}:
By Lemma \ref{lem:tree}, all trees belong to $\cC_\calP$.
Furthermore, by Lemma \ref{lem:grid:notBadPath} on grids, if $G = P_p \square P_q$ with $p$ or $q$ odd, then $G \in \cC_{\calP \calS \overline\calS}$.
Similarly, by Lemma \ref{lem:cyl:odd} on cylinders, if $G = C_p \square P_q$ with $p$ odd and $q$ even, then $G \in \cC_{\calP \calS \overline\calS}$.
 In both cases, the graph $G$ is connected, non-Eulerian and the condition $|V(G)| + |E(G)| \equiv_2 1$ is satisfied. So, by \textbf{(c)} of Theorem \ref{thm:cp_carac}, we have $G \in \cC_\calP$ .
\qed\newline

\noindent Proof of \textbf{(d)'}:
By Lemma \ref{lem:cycle}, all cycles belong to $\cC_{\calP \calS}\setminus \cC_\calP$.
Furthermore, by Lemma \ref{lem:Torus}, any tori $C_p \square C_q$ with $p>3$ or $q>3$ is in $\cC_{\calP \calS \overline\calS}$.
In these cases, those graphs are Eulerian and hence belong to $\cC_{\calP \calS} \setminus \cC_\calP$ from \textbf{(d)} of Theorem \ref{thm:cp_carac}.
\qed\newline

\noindent Proof of \textbf{(e)'}: Theorem \ref{thm:tore_carac} states that all cylinders belongs to  $\cC_{\calP \calS \overline\calS}$. Furthermore, $C_{2p+1}\square P_{2q+1}$ and $C_{2p}\square P_{q}$ are non Eulerian and verify $|V(G)| + |E(G)| \equiv_2 0$, thus $C_{2p+1}\square P_{2q+1}, C_{2p}\square P_{q}$ are in $\cC_{\calP \calS \overline\calS}\setminus\cC_{\calP \calS}$ from \textbf{(c)} and \textbf{(d)} of Theorem \ref{thm:cp_carac}.
\qed \newline

\subsection{Proof of (c''), (d'') and (e'')}

\noindent We first prove the following lemma on cliques:

\begin{lemma}\label{lem:K2p+1}
    Let $p> 1$ be an integer, and $T$ be any subset of $p-2$ vertices of a complete graph $K_{n}$ on $n\in \{2p+1, 2p+2\}$ vertices.
We have:
    \begin{itemize}
    \item [-] $K_{2p+1}\in \{G \ | \ |V(G)| + |E(G)| \equiv_2 p+1\}\cap \{G \ | \ G \text{ is Eulerian} \}$.
    \item [-] $K_{2p}\in \{G \ | \ |V(G)| + |E(G)| \equiv_2 p\}\cap \{G \ | \ G \text{ is connected and not Eulerian} \}$.
    \item [-] $T$ verifies $\calP \calS \overline\calS$ in $K_n$.
    \item [-] If $n\geq 5$ then $K_n$ has no acyclic $T$-odd orientation.
    \end{itemize}
\end{lemma}

\begin{proof}
The first two parts come from the facts that for $n\in \{2p, 2p+1\}$, we have $|V(K_n)| =n$, $|E(K_n)| = \frac{n (n-1)}{2} \equiv_2 p$ and $d(v)\equiv_2 n-1$ for every vertex $v$.
For the third part, $\calP$ follows from $|T| \equiv_2 p$ and $\calS$ from $\vert V(K_n)\setminus T\vert \geq 2$, since $p>1$. Moreover, $\overline\calS$ follows from $d(v)\equiv_2 n-1$ for all $v$ in $V(K_n)$: indeed,  vertices in $V(G)\setminus T$ are potential sinks whenever $p$ is odd and vertices in $T$ are potential sinks whenever $p$ is even.

\noindent The last assertion is a direct consequence of the claim asserting that the complete graph on $n$ vertices admits an acyclic $T$-odd orientation if and only if $|T| = \lfloor \frac{n}{2}\rfloor$.
This claim can be easily proved by induction on $n$.
\end{proof}

\noindent Now, we know from Lemma \ref{lem:K2p+1}, that $K_n \not \in \cC_{\calP \calS \overline\calS}$ for any $n\geq 5$.
And we also know that:\newline
\noindent \textbf{-~} For any odd $p\geq 1$, we have that $K_{2p} \in \{G \ | \ |V(G)| + |E(G)| \equiv_2 1\}\setminus  \cC_{\calP\calS\overline\calS}$ and is non-Eulerian.
Thus, \textbf{(c'')} of Theorem \ref{thm:strict_inclusion} is a direct consequence of this last observation.

\noindent \textbf{-~} For any $p\geq 1$, we have $K_{2p+1} \in  \{G \ | \ G \text{ is Eulerian} \}$.
Thus, \textbf{(d'')} of Theorem \ref{thm:strict_inclusion} is a direct consequence of this last observation.

\noindent \textbf{-~} For any even $p\geq 1$, we have that $K_{2p} \in \{G \ | \ |V(G)| + |E(G)| \equiv_2 0\}\setminus  \cC_{\calP\calS\overline\calS}$ and is non-Eulerian.
Thus, \textbf{(e'')} of Theorem \ref{thm:strict_inclusion} is a direct consequence of this last observation.
\hfill \qed

\section{Conclusion}

\noindent This paper advances our understanding of the acyclic orientation problem with parity constraints. We introduced new necessary conditions, the first one, $\calS$ ensures the existence of a source which would not be the unique sink, and symmetrically $\overline\calS$ ensures the existence of a sink which would not be the unique source . 
From those necessary condition, we established a hierarchy of graph classes, $\cC_{\cN}$ for $\cN \subseteq \{\calP, \calS, \overline\calS\}$, which characterize when those conditions are also sufficient. Our preliminary characterization of these classes in Theorem \ref{thm:treilli} reveals strict inclusions and highlights the special role of Eulerian graphs (see Figure \ref{fig:class_intersections}).

\noindent In fact, we could take apart a condition $\calS\cap\overline\calS$ that is currently included in $\calS$ and $\overline\calS$. That is, if $Source(T) = Sink(T) = \{v\}$ for some vertex $v$, then no acyclic T-odd orientation exists unless the graph is a singleton. We would like to discuss some properties of this condition.
Remind that $Source(T) = Sink(T)$ if and only if the graph has no odd degree vertices. Now, in graphs satisfying $|E|+|V|\equiv_2 0$ the condition $\calP$ already imposes that $|V\setminus T|\equiv_2 0$ so in particular $\calS\cap\overline\calS$ is always satisfied. Hence, the condition $\calS\cap\overline\calS$, for graphs with at least $4$ vertices, finds its real meaning for Eulerian graphs verifying $|E|+|V|\equiv_2 1$ and taking it apart would allow only to refine the class $\cC_{\calP \calS}\setminus \cC_{\calP}$ into two subclasses depending of whether the graphs verify $|E|+|V|\equiv_2 1$ or $|E|+|V|\equiv_2 0$.\newline

\noindent A central contribution of this work is the complete characterization of solvable instances for grids, cylinders, and large tori. These results not only solve the problem for significant graph families but also serve to demonstrate the strictness of the inclusions within our proposed hierarchy. Methodologically, we introduced $T$-decompositions, a powerful constructive tool that provides a pathway to polynomial-time algorithms for these specific graph classes.\newline

\noindent Our findings open several avenues for future research. A primary open question is the computational complexity of recognizing graphs in the class $\cC_{\calP \calS \overline\calS}$. A polynomial-time algorithm for this problem would significantly clarify the landscape of tractable instances especially since it would directly provide from our results a polynomial-time algorithm for recognizing instances in $\cC_{\calP \calS}$ and $\cC_{\calP}$. Further investigation is also needed to fully characterize the graphs within this class, for which we have provided initial insights and counterexamples using cliques. We thus pose the following open problem:
\begin{problem}
	Given a graph $G$, does there exists a polynomial time algorithm to decide whether $G \in \cC_{\calP \calS \overline\calS}$?
\end{problem}

\noindent Finally, extending our results on Cartesian products presents a natural next step. Indeed, a consequence of Theorem $1.8$ in \cite{kiralyDualCriticalGraphsNotes} gives that, given two graphs $G,H \in \cC_{\calP}$ such that $|V(G)|\equiv_2 1$, we have $G \square H \in \cC_{\calP}$. More generally, determining the conditions on factors $G$ and $H$ such that $G \square H$ belongs to one of our classes remains a challenging problem. The case of small tori, such as $C_3 \square C_q$, which our current methods do not cover, stands as an intriguing open problem as well.\newline

\begin{figure}
	\centering
	\includegraphics[width=0.4\textwidth]{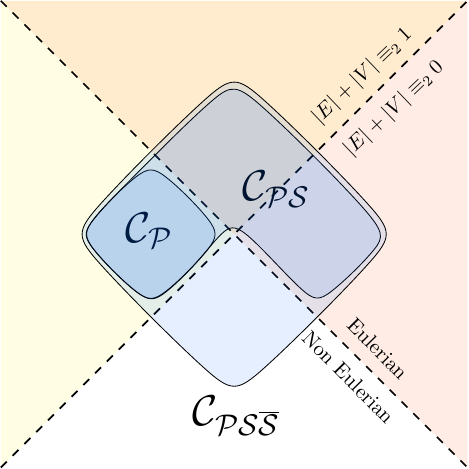}
	\caption{Hierarchy of classes $\cC_{\cN}$ for $\cN \subseteq \{\calP, \calS, \overline\calS\}$.}
	\label{fig:class_intersections}
\end{figure}

 \FloatBarrier
\printbibliography

\end{document}